%% file: main.tex
\documentclass{article} 

\usepackage{natbib}
\usepackage{authblk}
\usepackage[ruled]{algorithm2e}
\usepackage[margin=1in]{geometry}
\usepackage{palatino}
\usepackage{amsthm,amsmath,amssymb,amsfonts}
\usepackage{subcaption}
\usepackage{hyperref}
\usepackage{booktabs}
\usepackage{algpseudocode}
\usepackage[utf8]{inputenc}
\usepackage{xspace}
\usepackage[final]{pdfpages}
\usepackage{graphicx}
\usepackage{verbatim}
\usepackage{xfrac}
\usepackage{array}
\usepackage[normalem]{ulem}
\usepackage[labelfont=bf]{caption}
\usepackage{float}
\usepackage{todonotes}
\usepackage{outlines}
\usepackage[section]{placeins}
\usepackage{natbib}
\usepackage{mathtools}
\usepackage{comment}
\usepackage{makecell}

\usepackage{titlesec}
\titleformat{\section}
{\normalfont\normalsize\centering\MakeUppercase}{\thesection.}{1em}{}
\titleformat{\subsection}[runin]
  {\bfseries}{\thesubsection.}{1em}{\addperiod}
\titleformat{\subsubsection}[runin]
  {\bfseries}{\thesubsubsection.}{1em}{\addperiod}
\newcommand{\addperiod}[1]{#1.}
\usepackage{fancyhdr}

\input{myMacros.tex}

\begin{document}

\title{Slowly Scaling Per-Record Differential Privacy}

\author[1]{Brian Finley}
\author[1]{Anthony M Caruso}
\author[1]{Justin C Doty}
\author[2]{Ashwin Machanavajjhala}
\author[3]{Mikaela R Meyer}
\author[2]{David Pujol}
\author[2]{William Sexton}
\author[3]{Zachary Terner}
\affil[1]{U.S. Census Bureau\thanks{Works (articles, reports, speeches, software, etc.) created by U.S. Government employees are not subject to copyright in the United States, pursuant to 17 U.S.C. \S 105.  International copyright, 2024, U.S. Department of Commerce, U.S. Government. Any opinions and conclusions expressed herein are those of the authors and do not reflect the views of the U.S. Census Bureau.}\thanks{Contact: brian.finley@census.gov, anthony.m.caruso@census.gov, justin.c.doty@census.gov}}
\affil[2]{Tumult Labs\thanks{Contact: ashwin@tmlt.io, david.pujol@tmlt.io, william.sexton@tmlt.io}}
\affil[3]{The MITRE Corporation\thanks{Contact: mrmeyer@mitre.org, zterner@mitre.org}\thanks{Approved for Public Release by The MITRE Corporation; Distribution Unlimited. Public Release Case Number 24-2528.}}

\maketitle

\begin{abstract}
We develop formal privacy mechanisms for releasing statistics from data with many outlying values, such as income data. These mechanisms ensure that a per-record differential privacy guarantee degrades slowly in the protected records’ influence on the statistics being released. \\
\medskip
\\
Formal privacy mechanisms generally add randomness, or "noise," to published statistics. If a noisy statistic's distribution changes little with the addition or deletion of a single record in the underlying dataset, an attacker looking at this statistic will find it plausible that any particular record was present or absent, preserving the records’ privacy. More influential records -- those whose addition or deletion would change the statistics’ distribution more -- typically suffer greater privacy loss. The per-record differential privacy framework quantifies these record-specific privacy guarantees, but existing mechanisms let these guarantees degrade rapidly (linearly or quadratically) with influence. While this may be acceptable in cases with some moderately influential records, it results in unacceptably high privacy losses when records’ influence varies widely, as is common in economic data. \\
\medskip
\\
We develop mechanisms with privacy guarantees that instead degrade as slowly as logarithmically with influence. These mechanisms allow for the accurate, unbiased release of statistics, while providing meaningful protection for highly influential records. As an example, we consider the private release of sums of unbounded establishment data such as payroll, where our mechanisms extend meaningful privacy protection even to very large establishments. We evaluate these mechanisms empirically and demonstrate their utility.
\end{abstract}
\newpage

\section{Introduction}
\label{sec:introduction}

Many government organizations have an obligation to release statistics about the public without violating the public's confidentiality. For example, although the U.S. Census Bureau releases data products for public use, Title 13 of the U.S. Code prohibits the Census Bureau from publishing information that could disclose the data provided by any specific individuals or businesses~\citep{title13}.
Title 26 similarly prohibits the IRS from disclosing tax data~\citep{title26},
and Title 7 protects the data and identities of data providers to the U.S. Department of Agriculture~\citep{agcensus}. Ensuring that data releases from the government preserve the confidentiality of its respondents is critical to the overall mission of serving the public.

To achieve this mission, some statistical agencies have begun to use differential privacy~\citep{dwork_calibrating_2006} to provide mathematically provable guarantees on the privacy of respondents. 
For example, the U.S. Census Bureau has adopted differential privacy for the release of some of its data products, including those from the 2020 Census~\citep{haney_differentially_2021, census_dp, abowd_2020_2022}. The IRS has also adopted differential privacy to facilitate the release of the Department of Education's College Scorecard website~\citep{berghel2022tumult}.

The goal of differential privacy is to ensure that the presence or absence of any single record cannot be detected with high certainty from published statistics or queries. To accomplish this, differentially private systems add noise to queries to obfuscate whether any particular record is present. When a query's value can be dramatically altered by the presence or absence of a single record, a large amount of noise is required to hide such a record. Adding enough noise to hide very influential records can ruin the utility of the release.

In these situations, one may use \textit{per-record zero-concentrated differential privacy}, or PRzCDP~\citep{VLDBprivately}. Instead of enforcing a uniform privacy guarantee for all possible records, PRzCDP lets less influential records receive stronger privacy guarantees. If only a small number of very influential records are expected to be present, then PRzCDP allows for an improvement in utility by adding noise at a scale that provides a desired level of protection for a vast majority of records, while maintaining quantifiable, albeit lower, privacy guarantees for influential outliers. This gives data holders a finer ability to control the tradeoff between privacy and utility.

This is particularly useful when large, influential records have a lower expectation of confidentiality or when the increased utility outweighs the increased privacy loss. One such example comes in the form of establishment data, where large firms often must produce public statements about confidential values, and as such, are less susceptible to the risks of additional disclosure. In these cases, utility can be greatly increased while reducing the overall privacy guarantee only for the few records which have reduced privacy risk.

Existing mechanisms for PRzCDP use a method known as unit splitting, which incurs privacy losses that scale with the square of a record's magnitude. 
This can be helpful in cases where the dataset is moderately skewed. 
However, in cases with very large outliers (e.g., the largest record is several orders of magnitude larger than the other observations), these methods can easily result in unacceptable privacy losses.

To address this, we introduce a class of mechanisms we call \textit{slowly scaling mechanisms}, in which the privacy guarantee degrades slowly in the records' influence on the query.
Rather than scaling with the square of a record's magnitude, these privacy losses scale sublinearly in a record's influence, ensuring that even very large records incur a reasonable privacy loss.

The paper continues as follows. In Section~\ref{sec:prelim}, we provide some foundational definitions and theorems of differential privacy, discuss per-record differential privacy, and describe the unit splitting mechanism for per-record differential privacy~\citep{VLDBprivately}. Section~\ref{sec:motivation} introduces slowly scaling mechanisms and Sections~\ref{sec:transformation} and \ref{sec:additive} respectively detail our two classes of slowly scaling mechanisms: the transformation mechanisms and the additive mechanisms. Section~\ref{sec:empirical} empirically studies our mechanisms with applications to two datasets. The paper concludes with Section~\ref{sec:conclusion}. 

We make several contributions in this paper:

\begin{itemize}
    \item We demonstrate the inadequacy of unit splitting in heavily skewed data, where some records may be orders of magnitude larger than others.
    \item  We introduce the idea of \textit{slowly scaling mechanisms}, which incur privacy losses that scale sublinearly in the influence of a record.
    \item We develop two classes of slowly scaling mechanisms: the transformation and additive mechanisms.
    \item We empirically demonstrate the utility of these mechanisms when applied to real data. 
\end{itemize}

\section{Preliminaries}
\label{sec:prelim}

\subsection{Data Model}

We consider the table schema $\mathcal{D}(a_1, a_2, \dots, a_d)$, where $\mathcal{A} = \{a_1, a_2, \dots, a_d\}$ denotes the set of attributes of $\mathcal{D}$. The domain of each attribute $a_i$, denoted by $\Dom(a_i)$, need not be finite or bounded.
The full domain of $\mathcal{D}$ is $\Dom(\mathcal{D}) = \Dom(a_1) \times \Dom(a_2) \times \dots  \Dom(a_d)$. A database $D$ is an instance of relation $\mathcal{D}$. That is, $D$ is a multi-set of tuples $r = (x_1, \ldots, x_d)$ where $x_i \in \Dom(a_i)$. The number of tuples in $D$ is denoted as $|D| = n$.  We sometimes refer to a single record in a database as a ``unit.''

\subsection{Zero-Concentrated Differential Privacy}
Here, we discuss the properties of formal privacy protection. We turn to a framework of privacy standards called differential privacy (DP) \citep{dwork_calibrating_2006}. Informally, a randomized mechanism $M$ satisfies differential privacy if the output distribution of the mechanism does not change too much with the addition or removal of a single individual's record in the input dataset. This ensures that an adversary observing the output of $M$ cannot determine whether \textit{any} particular individual's data was part of the input of $M$.

In this paper, the per-record differential privacy notion that we primarily work with is a generalization of a widely used DP variant called \textit{zero-concentrated differential privacy} (zCDP) \citep{bun_concentrated_2016}.
zCDP quantifies the difference between output distributions induced by adding or removing a single individual's record in terms of the R\'enyi divergence, defined below. 

\begin{defi}[R\'enyi Divergence 
\label{def:Renyi} \citep{renyi1961measures}]
    Let $P$ and $Q$ be probability distributions on some probability space $\Omega$ and let them respectively have the density functions $p$ and $q$. Let the random variables $X_P \sim P$ and $X_Q \sim Q$. For $\alpha\in(1,\infty)$, we define the R\'enyi divergence of order $\alpha$ between $P$ and $Q$ as
    \begin{align*}
        d_{\alpha} \infdivx*{X_P}{X_Q} \equiv d_{\alpha} \infdivx*{P}{Q} \equiv \frac{1}{\alpha-1} \ln \int_{\Omega} p(x)^{\alpha} q(x)^{1-\alpha} \, dx.
    \end{align*}

    We also define the R\'enyi divergence of order $\infty$ as
    \begin{align*}
        d_{\infty} \infdivx*{X_P}{X_Q} \equiv d_{\infty} \infdivx*{P}{Q} \equiv \ln(\text{ess sup}_P \frac{p}{q}).
    \end{align*}
    
\end{defi}

 We begin by formalizing the notion of neighboring databases, which are databases that differ only in the presence of a single record.

\begin{defi}[Neighboring Databases] \label{def:UnbndNeighbors}
Where $\ominus$ denotes the symmetric difference, two databases $D, D' \in \mathcal{D}$ are considered \emph{neighboring databases} if $D \ominus D' = \{r\}$ for some record $r$.
\end{defi}

From here, we can define the formal notion of privacy under zCDP.

\begin{defi}[Zero-Concentrated Differential Privacy]\label{def:zCDP}
The randomized mechanism $M$ satisfies $\rho$-zCDP for $ \rho \geq 0$ iff, for any two neighboring databases $D,D' \in \mathcal{D}$,
\begin{equation*}
d_\alpha\left(M(D) \| M(D')\right) \leq \rho \alpha \text{ for all } \alpha \in (1, \infty).
\end{equation*}
\end{defi}
A mechanism which satisfies zCDP ensures that no record contributes too much to the output of a mechanism by ensuring that the R\'enyi divergence between the output distributions of the mechanism with and without any arbitrary record scales linearly in the order term $\alpha$. The parameter $\rho$ is often called the privacy loss and governs how much the two distributions are allowed to differ from one another. Large values of $\rho$ allow the distributions to vary more and correspond to weaker privacy protections. Smaller values of $\rho$ restrict how much the distributions can vary and result in stronger privacy protections. 
\par

Complex zCDP mechanisms are usually constructed by composing smaller primitive mechanisms together. To do this, we employ the following properties of zCDP: sequential composition, parallel composition, and post-processing invariance. 
Sequential composition allows multiple mechanisms which satisfy zCDP to compose together and jointly satisfy zCDP. In these cases, the final privacy loss is the sum of the individual privacy losses of each mechanism.
\begin{thm}[zCDP Sequential Composition \citep{bun_concentrated_2016}] \label{thrm:zCDP-sequential}
Let $M_1, M_2$ be randomized mechanisms which satisfy $\rho_1$-zCDP and $\rho_2$-zCDP respectively. Then the mechanism $M'(D) = \left(M_1(D), M_2(D)\right)$ satisfies $(\rho_1 + \rho_2)$-zCDP.
\end{thm}
Additionally, under parallel composition, if two zCDP mechanisms are run on disjoint subsets of the database, their combination satisfies zCDP. In these cases, the final privacy loss is the maximum privacy loss of the component mechanisms.
\begin{thm}[zCDP Parallel Composition \citep{bun_concentrated_2016}] \label{thrm:zCDP-parallel}
Let $M_1, M_2$ be randomized mechanisms which respectively satisfy $\rho_1$-zCDP and $\rho_2$-zCDP. Let $D_1, D_2$ be two disjoint subsets of a database, $D$. The mechanism $M'(D) = \left(M_1(D_1), M_2(D_2)\right)$ satisfies  $\max(\rho_1 , \rho_2)$-zCDP.
\end{thm}

As with all formal privacy definitions, zCDP is invariant to post-processing. That is, if the output of a mechanism satisfies zCDP, any post-processing on that output (without accessing the input data) preserves the zCDP guarantee. 

\begin{thm}[zCDP Post-processing \citep{bun_concentrated_2016}] \label{thrm:zCDP-postprocessing}
Let $M$ be a randomized mechanism which satisfies $\rho$-zCDP. Let $M'(D) = f(M(D))$ for some arbitrary function $f$. Then $M'$ satisfies $\rho$-zCDP.
\end{thm}

Using the composition and post-processing theorems, complex zCDP mechanisms are often built from simple primitive mechanisms, such as the Gaussian mechanism \citep{bun_concentrated_2016}. To define this mechanism, we must first define the sensitivity of a query. 

We call the maximum change to a function due to the removal or addition of a single record the \textit{sensitivity} of the function. This can be thought of as the maximal influence that a single record can have on the function.

\begin{defi}[$\ell_2$-Sensitivity]\label{def:sensitivity}
Denote the $\ell_2$-norm by $|\cdot|_2$. The sensitivity of the vector function $f$ is 
$$\Delta \equiv \sup_{D',D \in \mathcal{D} \text{ such that } |D' \ominus D| = 1}{|f(D)-f(D')|_2}.$$
\end{defi}

The Gaussian mechanism simply adds Gaussian noise to a query, achieving a zCDP guarantee proportional to the square of the query's sensitivity.

\begin{defi}[Gaussian Mechanism \citep{bun_concentrated_2016}]
\label{def:gaussian_mech}
Let $q$ be a query with sensitivity $\Delta$. Consider the mechanism $M$ that, on input $D$, releases a sample from $\mathcal{N}(q(D),\sigma^2)$. Then $M$ satisfies$ \frac{\Delta ^2}{2\sigma^2}$-zCDP.
\end{defi}
While the definition of zCDP protects the privacy of an individual, these protections extend to groups of arbitrary size. The group privacy property says that, like an individual record, groups of records receive similar privacy protection, though the strength of that protection decays quadratically in the size of the group.

\begin{thm}[zCDP Group Privacy \citep{bun_concentrated_2016}]\label{thrm:zCDP-group}
Let $M$ be a randomized mechanism which satisfies $\rho$-zCDP. Then $M$ guarantees $(k^2\rho)$-zCDP for groups of size $k$. That is, for every pair of databases $D,D' \in \mathcal{D}$ satisfying $|D \ominus D'| = k$, we have the following.
\begin{equation*}
d_\alpha \left(M(D) \| M(D') \right) \leq (k^2\rho)\cdot \alpha \text{ for all }\alpha \in (1,\infty).
\end{equation*}
\end{thm}

\subsection{Per-Record Zero-Concentrated Differential Privacy}
\label{sec:PRzCDP}
Here, we describe \textit{per-record zero-concentrated differential privacy} (PRzCDP) \citep{VLDBprivately}, a relaxation of zCDP designed for skewed data tasks. 

Under PRzCDP, the privacy loss is parameterized in the form of a function of a record's confidential value.
For example, if a record contains strictly positive numerical values, we can consider a privacy loss that grows proportionally to that record value. To capture this relationship, PRzCDP employs a \textit{record-dependent policy function}, which returns the privacy loss of any arbitrary record.

\begin{defi}[Record-dependent policy function] \label{def:policy}
A record-dependent policy function $P: \mathcal{T} \rightarrow \mathbb{R}_{\geq 0}$ denotes a maximum allowable privacy loss associated with a particular record value $r \in \mathcal{T}$, where $\mathcal{T}$ is the universe of possible records. 
\end{defi}

Since the policy function's value $P(r)$ is dependent on the private record $r$, the value $P(r)$ cannot be released for any particular record without violating that record's privacy. We can, however, release the policy function $P$ itself, which can be used to compute the privacy loss of any arbitrary and well-formed record. Given a policy function,  per-record zero-concentrated differential privacy is defined as follows.

\begin{defi}[$P$-Per-Record Zero-Concentrated DP ($P$-PRzCDP)] \label{def:PRzCDP}
The randomized mechanism $M$ satisfies $P$-per-record zero-concentrated differential privacy ($P$-PRzCDP) iff, for any two neighboring databases $D,D' \in \mathcal{D}$ satisfying $D \ominus D' = \{r\}$,
$$
d_\alpha \left(M(D) \| M(D') \right) \leq \alpha P(r) \text{ for all } \alpha \in (1, \infty).
$$
\end{defi}

In this definition, the privacy loss associated with each record scales according to the policy function, as opposed to having equal privacy loss for all records. Traditional zCDP can be stated as a special case of $P$-PRzCDP, where the policy function is a constant. In the remainder of the paper, we will refer to the value $P(r)$ for a particular record $r$ as that record's privacy loss.

\subsection{Properties of Per-Record Differential Privacy}
\label{sec:properties}
PRzCDP satisfies many of the same properties as zCDP, which allows complex mechanisms to be built from smaller primitive mechanisms. Of these properties, PRzCDP satisfies post-processing invariance, basic adaptive sequential composition, parallel composition, and group privacy.

\begin{lem}[Closure under post-processing]
\label{lem:PRzCDP-post-processing}
    Let $\mathcal{T}^*$ be a multiset of records in $\mathcal{T}$. Given $M:\mathcal{T}^* \to \mathcal{Y}$, a $P$-PRzCDP  mechanism $M$, and any function $f$, it is the case that $f\circ M$ is also $P$-PRzCDP. 
\end{lem}

PRzCDP satisfies \textit{basic adaptive} sequential composition. That is, as long as the policy functions are chosen prior to running any of the mechanisms, any two adaptively chosen mechanisms in combination satisfy PRzCDP as follows.

\begin{lem}[Basic adaptive sequential composition for $P$-PRzCDP]
\label{lem:PRzCDP-seq}
Let $M_1$ satisfy $P_1$-PRzCDP, and let $M_2$ satisfy $P_2$-PRzCDP. Then $M_3(D) = M_2 \left( M_1(D), D \right)$ satisfies $\left(P_1(r) + P_2(r) \right)$-PRzCDP.
\end{lem}

Just as with zCDP, the privacy losses sum when the mechanisms are composed together in this way. In PRzCDP, since the privacy loss is in the form of the policy function, the equivalent is the sum of policy functions. PRzCDP also satisfies a form of parallel composition. When multiple mechanisms which satisfy $P$-PRzCDP are run on disjoint subsets of the database, the joint result also satisfies $P$-PRzCDP:

\begin{lem}[Parallel composition for $P$-PRzCDP]
\label{lem:PRzCDP-par}
Let $T$ be a partition of size $J \in \mathbb{N} \cup \{ \infty \}$ over the universe of possible records. That is,
$$
T \equiv \bigcup_{j=1}^J C_j, \qquad C_i \cap C_j = \emptyset, i \neq j.
$$
Let $\mathcal{D}_j$ be the space of all databases containing only records in $C_j$ for $j \in [J]$. Let $\{ M_j \}_{j=1}^J$ be mechanisms satisfying $P$-PRzCDP for databases $D_j \in \mathcal{D}_j$ for $j \in [J]$, respectively. Then:
$$
M(D) \equiv \left\{ M_j(D \cap C_j) \mid j \in [J] \right\}
$$
satisfies $P$-PRzCDP. Note that the $M_j$s can depend on their respective $C_j$s, allowing for adaptivity.
\end{lem}

Note that Lemma~$\ref{lem:PRzCDP-par}$ does not require the individual mechanisms $\{ M_j \}^J_{j=1}$ to be the same for all $j \in [J]$, allowing for different analysis for each partition. PRzCDP satisfies the group privacy notion as well. 

\begin{lem}[Simple group privacy for $P$-PRzCDP]
\label{lem:PRzCDP-group}Consider a sequence of databases $D_0, \dots, D_{J}$ where $D_0 = D$ and $D_{j} \ominus D_{j-1} = \{ r_j \}$ for $j \in [J]$. Let $M$ be a randomized mechanism satisfying $P$-PRzCDP. Then we have:
\begin{equation}
d_\alpha \left(M(D_0)\|M(D_J) \right) \leq \alpha J \sum_{j=1}^J P(r_j).
\end{equation}
\end{lem}

While most of the results presented here will pertain to single queries, they can all be expanded to workloads of queries by using the composition theorems. For example, if a workload contains two queries, each query can be answered independently using a $P_1(r)$ and $P_2(r)$ PRzCDP mechanism, respectively. By Lemma~\ref{lem:PRzCDP-seq} the combination of both mechanisms satisfies PRzCDP with a policy function $P(r) = P_1(r) + P_2(r)$. These lemmas apply to zCDP mechanisms as well, since they can be considered a special case of PRzCDP mechanisms with a constant policy function. For example, a low sensitivity query such as a count can be answered with a $\rho$-zCDP mechanism followed by a high sensitivity query such as a sum which is answered with a $P(r)$-PRzCDP mechanism. The combination of the two satisfies $(\rho + P(r))$-PRzCDP.

\subsection{Unit Splitting}

Unit splitting \citep{VLDBprivately} is a preprocessing step which can be used to ensure that traditional zCDP mechanisms satisfy PRzCDP. Informally, each record in the dataset is split into multiple records according to a splitting rule. Then, a traditional zCDP mechanism, such as the Gaussian mechanism, is used to compute the query on the split dataset. This results in a PRzCDP guarantee with a policy function of, $P(r) = \rho |A(r)|^2$ where $A(r)$ is the number of times record $r$ was split and $\rho$ is the privacy loss budget used for the Gaussian mechanism. 

\begin{table}[t]
    \centering
    \begin{tabular}{|c|c|c|}
        \hline
         \textbf{ID}& \textbf{Industry} & \textbf{Employees}  \\
         \hline
         1 & Retail & 5 \\
         \hline
         2 & Retail & 5 \\
         \hline
         3 & Retail & 10 \\
         \hline
         4 & Services & 20\\
         \hline
         5 & Hospitality & 30 \\
         \hline
         6 & Technology & 10,000 \\
         \hline

    \end{tabular}
    \caption{An example database where each record represents a single establishment with its associated industry and number of employees. Applying unit splitting with a splitting threshold of 10 employees on Establishment 6 results in a privacy loss that is 1,000,000 times that of Establishment 1.}
    \label{tab:sample_table}
\end{table}
\begin{exa} \label{ex:unit_splitting}
    Consider answering a sum over the Employees column on the sample data in Table~\ref{tab:sample_table} by unit splitting. Consider the splitting rule that splits each establishment into multiple rows containing at most $10$ employees. In this case, the first three establishments are not split at all. Establishment 4 is split into two rows, Establishment 5 is split into three rows and Establishment 6 is split into 1,000 rows. Since the rows now have at maximum 10 employees, the sensitivity of the sum query is 10. If we answer the sum query using the Gaussian mechanism with $\sigma^2 = 50$, then the establishments that are not split incur a privacy loss of $P(r) = \rho = 1$. Establishment 4, which is split into 2 rows, incurs a privacy loss of $\rho |A(r)|^2 = 4$. Likewise, Establishment 5 incurs a privacy loss of $9$, and Establishment 6 incurs a privacy loss of 1,000,000.
\end{exa}
This method has several benefits. First, it allows a practitioner to introduce an upper bound to the contribution of any record where there may not have been one previously. Additionally, it allows a practitioner to split larger records with large contributions into multiple smaller ones with smaller contributions. This results in less noise overall for the query, with the consequence of larger privacy losses for larger records.
However, in cases with heavily skewed data like Example~\ref{ex:unit_splitting}, some records can have extremely large privacy losses. To mitigate this effect, we introduce novel slowly scaling mechanisms which ensure that even large records incur reasonable privacy losses.

\section{Slowly Scaling Mechanisms}
\label{sec:motivation}

Unit splitting has been shown to provide high utility in cases where the user contribution may be unbounded but on average is small \citep{adeleye2023publishing}. As shown in Example~\ref{ex:unit_splitting}, records that are split into a small number of additional rows still incur reasonable privacy losses, but records that are very large incur unacceptable privacy losses. This is a result of the policy function which scales in the square of the record size. Even a record which is split into only 10 rows incurs 100 times the privacy loss of a record which is not split. This is particularly problematic in cases where there are highly skewed records which may be several orders of magnitude larger than the median record. \par 

To address the problem of rapidly scaling privacy loss, we introduce a class of PRzCDP mechanisms called \textit{slowly scaling mechanisms}. These mechanisms satisfy PRzCDP with a policy function that is \textit{sublinear}. In the following section, we will introduce two classes of slowly scaling mechanisms: transformation mechanisms and additive mechanisms. \par 

The transformation mechanisms first pass the query result through a transformation function. They then add Gaussian noise to the result and invert the transformation function. These mechanisms result in policy functions which scale in the square of the transformation function applied to a quantification of a record's influence. We demonstrate below that this can result in policy functions which scale as slowly as the log-squared rate.

The additive mechanisms act similarly to the Gaussian mechanism in that they simply draw a sample from a noise distribution and add it to the query. These mechanisms, however, use distributions which have thicker tails than the Gaussian distribution, resulting in policy functions that can scale in the records' influence as slowly as the log rate.

\section{Transformation Mechanisms }
\label{sec:transformation}

In this section, we describe the transformation mechanism, a mechanism which privatizes a univariate, nonnegative, real-valued query under the constraints of PRzCDP. As input, the mechanism takes a query value $q(D) \in [0,\infty)$; an offset parameter $a \in \R$; a noise scale parameter $\sigma \in (0,\infty)$; a concave, strictly increasing  transformation function $f:[a,\infty) \to \mathcal{F} \subseteq \R$; and an estimator $g:\mathcal{F} \to \mathcal{G} \subseteq \R$ that ``undoes" the transformation once noise has been added.

\cite{haney_utility_2017} uses a similar mechanism that adds Laplace noise to transformed queries to provide protections under a different privacy concept which they develop specially for use with linked employer-employee data. Concurrent, independent work by \cite{WebbEtAl2023} uses special cases of the transformation mechanism to provide protections under yet another privacy concept. In this paper, we contribute novel estimators $g$ and an analysis of the algorithm's PRzCDP guarantees. See Appendix~\ref{sec:metric-DPNew}
for an analysis of our slowly scaling mechanisms using a privacy concept closely related to the one used by \cite{WebbEtAl2023}. This privacy concept complements PRzCDP nicely, by protecting different features of the data.

The transformation mechanism operates by mapping the private query value $q(D)$ to $f(q(D)+a)$ and adds noise sampled from a mean-zero Gaussian distribution with standard deviation $\sigma$. This quantity, called $\tilde{v}$ in Algorithm~\ref{alg:transform}, is essentially a privatized version of $f(q(D)+a)$, rather than of $q(D)$. The algorithm finishes, then, by returning $g(\tilde{v})$, an estimator of $q(D)$. Different estimators can be used, depending on what properties the user desires and which transformation function is used. We provide mean-unbiased estimators in Section~\ref{subsubsec:Estimators} and median-unbiased estimators in Appendix~\ref{sec:transMechMedianUnbiased}.

\begin{algorithm}
\caption{\label{alg:transform} Transformation-based mechanism }
\begin{algorithmic}[1]
    \Procedure{TransformationPrivatize}{Private query answer $q(D)$,  offset parameter $a$, \newline scale parameter $\sigma$, transformation function $f:[a,\infty)\to \mathcal{F} \subseteq \R$, estimator $g:\mathcal{F} \to \mathcal{G} \subseteq \R$}
    \State $v \leftarrow f(q(D) + a)$
    \State $\tilde{v} \leftarrow N(v,\sigma^2)$ \label{line:transform_noise}
    \State $\tilde{S}\leftarrow g(\tilde{v})$
    \State \textbf{output} $\tilde{S}$
    \EndProcedure
\end{algorithmic}    
\end{algorithm}

To derive this algorithm's PRzCDP guarantee, we first recall a bound on the R\'enyi divergence between two homoscedastic Gaussian distributions.

\begin{lem}[Lemma 2.4 of \cite{bun_concentrated_2016}]
    \label{lem:BS-Gauss}
    Let $\mu,\nu\in \mathbb{R}$, $\sigma\in\mathbb{R}_{+}$, and $\alpha \in [1,\infty)$. Then
    \begin{align*}
        d_{\alpha} \left( N(\mu,\sigma^2) \| N(\nu,\sigma^2) \right) = \frac{\alpha (\mu-\nu)^2}{2\sigma^2}.
    \end{align*}
\end{lem}

We also introduce the per-record sensitivity. In the same way that the sensitivity in Definition~\ref{def:sensitivity} captures the maximal influence that \emph{any} record can have on a query, this captures the maximal influence that a \emph{particular} record can have on a query.

\begin{defi}[Per-Record Sensitivity] \label{def:PR-sensitivity}
The per-record sensitivity of the query $q:\mathcal{D} \to \R$ for record $r$ is $$\Delta(r) \equiv \sup_{D,D^\prime \text{ such that } D \ominus D^\prime = {r}} |q(D) - q(D^\prime)|.$$
\end{defi}

We now codify the basic assumptions on the inputs to the algorithm then state and prove the algorithm's privacy guarantee.

\begin{asm}[Assumptions on Transformation Mechanisms] \label{as:transformInput}
   Assume the query value $q(D) \in [0,\infty)$; the offset parameter $a \in \R$; the noise scale parameter $\sigma \in (0,\infty)$; the transformation function $f:[a,\infty) \to \mathcal{F} \subseteq \R$ is concave and strictly increasing; and the estimator $g:\mathcal{F} \to \mathcal{G} \subseteq \R$.
\end{asm}

\begin{thm}[PRzCDP Guarantees for Transformation Mechanisms]
    \label{thm:transform}
    Under Assumption~\ref{as:transformInput}, Algorithm~\ref{alg:transform}$\allowbreak(q(D), a, \sigma, f, g)$ satisfies $P$-PRzCDP for $P(r) = \frac{|f(\Delta(r)+a) - f(a)|^2}{2\sigma^2}$. 
\end{thm}

\begin{proof}
        By the definition of PRzCDP, for neighboring data sets $D$ and $D'$ where $D \ominus D' = \{r\}$, we would like to bound from above $d_{\alpha} \infdivx*{M(D)}{M(D')}$ and $d_{\alpha} \infdivx*{M(D')}{M(D)}$ by $\alpha P(r)$, for all $\alpha$, where $M(\cdot)$ is the notation for Algorithm~\ref{alg:transform} with transformation function $f$, estimator $g$, and privacy policy parameters $\sigma$ and $a$. By the symmetry of the neighbor relationship, we can proceed WLOG by establishing a bound only on $d_{\alpha} \infdivx*{M(D)}{M(D')}$.
        
        To start, we note that $M(D) = g(f(q(D) + a) + N(0,\sigma^2))$. By the data processing inequality \citep{renyi1961measures}, we have         
\begin{align*}
            d_{\alpha} \infdivx*{M(D)}{M(D')} \leq d_{\alpha} \infdivx*{f(q(D) + a) + N(0,\sigma^2)}{f(q(D') + a) + N(0,\sigma^2)}.
\end{align*}
From Lemma~\ref{lem:BS-Gauss}, we get 
\begin{align*}
   d_{\alpha} \infdivx*{f(q(D)+a) + N(0,\sigma^2)}{f(q(D')+a) + N(0,\sigma^2)} = \frac{\alpha|f(q(D)+a) - f(q(D')+a)|^2}{2\sigma^2}.
\end{align*}

We rewrite the right-hand side of the above equation as
\begin{align*}
  \frac{\alpha|f(q(D)+a) - f(q(D')+a)|^2}{2\sigma^2} = 
  \frac{\alpha|f(q(D)+a) - f(q(D)+a + [q(D')-q(D)])|^2}{2\sigma^2}.
\end{align*}

Noting that $f$ is strictly increasing, we can bound the right-hand side above using the per-record sensitivity, $\Delta(r)$ (see Definition~\ref{def:PR-sensitivity}): 
\begin{align*}
  \frac{\alpha|f(q(D)+a) - f(q(D)+a + [q(D')-q(D)])|^2}{2\sigma^2} \leq
  \frac{\alpha|f(q(D)+a) - f(q(D)+a + \Delta(r))|^2}{2\sigma^2} .
\end{align*}

Since $f$ is concave and increasing, the right-hand side above is bounded by $\alpha P(r)$, our policy function in the theorem statement:
\begin{align*}
  \frac{\alpha|f(q(D)+a) - f(q(D)+a + \Delta(r))|^2}{2\sigma^2} \leq
  \frac{\alpha|f(a) - f(a + \Delta(r))|^2}{2\sigma^2}.
\end{align*}

Therefore, Algorithm 1 satisfies $P$-PRzCDP.
\end{proof}

We can see that the policy function of Theorem~\ref{thm:transform} grows in the per-record sensitivity at a rate determined by the transformation function, $f$. With the variety of transformation functions available, this can capture many possible policy functions. 
For example, we can achieve a quadratically scaling policy function, like that of the unit splitting mechanism, by choosing the identity function $f(x) = x$ as the transformation. This results in the policy function $P(r)= \frac{\Delta(r)^2}{2 \sigma^2}$. To get a more slowly scaling policy function, the transformation itself must scale more slowly as well. In the following sections, we give two examples of such mechanisms, the kth root and log transformation mechanisms.

\subsection{Kth Root Transformation Mechanism $f(x) = \sqrt[k]{x}$}
\begin{cor}[]
\label{cor:kth_root_transform}
Under the conditions of Assumption~\ref{as:transformInput}, with transformation function $f(x) = \sqrt[k]{x}$, Algorithm~\ref{alg:transform}$(q(D), a, \sigma, f, g)$  satisfies $P$-PRzCDP with 
$$P(r) = \frac{(\sqrt[k]{\Delta(r)+a} -\sqrt[k]{a})^2}{2\sigma^2}$$
for any estimator $g$ and for all $a \geq 0$.
\end{cor}
\begin{proof}
This follows from a straightforward application of Theorem \ref{thm:transform}, with $f(x) = \sqrt[k]{x}$. 
In this case, $f$ is clearly concave and strictly increasing on the range $[a, \infty)$ for $a \geq 0$ and $k \geq 1$. Therefore, this mechanism satisfies $P$-PRzCDP with 
$$P(r) = \frac{|f(\Delta(r)+a) - f(a)|^2}{2\sigma^2} = \frac{(\sqrt[k]{\Delta(r)+a} -\sqrt[k]{a})^2}{2\sigma^2}$$
\end{proof}

The kth root transformation mechanism is the result of using the transformation function $f(x) = \sqrt[k]{x}$ for any $k \geq 1$. This results in a policy function that scales in the square of the kth root. Notable special cases include the square root transformation mechanism and the fourth root transformation mechanism which result in linear and square root policy functions respectively. To see how this results in a slowly scaling policy function, we apply the fourth root transformation to the records from Example~\ref{ex:unit_splitting}.

\begin{exa} \label{ex:transformation_kth_root}
    Consider answering a sum over the Employees column on the sample data in Table~\ref{tab:sample_table} using the transformation mechanism. Let the transformation function $f(x) = \sqrt[4]{x}$ and the offset parameter $a = 0$. Likewise, assume Gaussian noise with scale $\sigma =2$. The per-record sensitivity of a sum is simply the value of the record being summed over -- in this case, the record's employee count. \par 
    The first two establishments incur the same privacy loss of $\frac{\sqrt[4]{5}^2}{8} \approx 0.3$. Establishment 3 incurs a privacy loss of $\frac{\sqrt[4]{10}^2}{8} \approx 0.4$. Establishment 4 incurs a privacy loss of $\frac{\sqrt[4]{20}^2}{8} \approx 0.6$. Establishment 5 incurs a privacy loss of $\frac{\sqrt[4]{30})^2}{8} \approx 0.7$. Establishment 6 incurs a privacy loss of $\frac{\sqrt[4]{10000}^2}{8} = 12.5$. 
\end{exa}
Compared to the unit splitting results, the privacy losses incurred by each establishment are much closer together, and large establishments incur significantly smaller privacy losses.

\subsection{Log Transformation Mechanism  $f(x) = \ln(x)$}

In order to achieve a policy function which scales even more slowly, we can use the $\ln$ function as the transformation.
\begin{cor}[]
\label{cor:log_transform}
Under the conditions of Assumption~\ref{as:transformInput}, with transformation function $f(x) = \ln(x)$, Algorithm~\ref{alg:transform}$(q(D), a, \sigma, f, g)$  satisfies $P$-PRzCDP with 
$$P(r) = \frac{(\ln(\Delta(r)+a) -\ln(a))^2}{2\sigma^2}$$
for any estimator $g$ and for all $a > 0$.
\end{cor}
\begin{proof}
This follows from a straightforward application of Theorem \ref{thm:transform}, with $f(x) = \ln(x)$. 
In this case, $f$ is clearly concave and strictly increasing on the range $[a, \infty)$ for $a > 0$. Therefore, this mechanism satisfies $P$-PRzCDP with 
$$P(r) = \frac{|f(\Delta(r)+a) - f(a)|^2}{2\sigma^2} = \frac{(\ln(\Delta(r)+a) -\ln(a))^2}{2\sigma^2}$$
\end{proof}

\begin{exa} \label{ex:transformation}
    Consider answering a sum over the Employees column on the sample data in Table~\ref{tab:sample_table} using the transformation mechanism. Let the transformation function $f(x) = \ln(x)$ and the offset parameter $a = 1$. Likewise, assume Gaussian noise with scale $\sigma =2$. The per-record sensitivity of a sum is simply the value of the record being summed over -- in this case, the record's employee count. \par 
    The first two establishments incur the same privacy loss of $\frac{\ln(6)^2}{8} \approx 0.4$. Establishment 3 incurs a privacy loss of $\frac{\ln(11)^2}{8} \approx 0.7$. Establishment 4 incurs a privacy loss of $\frac{\ln(21)^2}{8} \approx 1.2$. Establishment 5 incurs a privacy loss of $\frac{\ln(31)^2}{8} \approx 1.5$. Establishment 6 incurs a privacy loss of $\frac{\ln(10001)^2}{8} \approx 10.6$. 
\end{exa}
Like in the case of the fourth root transformation, the privacy losses for all records are closer together, and large records incur far less extreme privacy losses. In this case, smaller records incur a larger privacy loss under the log transformation when compared to the fourth root transformation. However, since the privacy loss scales more slowly than the fourth root transformation mechanism, larger records incur a smaller privacy loss. 

\subsection{Unbiased Estimators ($g$)} \label{subsubsec:Estimators}

In this section, we present mean-unbiased estimators of $q(D)$ for a rich set of transformation functions. Mean-unbiasedness may be essential in settings where users are expected to aggregate a large number of queries. For example, if block-level populations are published via a mean-unbiased mechanism, users can sum these to obtain mean-unbiased estimators of state or national populations. We also develop median-unbiased estimators in Appendix~\ref{sec:transMechMedianUnbiased}.

Per Algorithm~\ref{alg:transform}, the noisy transformed query value, $\tilde{v}$, is distributed according to $\tilde{v} \sim N(v, \sigma^2)$, with $v \equiv f(q(D) + a)$. Our problem, then, is to find an unbiased estimator of $f^{-1}(v)$; simply subtracting $a$ from this will give an unbiased estimator of $q(D)$. To do this, we use estimators from \cite{WashioEtAl1956} to obtain unbiased output from Algorithm~\ref{alg:transform} when the transformation function is of the forms $f(x) = \ln(x)$ and $f(x) = \sqrt[k]{x}$ for positive integers $k$.

\begin{thm}[Mean-Unbiased Estimators \citep{WashioEtAl1956}]
    \label{thm:meanUnbiasTransform}
    Let Assumption~\ref{as:transformInput} hold. Denote the output of Algorithm~\ref{alg:transform}\allowbreak$(q(D), a, \sigma, f, g)$ by $M(D)$. Denote the probabilist's $k^{th}$ Hermite polynomial by $\text{He}_k(x) \equiv (-1)^k e^{\frac{x^2}{2}} \frac{\partial^k}{\partial x^k} e^\frac{-x^2}{2}$.
    \begin{itemize}
      \item{If $a \geq 0$, $k$ is a positive integer, $f(x) = \sqrt[k]{x}$, and $g(x) = (-\sigma)^k \text{He}_k \pr*{-\frac{x}{\sigma}} - a$, then $\E[M(D)] = q(D)$ and $\V[M(D)] = \sum_{i=0}^{k-1} \binom{k}{i}^2(k-i)!\sigma^{2(k-i)}(q(D) + a)^{\frac{2i}{k}}$.}
      \item{If $a > 0$, $f(x) = \ln{(x)}$, and $g(x) = e^{x-\frac{\sigma^2}{2}} - a$, then $\E[M(D)] = q(D)$ and $\V[M(D)] = (e^{\sigma^2} - 1)(q(D) + a)^2$.}
    \end{itemize}
\end{thm}
It is worth noting that Theorem 1 of \cite{WashioEtAl1956} provides a formula that could be used to obtain mean-unbiased estimators for a very general class of transformation functions. This formula is unwieldy, however, involving an inverse Laplace transform. We restrict our analysis to the analytically convenient transformation functions treated already (see Table \ref{tab:unbiasedestimators}), and note that they induce an array of policy functions that should be wide enough for practical use.

\begin{table}[t]
	\centering
	\begin{tabular}{|c|c|c|c|}
		\hline
			\textbf{Name} & 
			$\boldsymbol{f(x)}$ &   
			$\boldsymbol{g(x)}$ &
			$\boldsymbol{\mathbb{V}[g(\tilde{v})]}$ \\
		\hline
			Identity & 
			$x$ &
			$x - a$ &
			$\sigma^2$ \\
		\hline
			Square Root & 
			$\sqrt{x}$ &
			$x^2 - \sigma^2 - a$ &
			$2\sigma^4 +4\sigma^2(q(D)+a)$ \\
		\hline
			Fourth Root & 
			$\sqrt[4]{x}$ &  
			$x^4 - 6x^2\sigma^2 + 3 \sigma^4 - a$ &
			$24\sigma^8 + 96\sigma^6(q(D) + a)^{\frac{1}{2}}$ \\
			{} & {} & {} & $ + 72 \sigma^4 (q(D) + a) + 16 \sigma^2 (q(D) + a)^{\frac{3}{2}}$ \\
		\hline
			$k$th Root & 
			$\sqrt[k]{x}$ &  
			$(-\sigma)^k \text{He}_k (-\frac{x}{\sigma}) - a$ &
			$\sum_{i=0}^{k-1} \binom{k}{i}^2(k-i)!\sigma^{2(k-i)}(q(D) + a)^{\frac{2i}{k}}$ \\
		\hline
			Log &  
			$\ln(x)$ & 
			$e^{x-\frac{\sigma^2}{2}} - a$ &
			$(e^{\sigma^2}-1)(q(D)+a)^2$ \\
		\hline
	\end{tabular}
	\caption{Mean-unbiased estimators and variances for selected transformation mechanisms}
	\label{tab:unbiasedestimators}
\end{table}

An important feature of all the estimators discussed here is that, due to their nonlinearity, the shape and scale of the transformation mechanism's distribution become dependent on the true query value, $q(D)$. In particular, the mechanisms' variances tend to grow in $q(D)$. This can be seen in the variance formulae in Theorem~\ref{thm:meanUnbiasTransform}.

\section{Additive Mechanisms}
\label{sec:additive}

In this section, we develop a class of additive mechanisms which achieve slowly scaling policy functions by drawing noise from relatively fat-tailed distributions. An additive mechanism simply adds data-independent random noise to the query. Algorithm~\ref{alg:additive} represents a generic additive mechanism.

\begin{algorithm}
\caption{\label{alg:additive} Additive mechanism }
\begin{algorithmic}[1]
    \Procedure{AdditivePrivatize}{Private query answer $q(D)$, probability density function $f_Z:\R \to [0,\infty)$}
    \State $Z \sim f_Z$
    \State $\tilde{S}\leftarrow q(D) +  Z$ 
    \State \textbf{output} $\tilde{S}$
    \EndProcedure
\end{algorithmic}    
\end{algorithm}

As an example, the Gaussian mechanism from Definition~\ref{def:gaussian_mech} is an additive mechanism. As discussed in Section~\ref{sec:transformation} (as a special case of the transformation mechanism), this has the policy function $P(r)=\frac{\Delta(r)^2}{\sigma^2}$ which grows at the same rapid rate as the unit splitting algorithm's policy function.
The additive mechanisms developed in this section add noise, $Z$, with densities of the form $f_Z(z) \propto e^{f\left(|z|\right)}$, where $f$ is a decreasing, convex function. These distributions are symmetric about 0, and the convexity of $f$ gives these densities a kink at 0 as well as tails that are thicker than Gaussian. This tail thickness makes these mechanisms' policy functions scale slowly. For intuition as to why, note that privacy loss is determined by the worst-case divergence between the mechanism's distributions at each of two neighboring databases. For an additive mechanism, the distributions realizing this worst-case divergence are the same up to a location shift equal to the per-record sensitivity. As the sensitivity grows, these distributions move apart, leading to greater divergence and greater privacy loss. When the distributions have fatter tails, however, their densities can still have substantial overlap even after large location shifts, leading to a lower divergence, and therefore lower privacy loss. This slows the growth of the privacy loss (i.e., the policy function) in the sensitivity, giving us slowly scaling behavior.

Furthermore, the additive mechanisms satisfy per-record differential privacy (PRDP), a stronger privacy concept than PRzCDP. This concept straightforwardly formalizes the suggestion by \cite{VLDBprivately} (page 3140) of a "pure $\epsilon$-DP" analogue of PRzCDP. To compare these mechanisms' privacy guarantees with those of the transformation mechanisms, we derive the PRzCDP guarantees implied by their PRDP guarantees.

\begin{defi}[$P$-Per-Record Differential Privacy ($P$-PRDP)]
\label{def:PRDP}
The randomized mechanism $M$ satisfies $P$-per-record differential privacy ($P$-PRDP) iff, for any two neighboring databases $D,D' \in \mathcal{D}$ satisfying $D \ominus D' = \{r\}$,
$$
d_\infty \infdivx*{M(D)}{M(D')} \leq P(r).
$$
\end{defi}

This privacy guarantee is stronger than PRzCDP in the sense that PRzCDP is implied by PRDP.

\begin{lem}[PRDP Implies PRzCDP]
\label{lem:PRDP-implies-PRzCDP}
   If a mechanism, $M$, is $P$-PRDP, it is also $P$-PRzCDP and $P'$-PRzCDP, where $P'(r) = \tanh(P(r)/2)P(r)$.
\end{lem}

\begin{proof}
    By Theorem 5 of \cite{pdp-to-zcdp}, for any fixed neighboring databases and for all $\alpha \in (1,\infty)$ the following holds. 
    $$d_\infty \infdivx*{M(D)}{M(D')} \leq \epsilon \implies d_\alpha \infdivx*{M(D)}{M(D')} \leq \alpha \frac{e^{\epsilon}-1}{e^{\epsilon}+1}\epsilon = \alpha \tanh(\epsilon/2)\epsilon \leq \alpha \epsilon $$

    Since this applies to any arbitrary neighboring databases the following is immediate.
    For any neighboring databases $D,D'$ where $D \ominus D' = \{ r \}$ the following holds:
    $$d_\infty \infdivx*{M(D)}{M(D')} \leq P(r) \implies d_\alpha \infdivx*{M(D)}{M(D')} \leq  \alpha \tanh(P(r)/2)P(r) \leq \alpha P(r) $$
    This yields $P'$-PRzCDP where $P'(r) = \tanh(P(r)/2)P(r)$ or, as a looser bound, $P'(r) = P(r)$.
\end{proof}

To derive the privacy guarantees of our additive mechanisms, we use intermediate results from the following two theorems. The first of these is a technical result establishing that, to determine the PRDP guarantee of an additive mechanism, we need to find bounds on the absolute log probability ratio of the noise distribution centered at 0 and the noise distribution centered at a value no greater than the per-record sensitivity.

\begin{thm}
\label{thm:AddMechLogProbForm}
    An additive mechanism $M\left(D\right) \equiv q\left(D\right) +  Z$, with a continuous random variable $Z$ having density $f_Z$, satisfies $P$-PRDP if and only if, for all neighboring databases $D,D^\prime \in \mathcal{D}$ satisfying $D \ominus D' = \{r\}$,
    $$\sup_{z} \left|\ln\left(f_Z\left(z\right)\right) - \ln\left(f_Z\left(z+q\left(D\right)-q\left(D^\prime\right)\right)\right)\right| \leq P\left(r\right).$$
\end{thm}

\begin{proof}
First, we note that the $\infty$-R\'enyi divergence between the distributions of the mechanism evaluated at the two databases $D$ and $D^\prime$ is 
\begin{align}
d_\infty\left(M\left(D\right) \| M\left(D^\prime\right)\right) &= \ln\left(\sup_y \frac{f_Z\left(y-q\left(D\right)\right)}
{f_Z\left(y-q\left(D^\prime\right)\right)}\right) \\
\label{eq:logProbRatioY}
&= \sup_y \ln\left(f_Z\left(y-q\left(D\right)\right)\right) - \ln\left(f_Z\left(y-q\left(D^\prime\right)\right)\right).
\end{align}

By the symmetry of the neighbor relation used in defining $P$-PRDP (that is, the fact that $D \ominus D^\prime = D^\prime \ominus D$), $M$ is $P$-PRDP if and only if, for all $D,D^\prime \text{ such that } D \ominus D^\prime = {r},$
\begin{align}
    \label{eq:maxSupRenyiDivPRDP}
    \max\left[d_\infty(M(D) || M(D^\prime)), d_\infty(M(D^\prime) || M(D))\right] \leq P(r).
\end{align}

Substituting Equation~\ref{eq:logProbRatioY} into the left-hand-side of Inequality (\ref{eq:maxSupRenyiDivPRDP}), we can simplify and substitute in $z=y-q(D)$ to obtain Theorem~\ref{thm:AddMechLogProbForm} as follows:

\begin{align}
\max\left[d_\infty(M(D) || M(D^\prime)), d_\infty(M(D^\prime) \right.&\left.|| M(D))\right] \nonumber\\
&= \max\left[\sup_y \ln(f_Z(y-q(D))) - \ln(f_Z(y-q(D^\prime))), \nonumber \right.\\ 
&\left. \,\,\,\,\,\,\,\, \sup_y \ln(f_Z(y-q(D^\prime))) - \ln(f_Z(y-q(D)))\right] \\
\intertext{Noting that the $\sup$ and $\max$ commute, and factoring $-1$ out of the second expression, we rearrange the above to get}\nonumber
&= \sup_y \max\left[\ln(f_Z(y-q(D))) - \ln(f_Z(y-q(D^\prime))), \nonumber \right.\\
&\left. \,\,\,\,\,\,\,\, -\left[\ln(f_Z(y-q(D))) - \ln(f_Z(y-q(D^\prime)))\right]\right] \\
&= \sup_y \left|\ln(f_Z(y-q(D))) - \ln(f_Z(y-q(D^\prime)))\right| \\
&= \sup_{z} \left|\ln(f_Z(z)) - \ln(f_Z(z+q(D)-q(D^\prime)))\right|.
\end{align}

\end{proof}

The second result, in Theorem~\ref{thm:BndLogProbSymCnvxDens} below, is the cornerstone of our analysis. It shows that the supremal absolute log probability ratio from Theorem~\ref{thm:AddMechLogProbForm} is easy to find when the noise distribution's density is symmetric and log-convex on either side of the mode. This is the class of distributions we consider for use in our slowly scaling additive mechanisms.

\begin{thm}[Bounds on Log-Probability Ratio for Exponential Absolute Convex Densities]
\label{thm:BndLogProbSymCnvxDens}
Consider the probability density $f_Z(z) \propto e^{f\left(|z|\right)}$ for $z \in \R$, where $f:[0,\infty) \to \R$ is (weakly) decreasing and (weakly) convex. The log probability ratio of this distribution and this distribution location-shifted by $l$ is tightly bounded as follows:
$$\sup_{z} \left|\ln\left(f_Z\left(z\right)\right) - \ln\left(f_Z\left(z-l\right)\right)\right| = f\left(0\right) - f\left(|l|\right).$$
\end{thm}

\begin{proof}
Denote the supremand in  Theorem~\ref{thm:BndLogProbSymCnvxDens} by $g(z;l) \equiv \left|\ln(f_Z(z)) - \ln(f_Z(z-l))\right|$. First, note that $$g(z;l) = \left|f\left(|z|\right) - f\left(|z-l|\right)\right|.$$

Now, observe that $g(z;l) = g(-z;-l)$. This implies that $\sup_{z} g(z;l) = \sup_{z} g(z;|l|)$ because $\sup_{z} g(z;l) = \sup_{z} g(-z;-l) = \sup_{z} g(z;-l)$.

Similarly, observe that $g(z;l)$ is symmetric about $z = .5l$. That is, $g(z;l) = g(.5l-(z-.5l);l) = g(l-z;l)$. We can therefore limit our search for the maximizing value of $z$ to $z \geq .5l$.

Combining these two observations, our problem is now to solve
$$\sup_{z \geq .5|l|} \left|f\left(|z|\right) - f\left(|z-|l||\right)\right|.$$

To do this, we will show that $g(z;|l|) \leq g(l;|l|)$ separately for the two cases where $z \in [.5|l|, |l|]$ and where $z \in [|l|, \infty].$ This will conclude our proof.

To simplify notation in analyzing the two cases, we will assume WLOG that $l \geq 0$, so that $|l| = l$.

\subsubsection*{First case: $z \in [.5l,l]$}
For $z \in [.5l, l]$, we have
$$g(z;l) = \left|f\left(|z|\right) - f\left(|z-l|\right)\right| = \left|f\left(z\right) - f\left(l-z\right)\right|.$$

Further, $z \in [.5l,l]$ and $l \geq 0$ implies that $0 \leq l-z \leq .5l \leq z$. Because $f$ is a decreasing function, $f(z) \leq f(l-z)$, and we have 
$$g(z;l) = \left|f\left(z\right) - f\left(l-z\right)\right| = f\left(l-z\right) - f\left(z\right).$$

Because $f$ is decreasing, both $f(l-z)$ and $-f(z)$ are increasing in $z$. For $z \in [.5l,l]$, then, $g(z;l)$ is maximized at $z=l$, where 

\begin{equation} \label{eqObjFuncValAtL}
g(l;l) = f(0) - f\left(l\right).    
\end{equation}

\subsubsection*{Second case: $z \in [l,\infty]$}
In the case where $z \in [l,\infty)$, we have, by reasoning like that in the first case,

\begin{equation} \label{eqObjFuncScndCase}
g(z;l) = \left|f\left(|z|\right) - f\left(|z-l|\right)\right| = \left|f\left(z\right) - f\left(z-l\right)\right| = f\left(z-l\right) - f\left(z\right).
\end{equation}

We want to show that $g(z;l) \leq g(l;l)$ for all $z \in [l,\infty)$. Plugging equations \ref{eqObjFuncValAtL} and \ref{eqObjFuncScndCase} into this inequality, this means we want to demonstrate that
\begin{align} 
f\left(z-l\right) - f\left(z\right) &\leq f(0) - f\left(l\right) \text{ or, equivalently,} \\
f\left(l\right) + f\left(z-l\right) &\leq f\left(z\right) + f(0). \label{eqTargetIneqScndCase}
\end{align}

We now invoke the convexity of $f$ (and recall that $z \geq l$) to derive the following two inequalities:
\begin{align}
f\left(l\right) &= f\left(\frac{l}{z}z + \left(1-\frac{l}{z}\right)0\right) \leq \frac{l}{z}f\left(z\right) + \left(1-\frac{l}{z}\right)f(0). \label{eqCnvxScndCase1} \\
f\left(z-l\right) &= f\left(\frac{z-l}{z}z + \left(1-\frac{z-l}{z}\right)0\right) \leq \frac{z-l}{z}f\left(z\right) + \left(1-\frac{z-l}{z}\right)f(0). \label{eqCnvxScndCase2}
\end{align}

Summing inequalities \ref{eqCnvxScndCase1} and \ref{eqCnvxScndCase2} gives us

\begin{align}
f\left(l\right) + f\left(z-l\right) &\leq \frac{l}{z}f\left(z\right) + \left(1-\frac{l}{z}\right)f\left(0\right) + \frac{z-l}{z}f\left(z\right) + \left(1-\frac{z-l}{z}\right)f\left(0\right)\nonumber \\
&= \frac{l}{z}f\left(z\right) + \frac{z-l}{z}f\left(0\right) + \frac{z-l}{z}f\left(z\right) + \left(1-\frac{z-l}{z}\right)f\left(0\right) \nonumber \\
&= f\left(z\right) + f\left(0\right).
\end{align}

This is \eqref{eqTargetIneqScndCase}, proving the second case.

\end{proof}

The following is our main result, giving the PRDP and PRzCDP guarantees for the additive mechanisms we consider.

\begin{thm}[Privacy of Symmetric, Absolute-Log-Convex Densities]
\label{thm:PrivSymCnvxDens}
Let $q$ be a query with per-record sensitivity $\Delta(r)$, as defined in Definition \ref{def:PR-sensitivity}. Let $f_Z$ be a noise distribution that has density proportional to $e^{f\left(|z|\right)}$, where $f:[0,\infty) \to \R$ is (weakly) decreasing and (weakly) convex.

For all records $r$, Algorithm~\ref{alg:additive}$(q(D), f_Z)$, denoted below by $M$, satisfies $P$-PRDP  with $P(r)=f(0) - f(\Delta(r))$. This policy function is tight, in the sense that $M$ is not $P^\prime$-PRDP for an alternative policy function $P^\prime$ such that there exists an $r$ satisfying $P^\prime(r) < P(r)$. $M$ also satisfies $P''$-PRzCDP with $P''(r) = \tanh(P(r)/2) P(r)$.
\end{thm}

\begin{proof}
Applying Theorem~\ref{thm:AddMechLogProbForm} and then Theorem~\ref{thm:BndLogProbSymCnvxDens}, we see that this algorithm is $P$-PRDP if and only if, for all $r$,
\begin{align}
    P(r) &\geq \sup_{D,D^\prime \text{ such that } D \ominus D^\prime = {r}} \sup_z \left|\ln(f_Z(z)) - \ln(f_Z(z+q(D)-q(D^\prime)))\right| \\
    &= \sup_{D,D^\prime \text{ such that } D \ominus D^\prime = {r}} f(0) - f\left(|q(D)-q(D^\prime)|\right).
\end{align}

Because $f$ is a decreasing function, $f(0) - f\left(|q(D)-q(D^\prime)|\right)$ is increasing in $|q(D)-q(D^\prime)|$ and has its supremum at $|q(D)-q(D^\prime)| = \Delta(r)$. Therefore, it satisfies $P$-PRDP with $P(r)=f(0) - f(\Delta(r))$. Furthermore, by Lemma~\ref{lem:PRDP-implies-PRzCDP}, it also satisfies $P''$-PRzCDP with $P''(r) = \tanh(P(r)/2) P(r)$.
\end{proof}

The appeal of this class of mechanisms lies in the ability to flexibly choose $P(r)$ by choosing $f$. Intuitively, the more convex $f$ is, the fatter the noise distribution's tails are, and the more slowly privacy loss grows. The next sections provide examples of slowly scaling additive mechanisms for different choices of $f$. In the first example, we show that a slowly scaling mechanism can be achieved by adding noise from a generalized Gaussian distribution, while our second example uses a distribution that, to our knowledge, is novel.

\subsection{Generalized Gaussian Mechanism $\left(f\left(|z|\right) = -\left(\frac{|z|}{\sigma}\right)^p\right)$}
\begin{cor}[]
\label{cor:genGaussMech}
Under the conditions of Theorem~\ref{thm:PrivSymCnvxDens}, with density function $f_Z(z) = \frac{p}{2\sigma\Gamma(\frac{1}{p})}e^{-\left(\frac{|z|}{\sigma}\right)^p}$, where $0 < \sigma$ and $0 \leq p \leq 1$, Algorithm~\ref{alg:additive}$(q(D), f_Z)$ satisfies $P$-PRDP and $P'$-PRzCDP where $P(r) = \left(\frac{\Delta(r)}{\sigma}\right)^p$ and $P'(r) = \tanh(P(r)/2)P(r)$.

\end{cor}

\begin{proof}
This follows from a straightforward application of Theorem \ref{thm:PrivSymCnvxDens}, with $f\left(|z|\right) = -\left(\frac{|z|}{\sigma}\right)^p$. 
In this case, $f\left(|z|\right)$ is clearly decreasing in $|z|\in[0,\infty)$, and, for $0 \leq p \leq 1$, $f\left(|z|\right)$ is convex. For this range of $p$, then, this mechanism satisfies $P$-PRDP and $P'$-PRzCDP where $P(r) = f(0) - f\left(\frac{\Delta(r)}{\sigma}\right) = \left(\frac{\Delta(r)}{\sigma}\right)^p$ and $P'(r) = \tanh(P(r)/2)P(r)$.
\end{proof}

The mechanism in Corollary~\ref{cor:genGaussMech} has a generalized Gaussian distribution with variance $\sigma^2\frac{\Gamma(\frac{3}{p})}{\Gamma(\frac{1}{p})}$, per \cite{Nadarajah2005}. 
The generalized Gaussian mechanism has been used in the context of differential privacy before \citep{Liu2019,GaneshZhao2021}, but our use of it to achieve a slowly scaling policy function for per-record variants of DP is novel.

\subsection{Exponential Polylogarithmic Distribution $\left(f\left(|z|\right) = -d\ln\left(\frac{|z|}{\sigma} + a\right)^p\right)$}

\begin{cor}[]
\label{cor:expolyMech}
    Under the conditions of Theorem~\ref{thm:PrivSymCnvxDens}, with density function $f_Z(z) \propto e^{-d\ln\left(\frac{|z|}{\sigma}+a\right)^p}$, where $\sigma > 0$, $a \geq e^{p-1}$, and either ($p>1$ and $d>0$) or ($p=1$ and $d>1$), Algorithm~\ref{alg:additive}$(q(D), f_Z)$  satisfies $P$-PRDP and $P'$-PRzCDP with $P(r) = d\left[\ln\left(\frac{\Delta(r)}{\sigma} + a\right)^p - \ln(a)^p\right]$ and $P'(r) = \tanh(P(r)/2)P(r)$.
\end{cor}

\begin{proof}
This result follows from Theorem \ref{thm:PrivSymCnvxDens} once we show that $f\left(|z|\right) = -d\ln\left(\frac{|z|}{\sigma} + a\right)^p$ is decreasing and convex. First, the function is decreasing in $|z|$ due to $-\ln\left(\frac{|z|}{\sigma}+a\right)$ being non-positive and decreasing whenever $1\leq a$. This holds because $p\geq 1$ and $e^{p-1}\leq a$. Second, the function's second derivative,
\begin{equation*}
-\frac{dp}{\left(\frac{|z|}{\sigma}+a\right)^{2}}\left[(p-1)\ln\left(\frac{|z|}{\sigma}+a\right)^{p-2}-\ln\left(\frac{|z|}{\sigma}+a\right)^{p-1}\right],
\end{equation*}
is non-negative when $p-1\leq \ln\left(\frac{|z|}{\sigma}+a\right)$. This condition holds when $e^{p-1}-a \leq \frac{|z|}{\sigma}$, which follows from the assumption that $e^{p-1}\leq a$ in Corollary \ref{cor:expolyMech}.

\end{proof}
We call the noise distribution used in Corollary~\ref{cor:expolyMech} the \textit{exponential polylogarithmic distribution}. This distribution is novel, to the best of our knowledge. It has a tractable density function, CDF, quantile function, and variance when $p \in \{1,2\}$, as seen in Table~\ref{tab:Foxdenscases}. See Appendix~\ref{sec:ExpPloyLog} for the general case, where these quantities involve obscure special functions.

\begin{table}[H]
	\centering
	\begin{tabular}{|c|c|c|}
		\hline
		\centering  & $\boldsymbol{p=1}$ & $\boldsymbol{p=2}$ \\
		\hline
		\textbf{Density} & $\frac{d-1}{2\sigma}a^{d-1}\left(\frac{|z|}{\sigma}+a\right)^{-d}\{d>1\}$ & $\frac{\sqrt{d}e^{-d\ln(\frac{|z|}{\sigma}+a)^{2}}}{2e^{\frac{1}{4d}}\sigma\sqrt{\pi}\left[1-\kappa\left(a, \frac{1}{2}\right)\right]}$  \\
		\hline
		\textbf{CDF} & $\frac{1}{2}+\frac{\text{sign}(z)}{2}\times\left(1-\left(\frac{|z|}{\sigma a}+1\right)^{1-d}\right)\{d>1\}$ & $\frac{1}{2}+\frac{\text{sign}(z)}{2}\left[\frac{\kappa\left(\frac{|z|}{\sigma}+a, \frac{1}{2}\right)-\kappa\left(a, \frac{1}{2}\right)}{1-\kappa\left(a, \frac{1}{2}\right)}\right]$  \\
  \hline
  \textbf{Quantile} & $ \makecell{
  \sigma a\text{sign}(q-0.5)\left(\left[1-2|q-0.5|\right]^{\frac{1}{1-d}}-1\right)
  \\
  \{d>1\}
  }$
  & $\makecell{\sigma\text{sign}(q-0.5)\left(\exp\left[(2d)^{-\frac{1}{2}}\times
  \right. \right. 
  \\
  \Phi^{-1}\left(2|q-0.5| \left[1-\kappa\left(a, \frac{1}{2}\right)\right] + \kappa\left(a, \frac{1}{2}\right)\right)
  \\
  \left. \left.+(2d)^{-1}\right]-a\right)}$\\
\hline
\textbf{Variance} & $\sigma^{2}a^{2}(d-1)\left[\frac{1}{d-3}-\frac{2}{d-2}+\frac{1}{d-1}\right]\{d>3\}$ & $\sigma^{2}\frac{e^{\frac{2}{d}}[1-\kappa(a, \frac{3}{2})]-2ae^{\frac{3}{4d}}[1-\kappa(a, 1)]}{[1-\kappa(a, \frac{1}{2})]}+a^{2}\sigma^{2}$ \\
		\hline
	\end{tabular}
 \caption{Special cases of the density, CDF, quantile, and variance of the exponential polylogarithmic distribution for $p=1$ and $p=2$ where $q\in[0,1], \kappa(x,y)=\Phi\left(\frac{\ln(x)-yd^{-1}}{(2d)^{-1/2}}\right)$, and $\Phi(\cdot)$ is the CDF of a standard normal distribution. The notation $\{d>.\}$ in the second column indicates conditions required for each cell's quantity to exist or be finite.}
 \label{tab:Foxdenscases}
\end{table}	

Given that this is an atypical family of distributions, it is helpful to visualize its density alongside the more commonly used Normal distribution.  Figure~\ref{fig:ep_pdfs} plots PDFs for random values sampled from equal variance and equal privacy loss distributions.  Tail weights for noise from equal privacy loss mechanisms are given in Table~\ref{tab:ep_probs}.

\begin{figure}[ht]
    \centering
    \includegraphics[width = \textwidth]{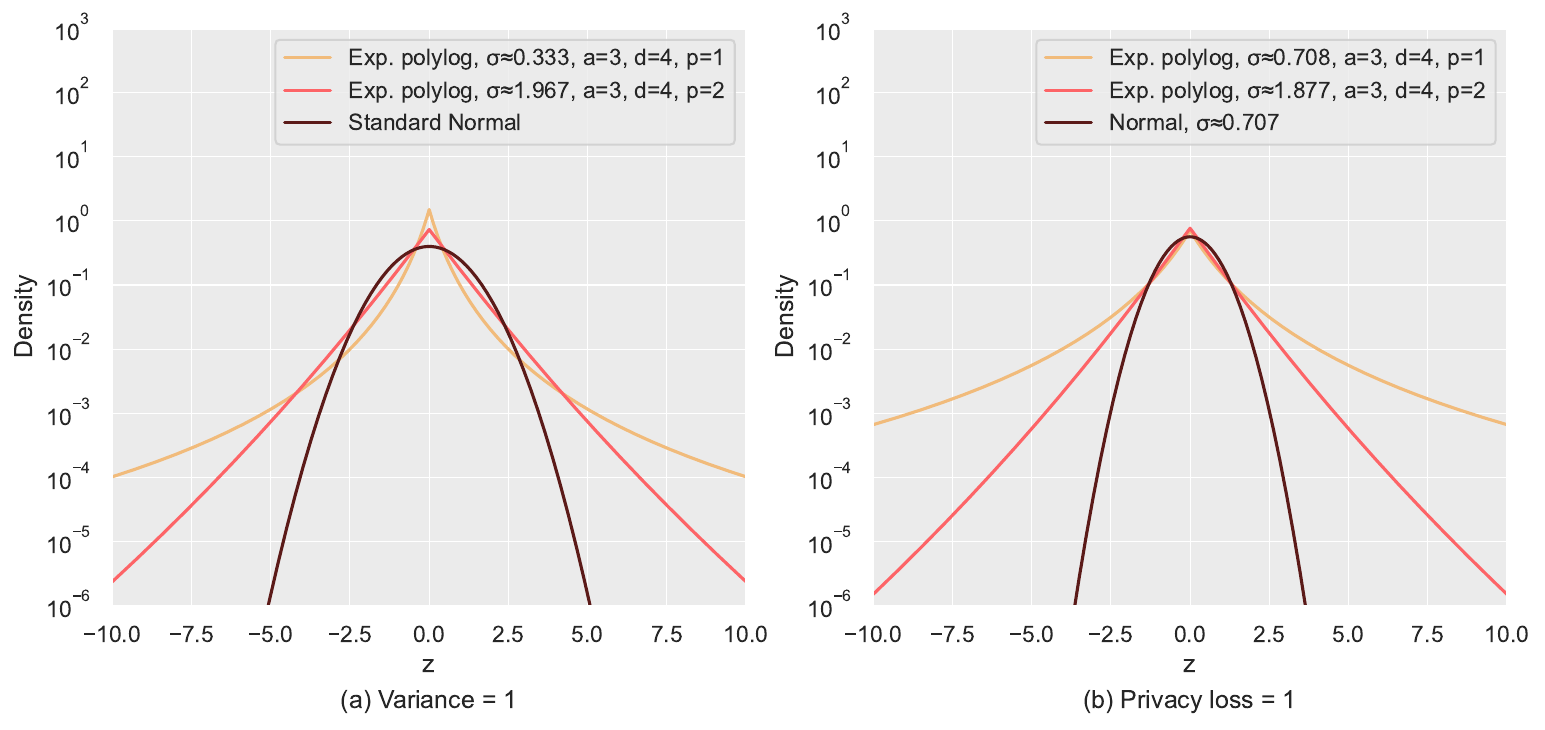}
    \caption{Probability densities for two exponential polylogarithmic distributions and a Normal distribution.  Distributions in plot (a) have unit variance.  Additive mechanisms using distributions in plot (b) have PRzCDP privacy loss of $1$ when $\Delta(r) = 1$ (that is, mechanism $i$ satisfies $P_i$-PRzCDP such that $P_i(1)=1$).}
    \label{fig:ep_pdfs}
\end{figure}

\begin{table}[ht]
	\centering
	\begin{tabular}{|c|c|c|c|c|}
		\hline
		\centering \textbf{Mechanism} & $\boldsymbol{P[|z|>1]}$ & $\boldsymbol{P[|z|>2]}$ & $\boldsymbol{P[|z|>3]}$ & $\boldsymbol{P[|z|>4]}$ \\
		\hline
		Exp polylog, $\sigma \approx 0.708$, $a=3$, $d=4$, $p=1$ & 0.314 & 0.137 & 0.071 & 0.042 \\
		\hline
		Exp polylog, $\sigma \approx 1.877$, $a=3$, $d=4$, $p=2$ & 0.221 & 0.051 & 0.013 & 0.003 \\
		\hline
		Gaussian,        $\sigma \approx 0.707$                      & 0.157 & 0.005 & 0.000 & 0.000 \\
		\hline
	\end{tabular}
	\caption{Tail weights for noise from additive mechanisms applied to sum queries with parameters set such that PRzCDP privacy loss is $1$ when $\Delta(r) = 1$.}
	\label{tab:ep_probs}
\end{table}

\section{Empirical Results}
\label{sec:empirical}

\begin{table}[t]
    \centering
    \begin{tabular}{|c|c|c|c|}
  \multicolumn{4}{c}{Transformation Mechanism with Transformation $f$}\\\hline
    \hline
     \textbf{Name} & $\boldsymbol{f(x)}$ &   
\textbf{Variance} &   
\textbf{PRzCDP Policy Function}\\
    \hline
     Identity & $x$ & $\sigma^2$ & $\frac{\Delta(r)^2}{2\sigma^2}$\\
    \hline
     Square Root & $\sqrt{x}$ & $2\sigma^4 +4\sigma^2(q(D)+a)$ & $\frac{(\sqrt{\Delta(r)+a} -\sqrt{a})^2}{2\sigma^2}$\\
    \hline
     Fourth Root & $\sqrt[4]{x}$ & $24\sigma^8 + 96\sigma^6(q(D) + a)^{\frac{1}{2}}$ & $\frac{(\sqrt[4]{\Delta(r)+a} -\sqrt[4]{a})^2}{2\sigma^2}$ \\
     {} & {} & $ + 72 \sigma^4 (q(D) + a) + 16 \sigma^2 (q(D) + a)^{\frac{3}{2}}$ & {} \\
    \hline
    $k$th Root & $\sqrt[k]{x}$ 
    & $\sum_{i=0}^{k-1} \binom{k}{i}^2(k-i)!\sigma^{2(k-i)}(q(D) + a)^{\frac{2i}{k}}$
    & $\frac{(\sqrt[k]{\Delta(r)+a} -\sqrt[k]{a})^2}{2\sigma^2}$ \\
    \hline
    Log &  $\ln(x)$ & $(e^{\sigma^2}-1)(q(D) +a)^2$ & $\frac{(\ln (\Delta(r) + a) -\ln(a))^2}{2\sigma^2}$\\
    \hline
 \multicolumn{4}{c}{}\\
  \multicolumn{4}{c}{Additive Mechanism with Noise Density $\propto e^{f(|z|)}$}\\\hline\hline
  \textbf{Name} & $\boldsymbol{f(|z|)}$& \textbf{Variance} & \textbf{PRDP Policy Function}\\\hline
  Generalized & $-\left(\frac{|z|}{\sigma}\right)^p$&     $\sigma^2\frac{\Gamma(\frac{3}{p})}{\Gamma(\frac{1}{p})}$&$(\frac{\Delta(r)}{\sigma})^p$ \\
  Gaussian & & & \\ \hline
  Exponential & & & \\
  Polylog & $-d\ln(\frac{|z|}{\sigma} + a)^2$ & $\sigma^{2}\frac{e^{\frac{2}{d}}[1-\kappa(a, \frac{3}{2})]-2ae^{\frac{3}{4d}}[1-\kappa(a, 1)]}{[1-\kappa(a, \frac{1}{2})]}+a^{2}\sigma^{2}$  & $d(\ln(\frac{\Delta(r)}{\sigma} + a)^2 - \ln(a)^2) $\\ ($p=2$) & & & \\ \hline
    \end{tabular}
    \caption{Variance and policy functions for mean-unbiased transformation mechanisms and additive mechanisms. Here, $\kappa(x,y)=\Phi\left(\frac{\ln(x)-yd^{-1}}{(2d)^{-1/2}}\right)$, and $\Phi(\cdot)$ is the CDF of a standard normal distribution. The PRDP policy functions of the additive mechanisms are also valid PRzCDP policy functions (see Lemma~\ref{lem:PRDP-implies-PRzCDP}), though tighter PRzCDP policy functions are given in Corollaries \ref{cor:genGaussMech} and \ref{cor:expolyMech}.}
   \label{tab:mechSummary}
\end{table}
In this section, we empirically evaluate the mechanisms presented above. In the following subsections, we briefly specify the datasets and queries we use, the mechanisms we chose, and the metrics we calculate before discussing the results.

\subsection{Datasets}
We conduct these experiments using groupby sums on two different datasets. The first, provided to us by the U.S. Census Bureau, is a simulated version of the data used to compute the Census Bureau's County Business Patterns (CBP) data product. Each record in the data corresponds to a single establishment (i.e., a single location of a single business). 

The groupby sums for the CBP data are grouped by the attributes NAICS3 (the three-digit North American Industry Classification System Code) and COUNTY. These sums are computed for three attributes: EMP (number of employees), PAYQTR1 (first quarter payroll, in thousands of dollars), and PAYANN (annual payroll, in thousands of dollars). For this experiment, we take a 5\% sample of all NAICS3$\times$COUNTY combinations (groups) in the dataset and use all records in the sampled groups. This results in a sample of 1,095 groups, totaling 323,402 records.

The second dataset is from the U.S. Department of Agriculture's (USDA's) Cattle Inventory Survey (CIS), which is a product of the National Agricultural Statistical Service~\citep{agcensus}. This dataset contains county-level survey records of total cattle inventory. In contrast to the simulated CBP dataset, records on individual establishments (farms) are not available. Therefore, we consider totals at the agricultural district level, which is nested within state, as our privacy unit. Our groupby sums are computed using State as the key, which groups 35,556 values into 43 groups (states). We perform this experiment on the attribute Cattle Value (total cattle inventory). 

Table~\ref{tab:summary-stats} provides summary statistics for each of the variables we consider in these experiments, where Cattle Value corresponds to the attribute from the USDA Cattle Inventory Survey and the remaining three attributes correspond to the simulated CBP subsample. Figures~\ref{fig:synth-violin} and ~\ref{fig:cattle-hist} also display the distributions of the variables we sum over as well as the grouped sums for each dataset. No noise has been added to the data or sums in these figures.

\begin{table}[h]
\centering
\begin{tabular}{|c|c|c|c|c|c|c|c|}
\hline
                      & \textbf{Min} & \textbf{1st Quartile} & \textbf{Median} & \textbf{Mean} & \textbf{3rd Quartile} & \textbf{Max} & \textbf{N} \\ \hline
\textbf{Cattle Value} & 100          & 3,100                 & 10,200          & 22,480        & 26,000                & 1,341,000  & 5,926  \\ \hline
\textbf{EMP}          & 0            & 1                     & 2               & 17        & 9                     & 19,589      & 323,402  \\ \hline
\textbf{PAYANN}       & 0            & 30                    & 106             & 919       & 388                   & 1,553,359   & 323,402   \\ \hline
\textbf{PAYQTR1}      & 0            & 7                     & 25              & 239       & 93                    & 402,324   & 323,402    \\ \hline
\end{tabular}
\caption{Summary statistics for the columns that are summed over in each of the experiments for both the cattle and sampled simulated CBP data set}
\label{tab:summary-stats}
\end{table}

\begin{figure}[t]
    \centering
    \includegraphics[width = 0.7\textwidth]{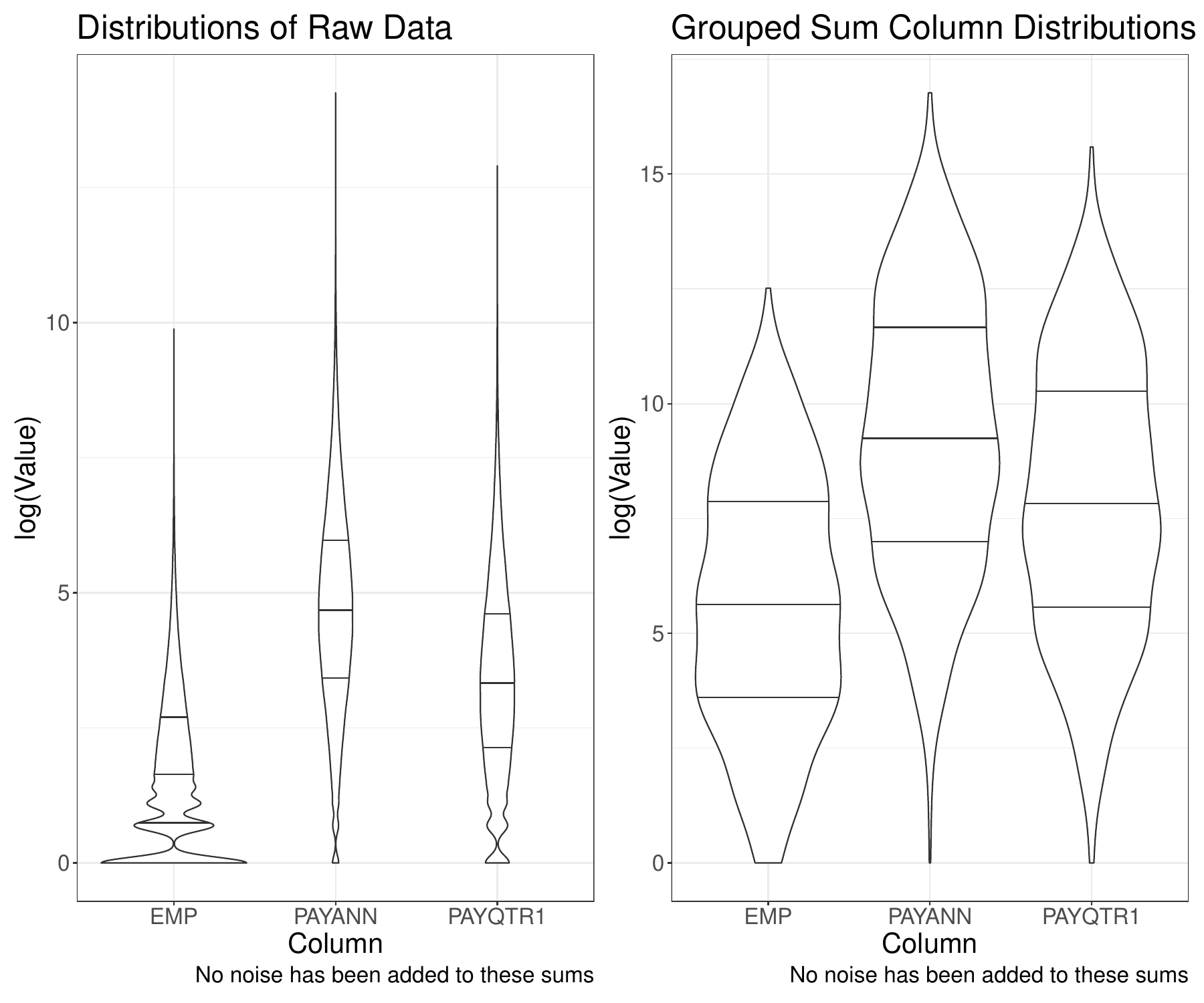}
    \caption{This figure shows violin plots for the three variables we consider from the sampled simulated CBP data. The image on the left shows the raw data and the image on the right shows the data grouped and summed by NAICS3$\times$COUNTY. There are 1,095 groups in total. The $y$-axis is on the log scale since all three variables contain large outliers. The black horizontal lines in each of the violins represent the 25\%, 50\%, and 75\% quantiles, respectively.}
    \label{fig:synth-violin}
\end{figure}

\begin{figure}[t]
    \centering
    \includegraphics[width = 0.7\textwidth]{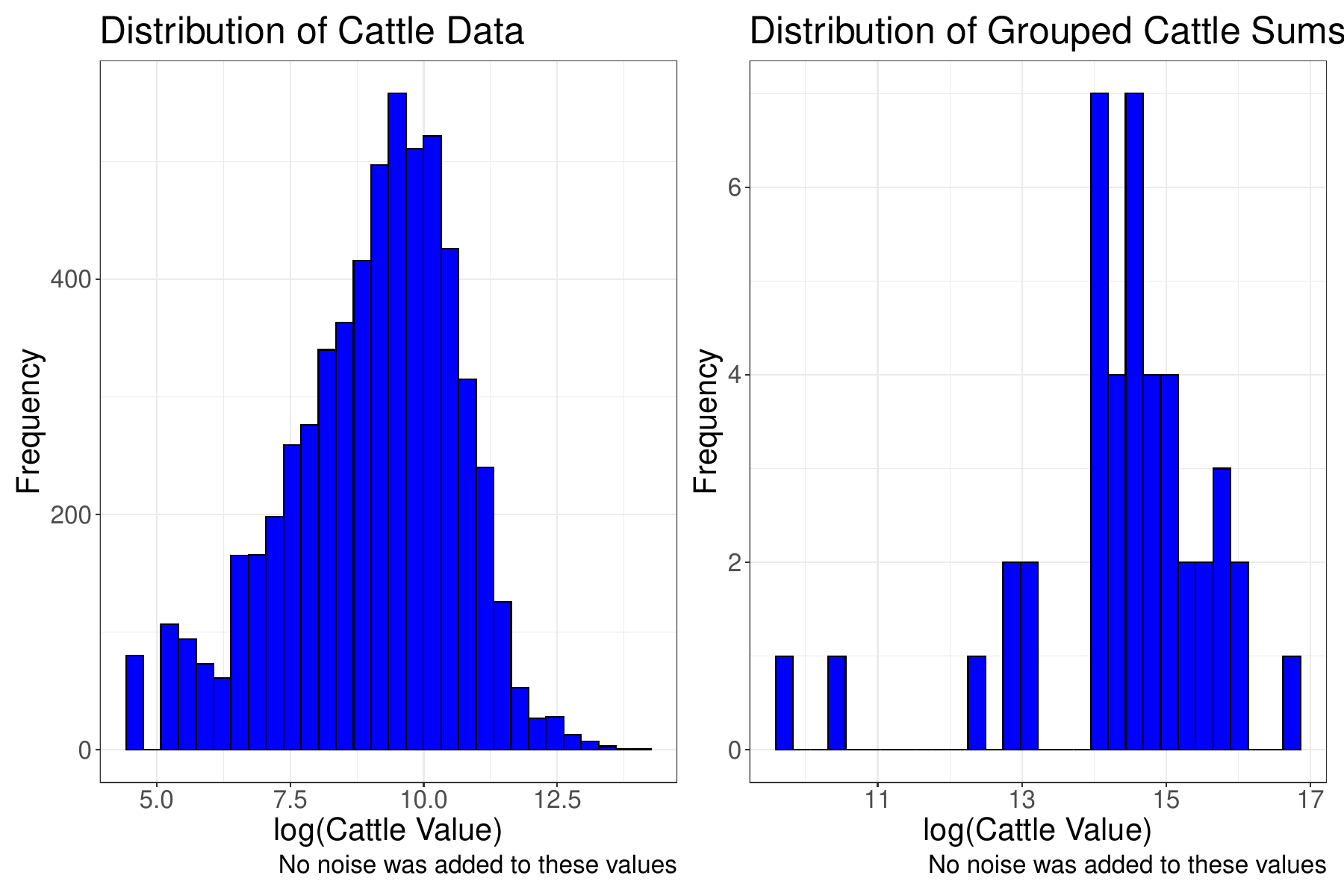}
    \caption{This histogram displays the distribution of the number of cattle in the USDA Cattle Inventory Survey dataset. The image on the left shows the raw data and the image on the right shows the cattle data grouped and summed by state. There are 43 groups in total. The $x$-axis is on the log scale since the dataset contains many large outliers.}
    \label{fig:cattle-hist}
\end{figure}

\subsection{Mechanisms}
We apply our mechanisms - developed for univariate queries - to the groupby sums by applying them independently to each variable's sum for each group.  

We select transformation and additive mechanisms whose PRzCDP policy functions scale at each of the $O(\Delta(r)^2)$, $O(\sqrt{\Delta(r)})$, and $O(\ln(\Delta(r))^2)$ rates. The transformation mechanisms use the identity, fourth root, and log transformations, and the additive mechanisms use the Gaussian, generalized Gaussian, and exponential polylog noise distributions. See Table~\ref{tab:mechSummary} for a summary of these mechanisms' policy and variance functions.

\begin{table}[]
\begin{tabular}{|c|c|c|c|c|}
\hline
\textbf{Mechanism}                                                 & \textbf{\begin{tabular}[c]{@{}c@{}}Parameter Values\\ (Cattle)\end{tabular}} & \textbf{\begin{tabular}[c]{@{}c@{}}Parameter Values\\ (EMP)\end{tabular}} & \textbf{\begin{tabular}[c]{@{}c@{}}Parameter Values\\ (PAYQTR1)\end{tabular}} & \textbf{\begin{tabular}[c]{@{}c@{}}Parameter Values\\ (PAYANN)\end{tabular}} \\ \hline
Gaussian                                                           & $\sigma = 7212$                                                                              & $\sigma = 1.41$                                                                        & $\sigma = 17.7$                                                                             & $\sigma = 75$                                                                               \\ \hline
\begin{tabular}[c]{@{}c@{}}Generalized \\ Gaussian\end{tabular}    & \begin{tabular}[c]{@{}c@{}}$\sigma = 658$,\\ $p = 0.5$\end{tabular}                            & \begin{tabular}[c]{@{}c@{}}$\sigma = 0.129$,\\ $p = 0.5$\end{tabular}                    & \begin{tabular}[c]{@{}c@{}}$\sigma = 1.61$,\\ $p = 0.5$\end{tabular}                          & \begin{tabular}[c]{@{}c@{}}$\sigma = 6.84$,\\ $p = 0.5$\end{tabular}                          \\ \hline
\begin{tabular}[c]{@{}c@{}}Exponential \\ Polylog\end{tabular}    & \begin{tabular}[c]{@{}c@{}}$\sigma = 1$,\\ $p = 2$,\\ $d = 0.113$,\\ $a = e$\end{tabular}          & \begin{tabular}[c]{@{}c@{}}$\sigma = 1$,\\ $p = 2$, \\ $d = 0.595$,\\ $a = e$\end{tabular}   & \begin{tabular}[c]{@{}c@{}}$\sigma = 1$,\\ $p = 2$, \\ $d = 0.337$,\\ $a = e$\end{tabular}        & \begin{tabular}[c]{@{}c@{}}$\sigma = 1$,\\ $p = 2$, \\ $d = 0.231$,\\ $a = e$\end{tabular}        \\ \hline
\begin{tabular}[c]{@{}c@{}}Identity \\ Transformation\end{tabular} & $\sigma = 7212$                                                                              & $\sigma = 1.41$                                                                        & $\sigma = 17.7$                                                                             & $\sigma = 75$                                                                               \\ \hline
Fourth Root                                                        & \begin{tabular}[c]{@{}c@{}}$\sigma = 1.67$,\\ $a = 0$\end{tabular}                             & \begin{tabular}[c]{@{}c@{}}$\sigma = 0.198$,\\ $a = 0$\end{tabular}                      & \begin{tabular}[c]{@{}c@{}}$\sigma = 0.372$,\\ $a = 0$\end{tabular}                           & \begin{tabular}[c]{@{}c@{}}$\sigma = 0.534$,\\ $a = 0$\end{tabular}                           \\ \hline
Log                                                                & \begin{tabular}[c]{@{}c@{}}$\sigma = 0.637$,\\ $a = 1$\end{tabular}                            & \begin{tabular}[c]{@{}c@{}}$\sigma = 0.448$,\\  $a = 1$\end{tabular}                     & \begin{tabular}[c]{@{}c@{}}$\sigma = 0.616$,\\  $a = 1$\end{tabular}                          & \begin{tabular}[c]{@{}c@{}}$\sigma = 0.632$,\\  $a = 1$\end{tabular}                          \\ \hline
\end{tabular}
\caption{This table shows the mechanism parameters that were used for each attribute in the experiments.}
\label{tab:experiment-params}
\end{table}

For each attribute summed over in our groupby sums, we set the parameters of the identity transformation mechanism and the equivalent Gaussian additive mechanism to have a privacy loss of 1 for the median record. This entails setting these mechanisms' standard deviations to $\sqrt{.5}$ times the median of each attribute's raw data and provides a reasonably strong privacy guarantee to at least half of the records. The parameters for the other mechanisms are set to provide comparable utility; the slowly scaling additive mechanisms have the same standard deviation as the Gaussian additive mechanism, and the slowly scaling transformation mechanisms have that same standard deviation when run on a query equal to the median value of each attribute. 

For all mechanisms except the exponential polylog mechanism, all parameters except $\sigma$ are set to sensible defaults or to realize the desired slowly scaling behavior. We then solved for $\sigma$ to obtain the standard deviations described above. For the exponential polylog mechanism, we solved for $d$ to obtain the standard deviations, while the other parameters are set to sensible defaults or to obtain the desired scaling. See Table~\ref{tab:experiment-params} for the mechanism parameters used for each attribute.

Since these parameter settings are based on the medians of the raw, confidential data, they are not feasible in a traditional DP setting, which does not permit access to the raw data. We use these settings here because they can suggest how well the mechanisms might perform with parameters set similarly, based either on private estimates or reasonably accurate prior knowledge of the medians.

\subsection{Metrics}
We compute three metrics to evaluate the performance of the different mechanisms. In Figures~\ref{fig:ecdf-all} and~\ref{fig:cattle-eCDF-additive}, we compute every record's PRzCDP privacy loss from each groupby sum and plot empirical CDFs of these values. For any point $(x,y)$ in a CDF plot, the corresponding value on the $y$-axis represents the proportion of records with privacy losses that are less than or equal to the $x$-value. A plot with a CDF that reaches 1 at a given value of $x$ means that no records have a privacy loss greater than that $x$ value. Thus, CDFs which peak quickly indicate that a large proportion of records have small privacy losses. 

To assess the mechanisms' utility, we compute the empirical absolute relative error (ARE), which, for a given mechanism output $M(D)$ and non-noisy query value $q(D)$, is defined as $$ARE(M(D), q(D)) = \frac{|M(D) - q(D)|}{|q(D)|}.$$ Figures~\ref{fig:synth-are-additive} and~\ref{fig:cattle-are-additive} show boxplots of the ARE for different quintiles of the true data to show how the distribution of the error changes as the size of the record grows.

Finally, in Figure~\ref{fig:smallPL-plot}, we plot the mechanisms' standard deviations as functions of the true query value, $q(D)$.

\subsection{Discussion}
We present our results below, discussing the mechanisms' privacy guarantees in Section~\ref{sec:privCDF} and then their utility in Sections~\ref{sec:ARE} and \ref{sec:ExperimentSDs}.

\subsubsection{Empirical Privacy Loss CDFs} \label{sec:privCDF}
\begin{figure}[t]
    \centering
\includegraphics[width = 0.49\textwidth]{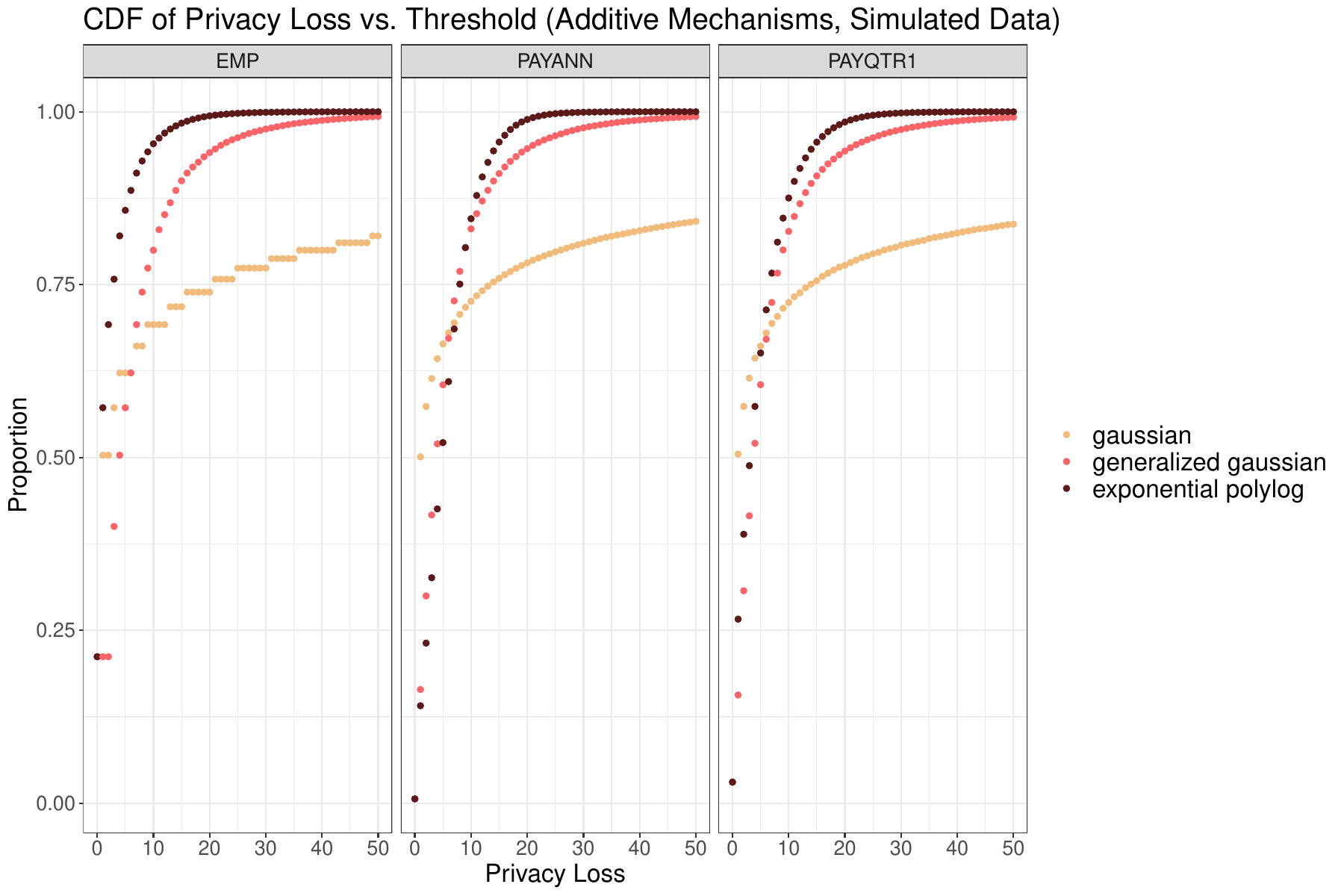}
\includegraphics[width = 0.49\textwidth]{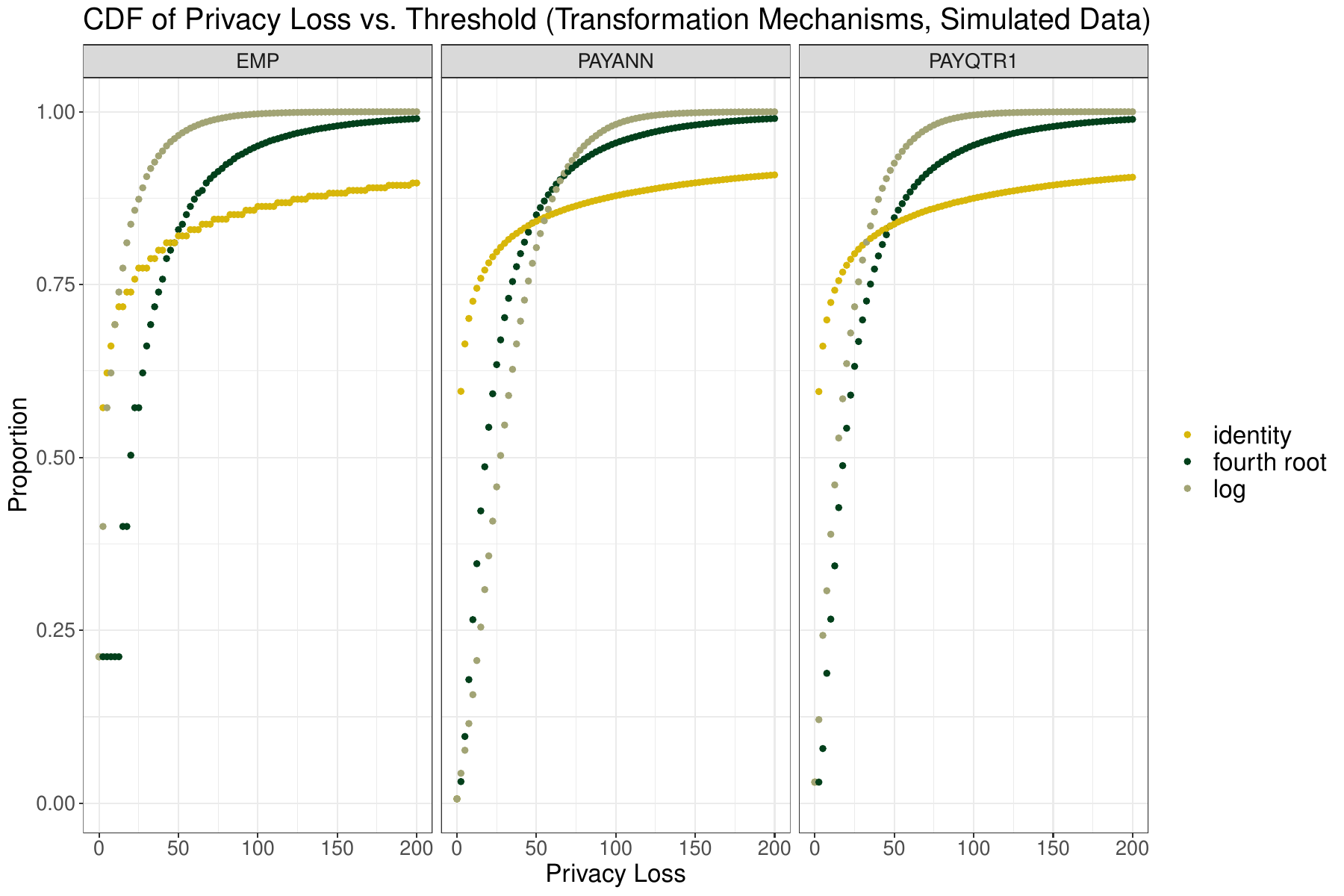}
    \caption{This figure shows empirical CDFs of the privacy loss for each mechanism for the sampled simulated data. For a given value on the $x$-axis, the $y$-axis shows the proportion of records which have at most that level of privacy loss. These CDFs were built by getting the privacy loss for each record and then computing, over a range of privacy loss thresholds, the proportion of privacy losses which were less than or equal to each threshold value. For each attribute, the mechanisms' parameters were calculated such that their standard deviations are equal to $\sqrt{0.5}*\text{median(attribute value)}$ on a query equal to \text{median(attribute value)}.}
    \label{fig:ecdf-all}
\end{figure}

Figure~\ref{fig:ecdf-all} displays the empirical CDFs of the privacy losses of the additive and transformation mechanisms when applied to the sampled simulated CBP data. To build the CDFs, we computed the privacy loss for each record and then determined, over a range of privacy loss thresholds, the proportion of privacy losses which were less than or equal to each threshold value. This provides a finer view on the shape of the CDF than would be attainable using just the data.

Recall that a CDF indicates stronger privacy guarantees when it has a larger $y$ value for any given $x$ value. The left-hand panel of Figure~\ref{fig:ecdf-all} shows that, in the simulated CBP data, both the generalized Gaussian and the exponential polylog mechanisms dominate the Gaussian mechanism for all values of privacy loss except for small values near $0$. The CDFs for the exponential polylog mechanism peak (reach 1) earliest, indicating that, by using this mechanism for these datasets, we can ensure that the maximum privacy loss is lower than the maximum privacy loss for other mechanisms. 

In the right-hand panel of Figure~\ref{fig:ecdf-all}, we see that the transformation mechanisms exhibit similar patterns, but the privacy losses of the slowly scaling transformation mechanisms are generally larger than their counterpart additive mechanisms -- note the larger scale of the $x$-axis in this panel. As a result, the identity mechanism dominates for privacy losses which are lower than about 50 for the PAYANN variable and about 30 for PAYQTR1. After these intersection points, the CDFs of the slowly scaling mechanisms again dominate.

\begin{figure}[t]
    \centering
    \includegraphics[width = 0.49\textwidth]{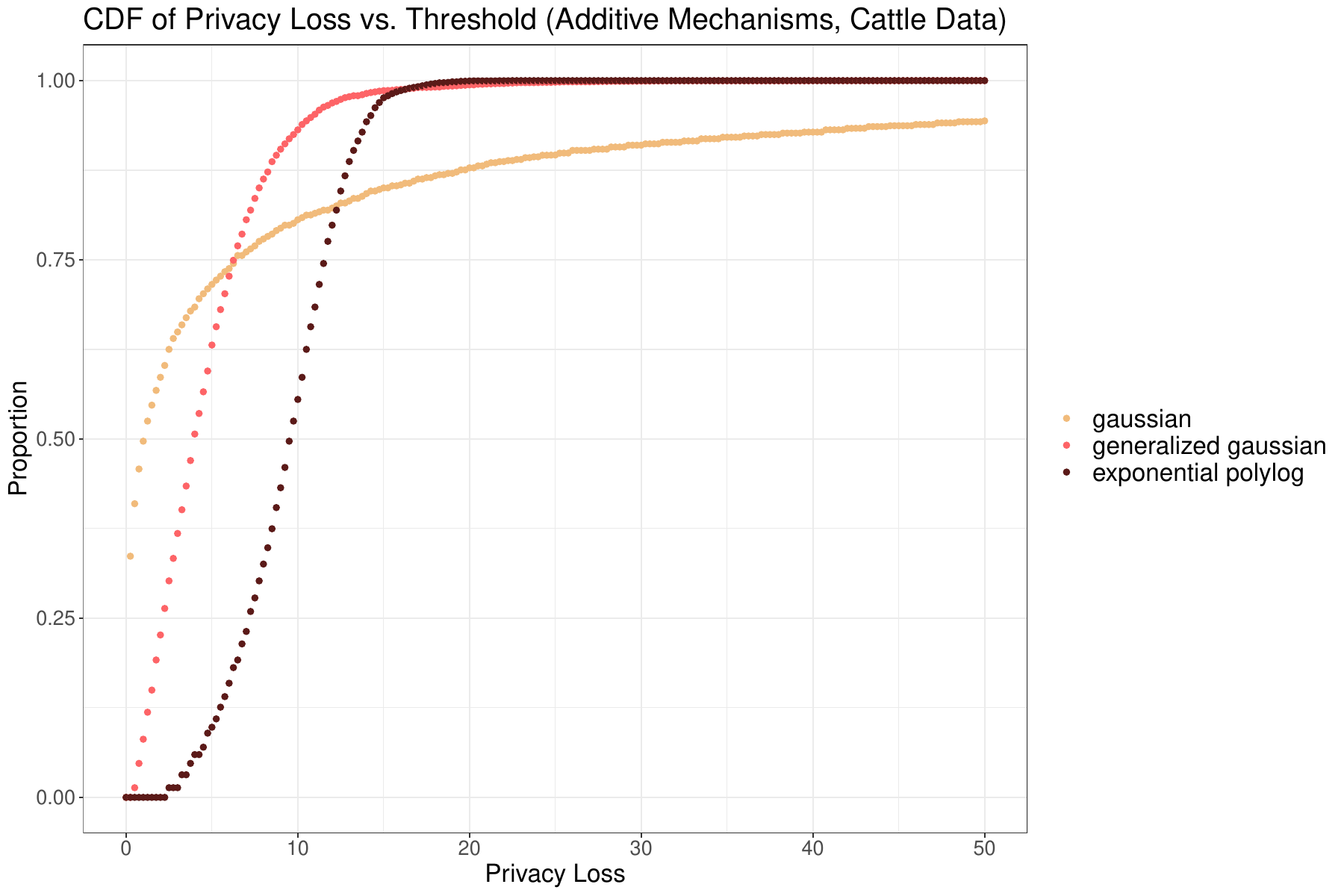}
    \includegraphics[width = 0.49\textwidth]{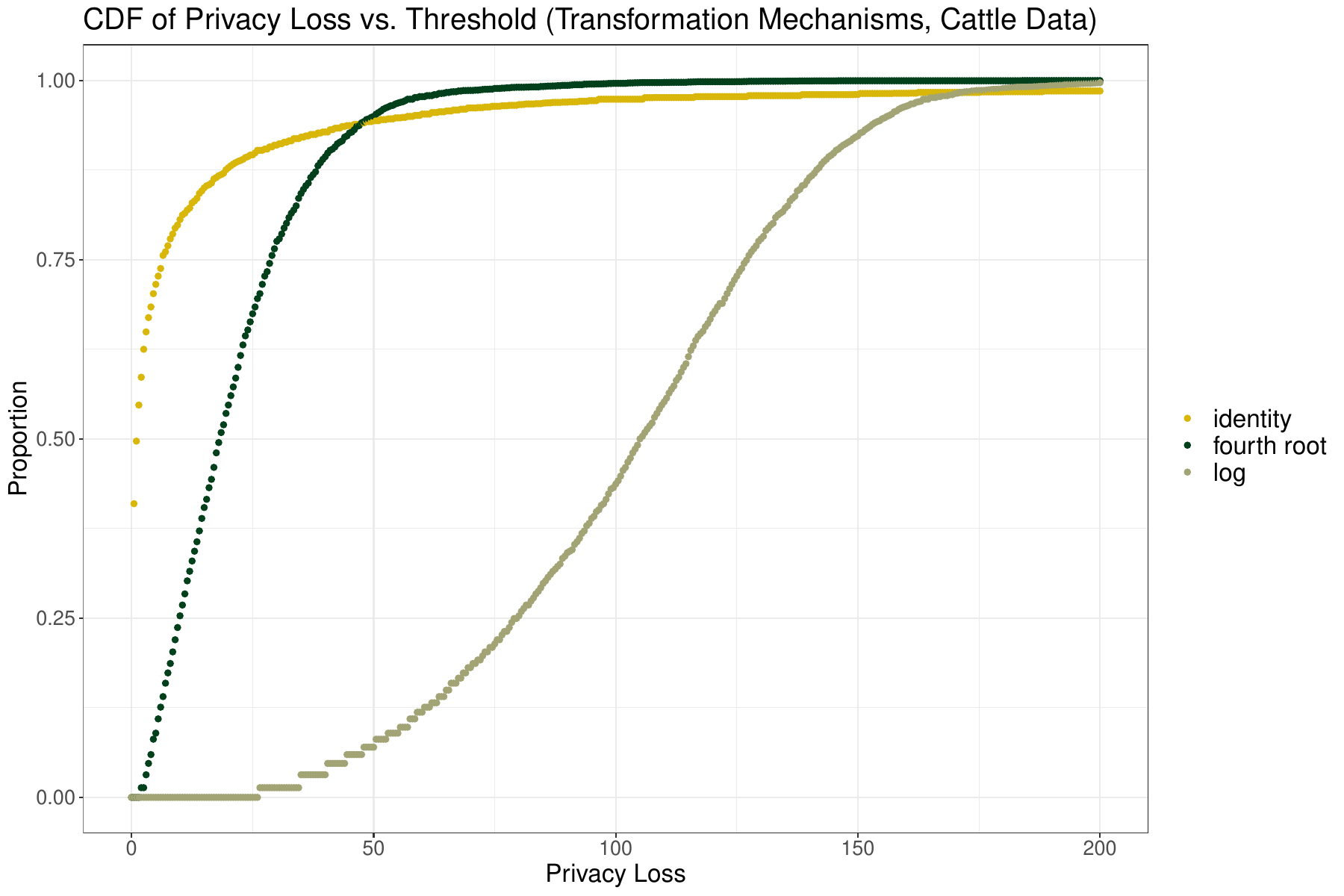}
    \caption{This figure shows an empirical CDF of the privacy loss for each mechanism for the cattle data. For every value on the $x$-axis, the $y$-axis shows the proportion of records which have at most that privacy loss. These CDFs were built by getting the privacy loss for each record and then computing, over a range of privacy loss thresholds, the proportion of privacy losses which were less than or equal to each threshold value. The mechanisms' parameters were calculated such that their standard deviations are equal to $\sqrt{0.5}*\text{median(cattle value)}$ on a query equal to \text{median(cattle value)}.}
    \label{fig:cattle-eCDF-additive}
\end{figure}

Figure~\ref{fig:cattle-eCDF-additive} presents the empirical CDFs of the same mechanisms using the cattle data. For the additive mechanisms, the Gaussian mechanism has the highest CDF until about $x = 6$ before being overtaken by the generalized Gaussian and later the exponential polylog. This follows a similar trend as we observed on the simulated CBP data.

The transformation mechanisms show similar results compared to the simulated CBP data. Here, the identity mechanism CDF is highest until about $x = 50$, when it is overtaken by the fourth root mechanism. The largest privacy losses for the fourth root mechanism and the log mechanism appear to be smaller than those for the identity mechanism, since the identity mechanism CDF does not reach 1 on this graph.

The CDFs for the cattle data differ from the CDFs for the simulated CBP data in that the less slowly scaling mechanisms (the generalized Gaussian and fourth root mechanisms) dominate the more slowly scaling mechanisms (the exponential polylog and log mechanisms) over a relatively wide range of privacy losses. Even when the more slowly scaling mechanisms eventually catch up, they dominate the less slowly scaling mechanisms' CDFs by a very thin margin, the latter already being very close to 1 at this point. On the whole, the privacy guarantees of the less slowly scaling mechanisms look substantially better for the vast majority of records.

\subsubsection{Absolute Relative Error} \label{sec:ARE}
\begin{figure}[t]
    \centering
    \includegraphics[width = 0.49\textwidth]{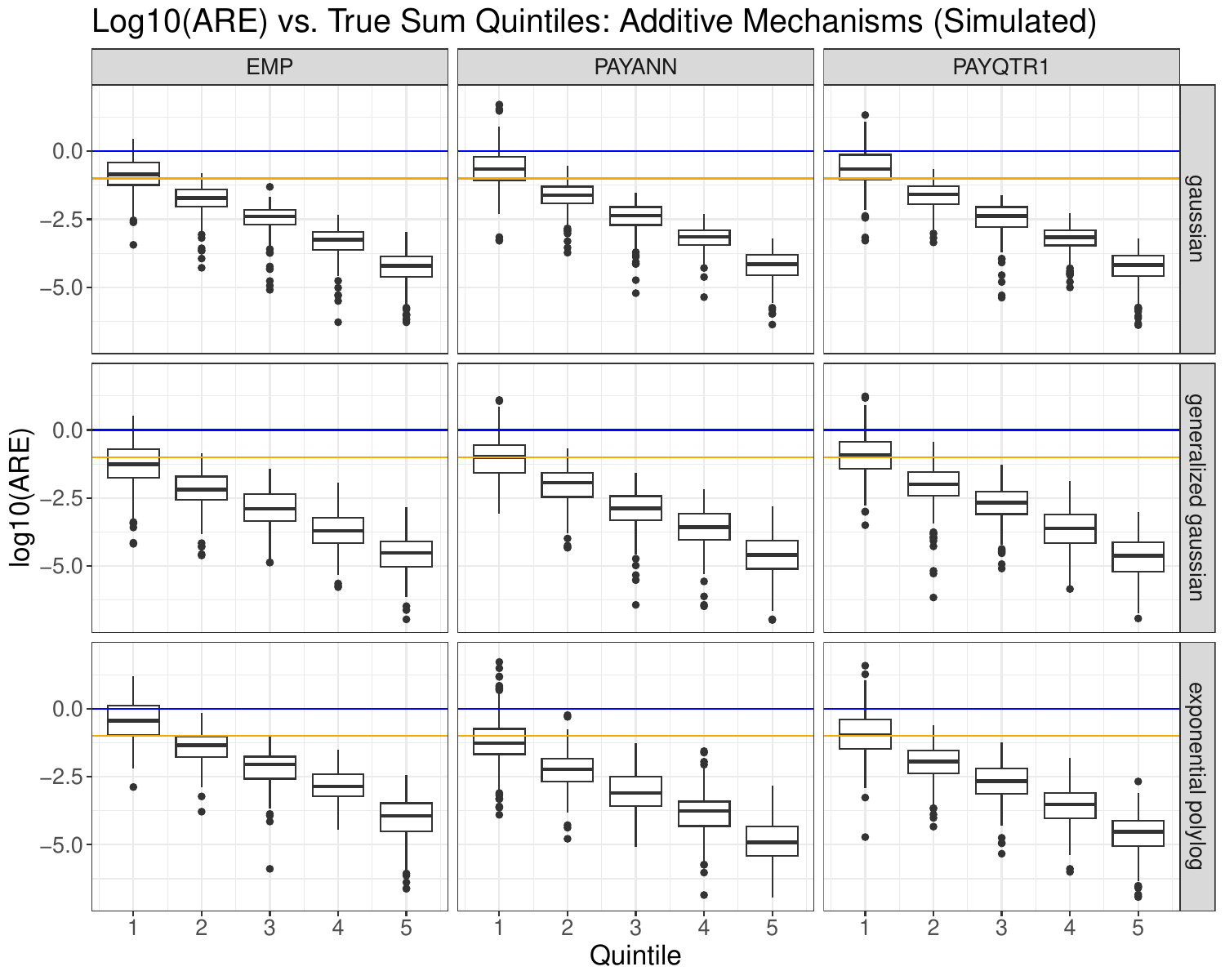}
    \includegraphics[width = 0.49\textwidth]{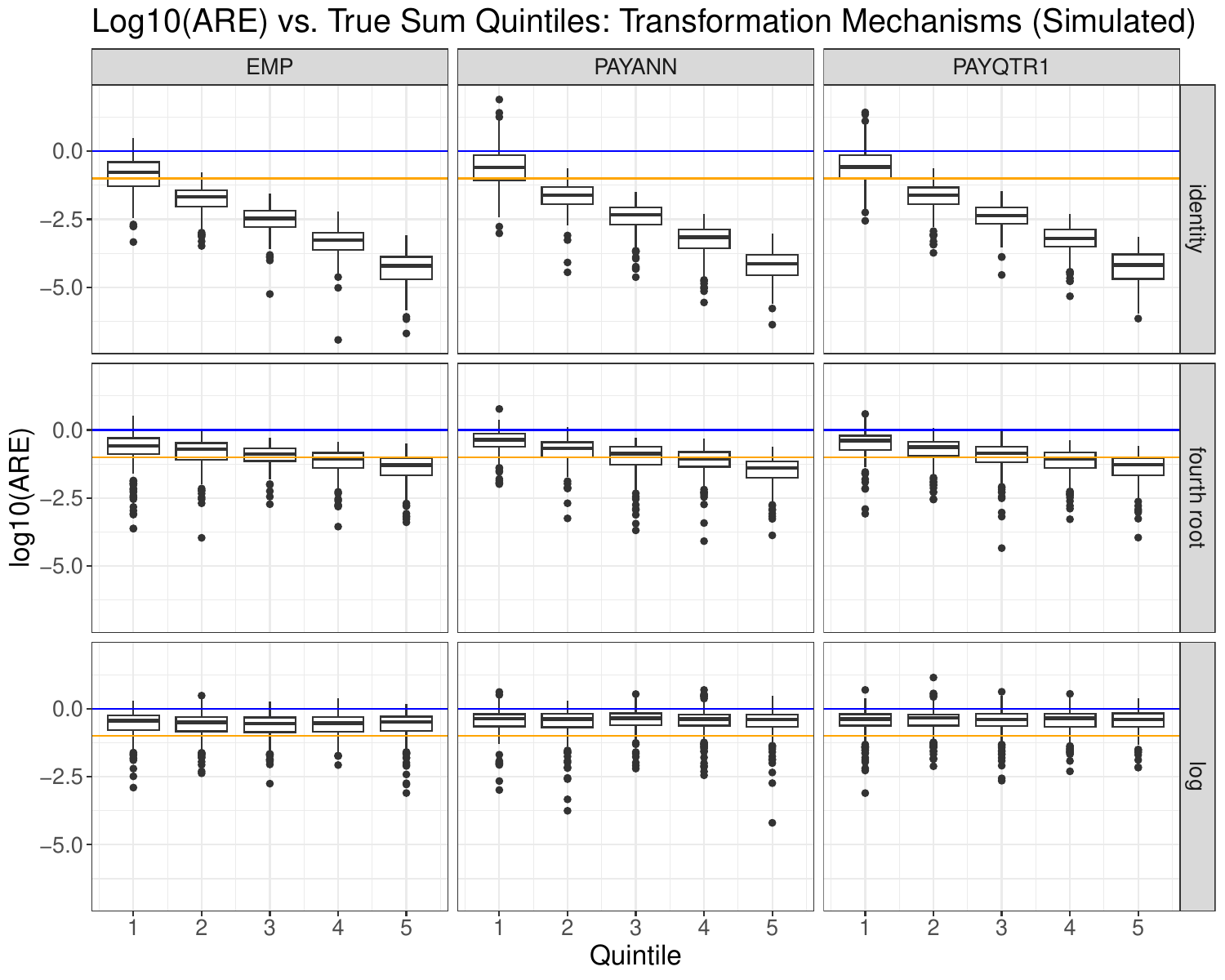}
     \caption{This figure shows the distribution of absolute relative error (ARE) on the $\text{log}_{10}$ scale for each mechanism using each of our three variables in the CBP dataset. The quintile represents where the true sum (without noise added) was in the first 20\%, the second 20\%, etc. The blue line represents where ARE = 100\% and the orange line represents where ARE = 10\%. The horizontal lines in the boxplot represent the 25th, 50th (the bold middle line), and 75th percentiles of error within that quintile. For each attribute, the mechanisms' parameters were calculated such that their standard deviations are equal to $\sqrt{0.5}*\text{median(attribute value)}$ on a query equal to \text{median(attribute value)}.}
     \label{fig:synth-are-additive}
\end{figure}
\begin{figure}[t]
    \centering
    \includegraphics[width = 0.49\textwidth]{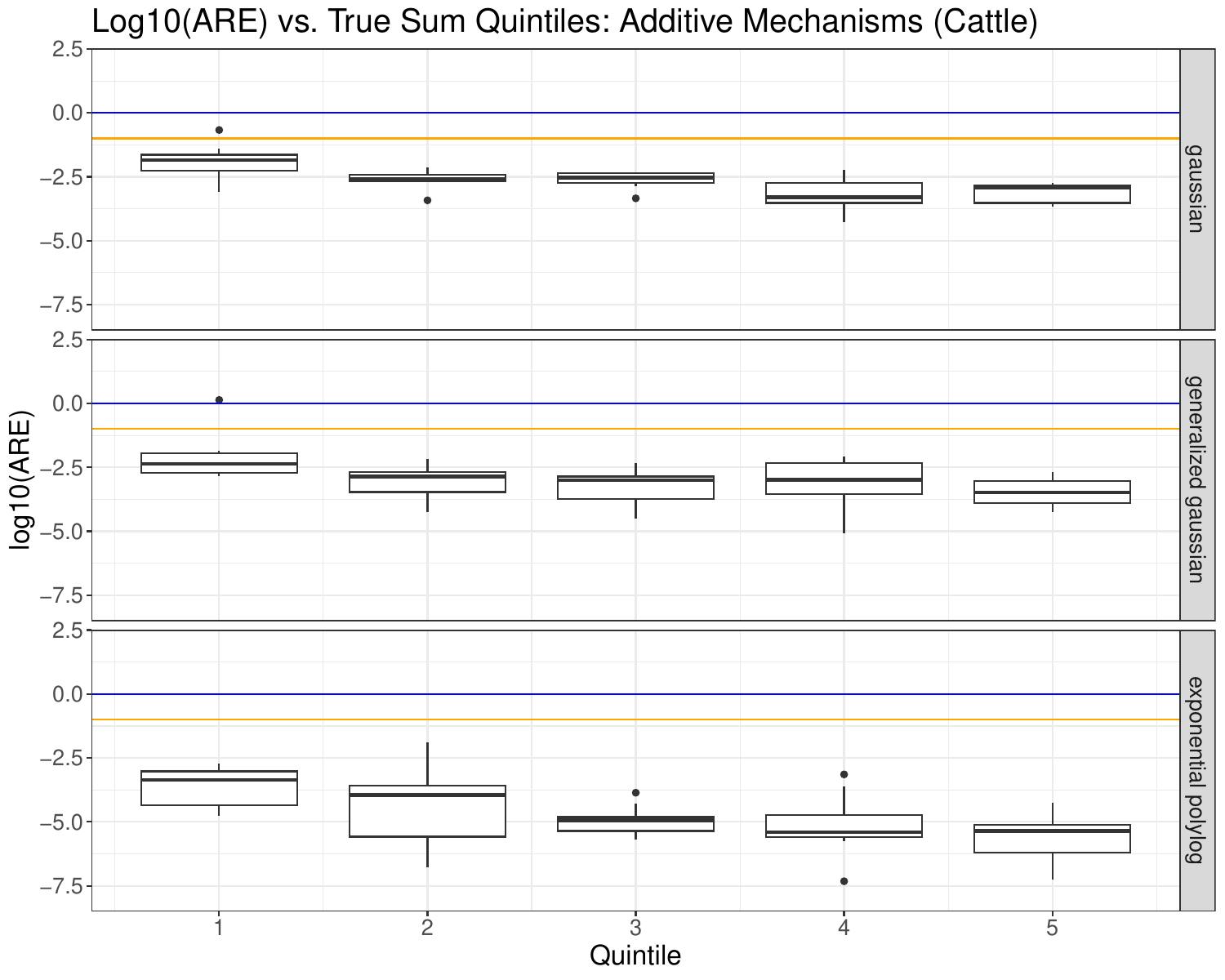}
    \includegraphics[width = 0.49\textwidth]{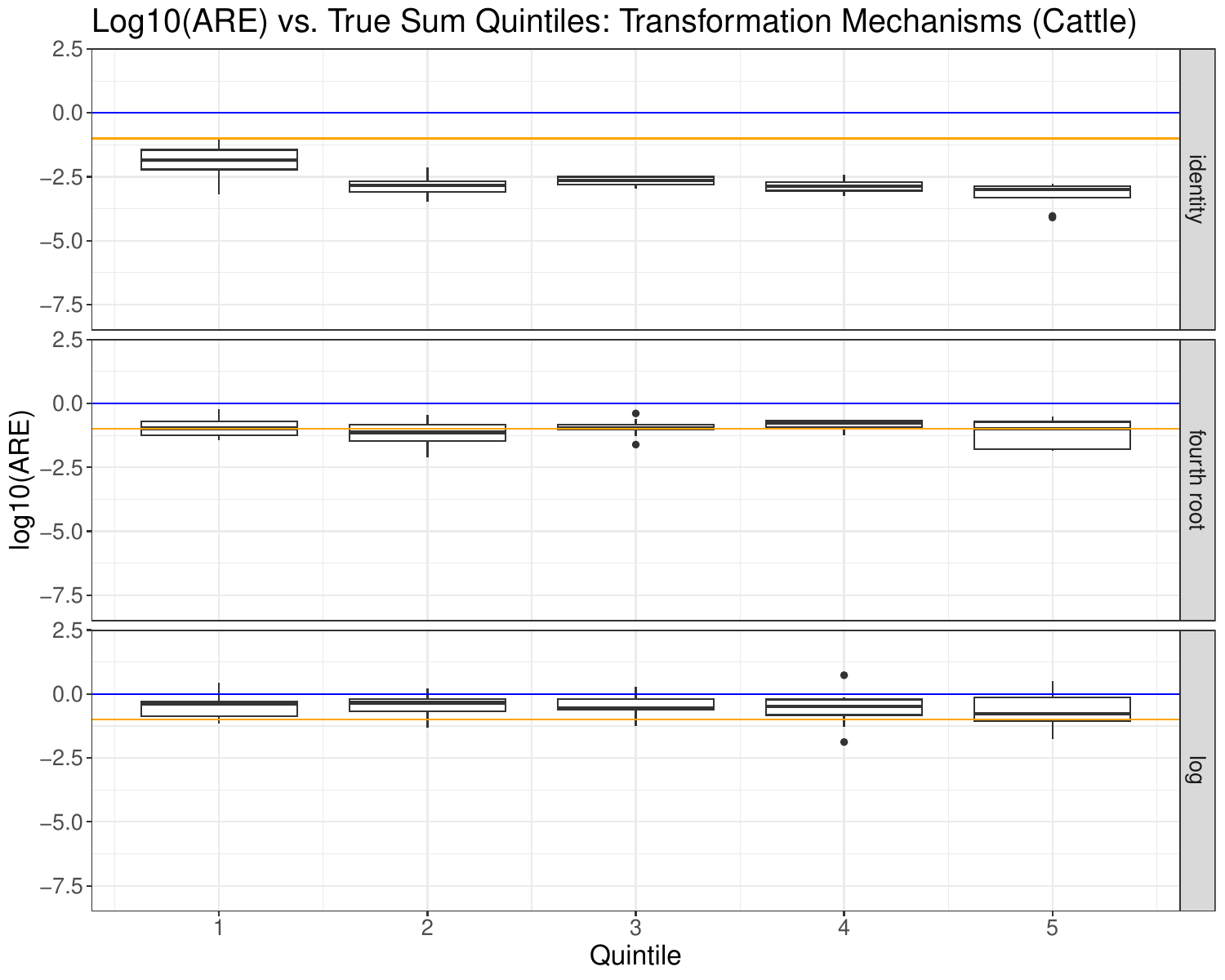}
    \caption{This figure shows the distribution of absolute relative error (ARE) on the $\text{log}_{10}$ scale for each mechanism using the cattle value variable in the cattle dataset. The quintile represents where the true sum (without noise added) was in the first 20\%, the second 20\%, etc. The blue line represents where ARE = 100\%, and the orange line represents where ARE = 10\%. The horizontal lines in the boxplot represent the 25th, 50th (the bold middle line), and 75th percentiles of error within that quintile. The mechanisms' parameters were calculated such that their standard deviations are equal to $\sqrt{0.5}*\text{median(cattle value)}$ on a query equal to \text{median(cattle value)}.}
    \label{fig:cattle-are-additive}
\end{figure}

Figure~\ref{fig:synth-are-additive} displays boxplots of the empirical ARE on the $\text{log}_{10}$ scale for the additive and transformation mechanisms. Separate boxplots are shown for each quintile of the true sum values, as indicated on the $x$-axis. In these plots, the blue line represents 100\% ARE and the orange line represents 10\% ARE. One may consider these as fitness-for-use targets, since 100\% error may be intolerable but 10\% error may be tolerable. 

For the additive mechanisms, whose variances are constant with respect to the true query value, we can see that the ARE decreases markedly as the quintile increases for all mechanisms and all attributes. In these plots, the additive mechanisms all have comparable error and exhibit similar patterns as the quintile increases. The generalized Gaussian and exponential polylog mechanisms typically have lower median ARE than the Gaussian mechanism, but with a wider spread.

Because the identity transformation mechanism is the same as the Gaussian additive mechanism, its ARE boxplots follow the same pattern, declining rapidly with the quintile on the $x$-axis. The slowly scaling transformation mechanisms, however, have variances that increase with the true query value, offsetting this effect. The log mechanism’s standard deviation asymptotically grows linearly in the query value (see Figure~\ref{sec:ExperimentSDs} for an illustration), leading the expected relative error to bottom out at a nonzero value. In Figure~\ref{fig:synth-are-additive}, this leads the log mechanism’s empirical ARE distributions to have 25th percentiles that never fall below 10\%. In this application, this would likely rule out the log mechanism as a viable choice. The fourth root mechanism fares only somewhat better, having 75th percentiles below 10\% only for the largest quintile.

Figure~\ref{fig:cattle-are-additive} presents the AREs for the cattle data. For the additive mechanisms, we still observe that ARE decreases as the quintile increases as we did with the simulated CBP dataset. Additionally, we note that the exponential polylog generally has the lowest error, followed closely by the generalized Gaussian and then the Gaussian mechanisms.

We also see that the ARE is generally lower for this dataset compared to the simulated dataset; each quintile's range of ARE values lies almost entirely below the orange 10\% line. This is a result of the generally larger query values in the cattle dataset, which follows from its groupby sum separating the data into fewer groups. This demonstrates that the fitness for use of the additive mechanisms is heavily dependent on the query values, much like other traditional differentially private mechanisms.

The slowly scaling transformation mechanisms again suffer from their variances growing with the query value. In every quintile, the fourth root and log mechanisms have median AREs above or near the 10\% line. The failure of the fourth root mechanism to improve much with the quintile appears to be due to the distribution of the cattle value queries being more closely clustered than those in the simulated CBP dataset. The result is a reduced difference in relative error between queries.

\subsubsection{Standard Deviation} \label{sec:ExperimentSDs}

\begin{figure}[t]
    \centering
    \includegraphics[width = 0.49\textwidth]{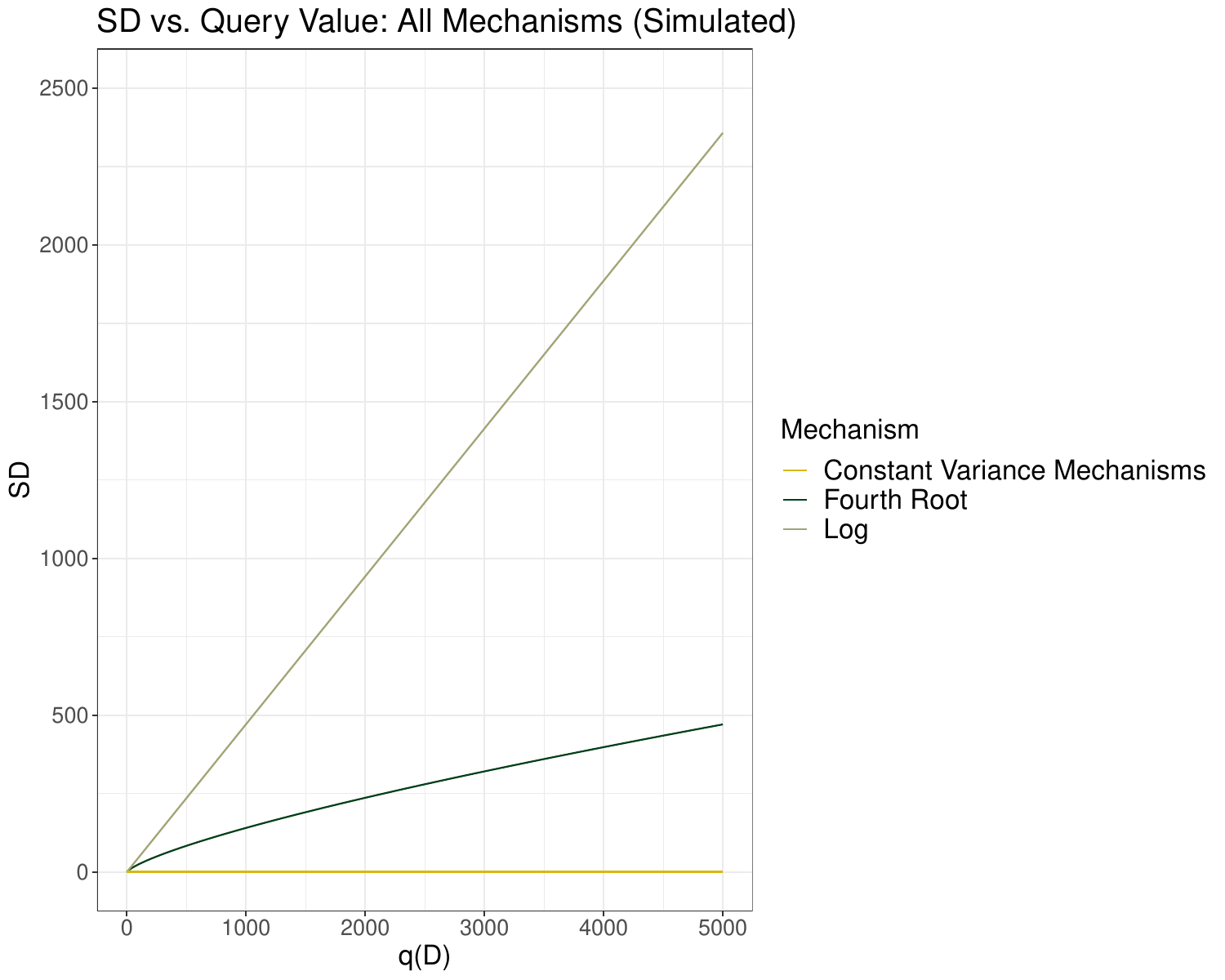}
    \includegraphics[width = 0.49\textwidth]{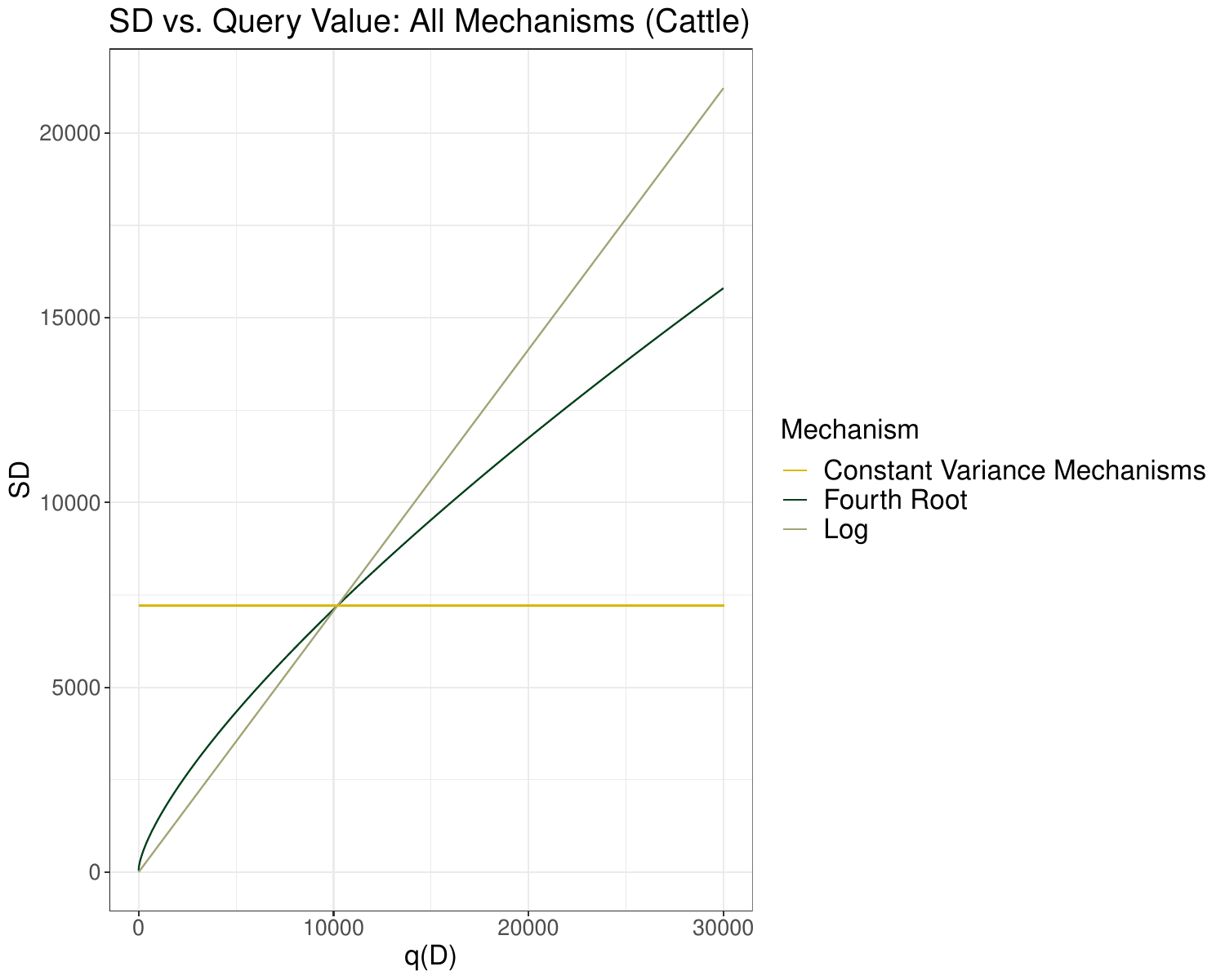}
    \caption{These plots display how the standard deviation for each mechanism changes as a function of the query value $q(D)$. See Table~\ref{tab:experiment-params} for the mechanisms' variance functions. See Table~\ref{tab:experiment-params} for the mechanisms' parameters. On the left, we use the parameters for the EMP attribute, and on the right, we use the parameters for the cattle value attribute. Every additive mechanism has constant standard deviation regardless of query value, as does the identity transformation mechanism. These mechanisms' standard deviations are all reflected in the dashed ``constant variance mechanisms'' line.}
    \label{fig:smallPL-plot}
\end{figure}

In Figure~\ref{fig:smallPL-plot}, we plot the standard deviations for our selected mechanisms as $q(D)$ grows, holding the privacy loss constant. We compute the variances using the formulas in Table~\ref{tab:experiment-params} but plot the standard deviations to ease interpretation. In the left-hand panel in Figure~\ref{fig:smallPL-plot}, we use the parameters used for the EMP attribute, and in the right-hand panel, we use the cattle value parameters (see Table~\ref{tab:experiment-params}). In these plots, the horizontal lines labeled ``Constant Variance Mechanisms" plot the variances of the additive mechanisms and the identity transformation mechanism, since their variances are equal and do not depend on $q(D).$ 

Because the mechanisms are calibrated to have the same variance for a query equal to the median record, and because the slowly scaling transformation mechanisms have variances that grow in the query value, the slowly scaling transformation mechanisms at these settings will have lower variances than the additive mechanisms for queries smaller than the median record. In Figure~\ref{fig:smallPL-plot}, we see that this leaves a fairly large space of query values (all values below 10200) for the cattle data where the transformation mechanisms have lower variances. For the EMP data, by contrast, this space only ranges from 0 to 2 and is invisibly small in the graph. It is clear from the graphs that the transformation mechanisms can have far worse utility than the additive mechanisms when the query value is large.

While the results of our experiments make the additive mechanisms look decidedly better for the datasets we use, the slowly scaling transformation mechanisms could well have lower variances and better ARE distributions than the additive mechanisms if one were to apply these mechanisms to very sparse data – data with lots of queries summing over a single record or even no records. With this improvement in utility, the apparently worse privacy guarantees of the slowly scaling transformation mechanisms in Figures~\ref{fig:ecdf-all} and \ref{fig:cattle-eCDF-additive} could also be remedied by increasing the mechanisms' variances. In Appendix~\ref{sec:analytical}, we discuss some conditions under which the transformation mechanism has a lower variance than the additive mechanism for a comparable privacy guarantee.

\section{Conclusion}
\label{sec:conclusion}
In this paper, we have developed two classes of slowly scaling PRzCDP mechanisms -- mechanisms whose policy functions scale sub-linearly in a record-specific measure of influence. We first demonstrated that unit splitting -- the previous method for satisfying PRzCDP -- results in unacceptably high privacy losses when individual records can be very large. This was because the policy function incurred by unit splitting scales in the square of the record's magnitude. \par 
The first class of slowly scaling mechanisms, the transformation mechanisms, add Gaussian noise to a concave transformation of the query. This results in a policy function that scales in the square of the chosen transformation function. For example, the square root transformation results in a linear policy function and the log transformation results in a policy function that scales in the square of the logarithm. \par 

The second class of mechanisms, the additive mechanisms, draw noise from specially chosen distributions with fat tails. These mechanisms can achieve a wide range of policy functions, including linear and polylogarithmic policy functions like those achieved by the transformation mechanisms. \par 

Last, we conducted experiments using our mechanisms with two different datasets: a sample of a simulated version of the County Business Patterns dataset and a dataset of cattle inventory from the U.S. Department of Agriculture. Both of these experiments showed that the empirical CDFs of our slowly scaling mechanisms generally peak earlier than a standard (i.e., Gaussian) mechanism, indicating that the maximum privacy loss (for the largest records) incurred by these mechanisms is lower.

Our results generally indicate that slowly scaling mechanisms can sharply reduce the privacy losses incurred by very influential records. With this check on extreme privacy losses, it becomes tenable to publish queries without clipping influential records. This, in turn, can reduce or eliminate bias. Future work could include the derivation of additional slowly scaling mechanisms; the derivation of tighter PRzCDP policy functions for the additive mechanisms, since our current results are based on the mechanisms' PRDP policy functions and may be loose; and the development of frequentist hypothesis testing semantics~\citep{kifer2022bayesian} for PRzCDP and PRDP.

\section*{Acknowledgment}
We would like to acknowledge the following individuals for their helpful discussion, feedback, and contributions toward this paper: Paul Bartholomew, Christina Bassis, Noah Fine, Monica Mendoza, Theresa Nguyen, Sayi Sathyavanan, and Rich Stevenson from The MITRE Corporation; Jordan Awan from Purdue University; Skye Berghel, Bayard Carlson, Casey Meehan, and Ruchit Shrestha from Tumult Labs; and Margaret Beckom, William C. Davie Jr., Philip Leclerc, and Rolando Rodriguez from the U.S. Census Bureau. We thank Ian Schmutte for suggesting a mechanism that directly inspired the transformation mechanisms developed here.

\bibliography{sample}
\bibliographystyle{abbrvnat}
\appendix

\section{Analytical Comparison of Additive and Transformation Mechanisms}
\label{sec:analytical}

When choosing a mechanism for a PRzCDP release, there are multiple tradeoffs. The choice of policy function not only decides how the privacy loss scales across records, but also the amount of error required to achieve specific privacy losses. Meanwhile, the class of mechanism (either transformation or additive) determines whether error will scale as a function of query value. In this section, we discuss these tradeoffs and how specific choices of policy function and mechanism class influence the utility of a data release. A summary of the variances and policy functions of selected mechanisms can be found in Table~\ref{tab:mechSummary}.

\subsection{Policy function}
\begin{figure}[t]
    \centering
\includegraphics[width = 0.8\textwidth]{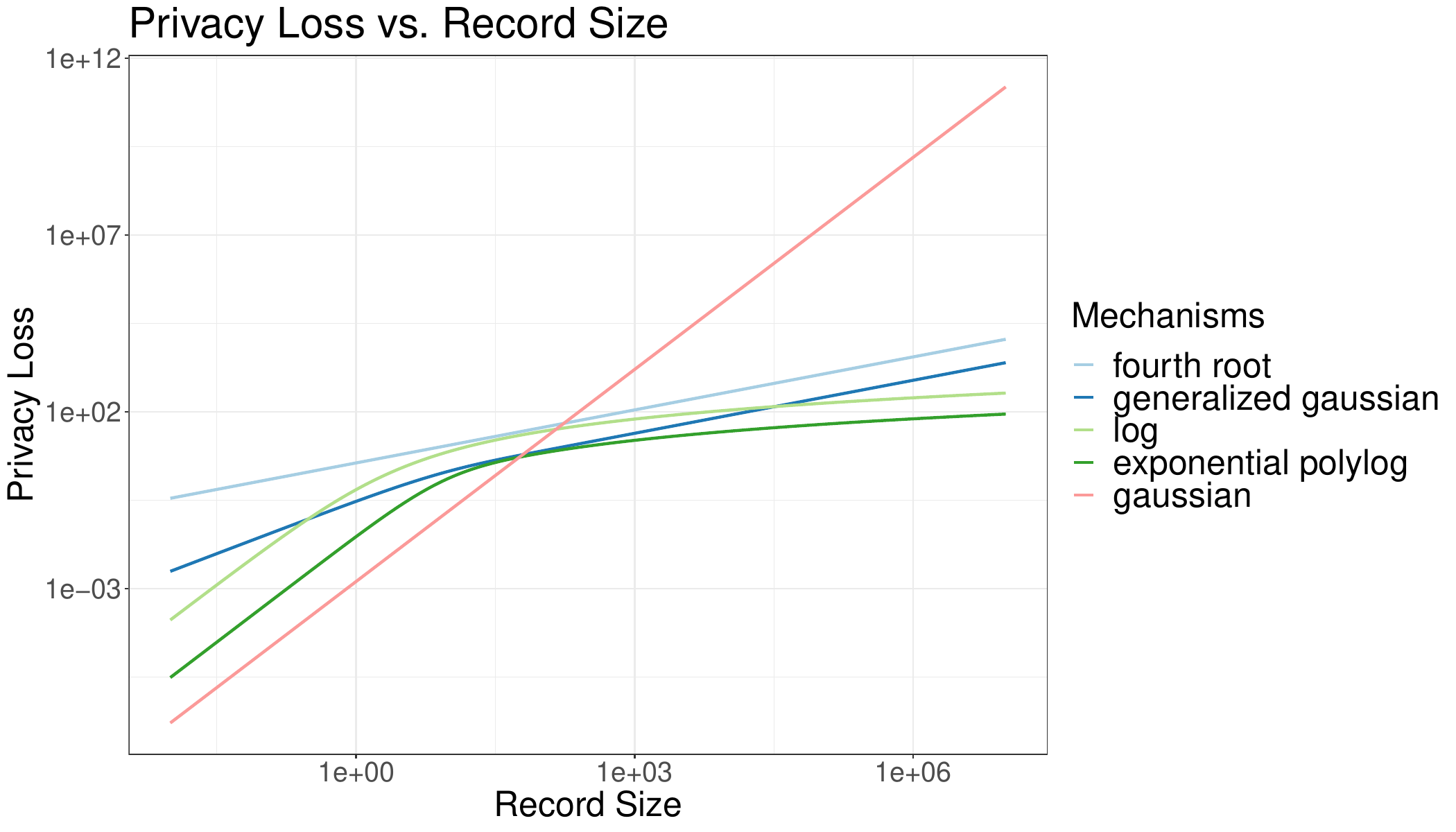}
    \caption{This figure shows the PRzCDP policy functions for several mechanisms over a large domain. The parameters are set for PAYQTR1 as indicated in Table \ref{tab:experiment-params}. Both the $x$ and $y$ axes are on the $\text{log}_{10}$ scale.}
    \label{fig:pl-all-mechs}
\end{figure}

The choice of policy function, for a given mechanism class, informs how privacy loss will be distributed across all the records. While more slowly scaling policy functions ensure that larger records incur less privacy loss, they may also result in smaller records incurring larger privacy losses. We visualize this in Figure~\ref{fig:pl-all-mechs}. For this reason, a data curator must use their knowledge of the data distribution (or private exploratory data analysis) in order to select a policy function that results in an appropriate distribution of privacy loss for the given context. \par 

\begin{exa}
Consider a sum being computed by different transformation mechanisms, whose parameters are set such that they have a standard deviation of 1,000 for a sum of size 10,000. That is, the square root transformation has parameters $a = 0$ and  $\sigma \approx 5$, resulting in the PRzCDP policy function $\frac{\Delta(r)}{50} $. The fourth root transformation has parameters  $a = 0$ and $\sigma \approx 0.2$, resulting in the policy function  $\frac{\sqrt{\Delta(r)}}{0.08} $. Finally, the log transformation has parameters $a = 1$ and $ \sigma \approx 0.1$, resulting in the policy function  $\frac{\log(\Delta(r)+1)}{0.02} $.
The square root transformation function incurs the lowest privacy losses for records whose value being summed over lies in $[0,390625]$. Then the fourth root transformation incurs the lowest privacy loss for records those value being summed over lies in $[390625 , 463584]$. Any records with values larger than $463584$ incur the lowest privacy loss when the log transformation mechanism is used.

\end{exa}

This demonstrates that small record values are best protected by faster scaling policy functions, such as that obtained with the square root transform. The fourth root transformation is preferable for medium-sized records, and the log transformation results in the lowest privacy loss for large records. As such, faster scaling policy functions may be preferable when the data's tails are thinner (i.e., when outliers are few and relatively small). We conduct a similar analysis on the additive mechanisms below.

\begin{exa}
Consider a sum being computed by the class of additive mechanisms, whose parameters are set such that they have a standard deviation of 1,000. That is, the generalized Gaussian with $p=1$ has parameter $\sigma \approx 707$, resulting in the PRDP policy function $\frac{\Delta(r)}{707} $. The generalized Gaussian with $p=1/2$ has parameter $\sigma \approx 91$, resulting in the PRDP policy function  $\sqrt{\frac{\Delta(r)}{91}} $. Finally, the exponential polylog has parameters $p=2, a = e, \sigma = 1,$ and $d \approx 0.145$, resulting in the policy function  $0.145( \ln(\Delta(r) + e)^2 - 1)$.

The generalized Gaussian with $p=1$ incurs the lowest privacy losses for records whose values lie in approximately $[0 , 5492]$. Then the generalized Gaussian with $p=1/2$ incurs the lowest privacy loss for records with values in approximately $[5492 , 17387]$. Any records with values larger than $17387$ incur the lowest privacy loss when the exponential polylog is used.
\end{exa}

Much like the transformation mechanism, the faster scaling policy functions incur less privacy loss for the smaller records, despite their larger privacy loss for larger records. This means that regardless of mechanism class, the data curator must decide which policy function results in an appropriate distribution of privacy loss across records. Depending on the context, that may mean choosing a faster scaling policy function that results in low privacy loss for the majority of the records, while few records incur high privacy loss. In other contexts, that may mean choosing a slowly scaling policy function which may incur larger privacy loss for smaller records but ensure all records incur reasonable privacy loss.

\subsection{Mechanism Class}
\label{sec:mechanismclass}
The transformation and additive mechanisms come with different advantages and disadvantages. A principal strength of the transformation mechanisms are their ease of implementation and development. Software implementations of privacy mechanisms often require that special floating-point-safe sampling routines be used to draw any random noise; naively implemented pseudorandom number generators that suffice for other purposes often draw from an approximate distribution in a way that can catastrophically degrade any privacy guarantees \citep{haney2022precisionbased}. The transformation mechanism can be implemented as long as Gaussian noise can be safely drawn, which existing software platforms already support (see, for example, \cite{berghel2022tumult} and \cite{OpenDP2020}). This extends to transformation functions beyond those described in this work. A new transformation mechanism can easily be developed for any suitable concave, strictly increasing function, allowing for precise targeting of specific policy functions. 

On the other hand, the additive mechanisms require the development of specific distributions for each policy function. For each of these distributions, a new floating-point safe sampling method must be developed in order to be used safely.  This may limit the range of policy functions one can practically attain with the additive mechanism. For the additive mechanisms we implement in this paper, their quantile functions depend on standard functions, which makes implementation of samplers for these distributions straightforward with the use of an arbitrary precision library. While this is true for the exponential polylogarithmic distribution when $p \in \{1,2\}$, for any other $p$ the distribution function becomes intractable and quantile functions are not easily derived.

The additive noise mechanisms have the advantage in terms of bias. Being symmetrically distributed about $q(D)$, these are unbiased for the mean, median, and mode. Transformation mechanisms exist that are either mean- or median-unbiased, but typically are not both (see Appendix~\ref{sec:transMechMedianUnbiased} for median-unbiased transformation mechanisms).

A notable drawback of the transformation mechanisms is that the mechanisms' variances typically grow with the query value, $q(D)$. By contrast, the variances of the additive mechanisms are constant. This reduces the transparency of the transformation mechanisms, as data users cannot be told exactly what variance the added noise has. Users can still estimate these variances, however. For example, a user could trim the mechanism's output to lie within the query's domain and plug the result into the variance formulas in Table~\ref{tab:unbiasedestimators} in place of $q(D)$.

This dependence of the transformation mechanisms' variance on $q(D)$ also immediately implies that the additive mechanism can always offer a lower variance for a given privacy guarantee if $q(D)$ is large enough. In the case of the transformation mechanism with the log transform, the mechanism's standard deviation grows at the $O(q(D))$ rate, and the ratio of $q(D)$ to this standard deviation limits to a positive constant, as shown in Figure~\ref{fig:se_rse}. 

In this sense, then, the noise added by this mechanism is non-negligible, even for very large $q(D)$. For the additive mechanism, by contrast, the amount of noise can become arbitrarily small, relative to an arbitrarily large $q(D)$.

The transformation mechanism, however, may offer lower variance for small query values. To illustrate this, we take a mechanism of each type that asymptotically achieves the same policy function and compare their variances in Examples~\ref{ex:variance_comarison_sqrt} and \ref{ex:variance_comparsion_log}, as well as Figure~\ref{fig:variance_comparison_sqrt}.

\begin{figure}[t]
    \centering
    \includegraphics[width = 1.0\textwidth]{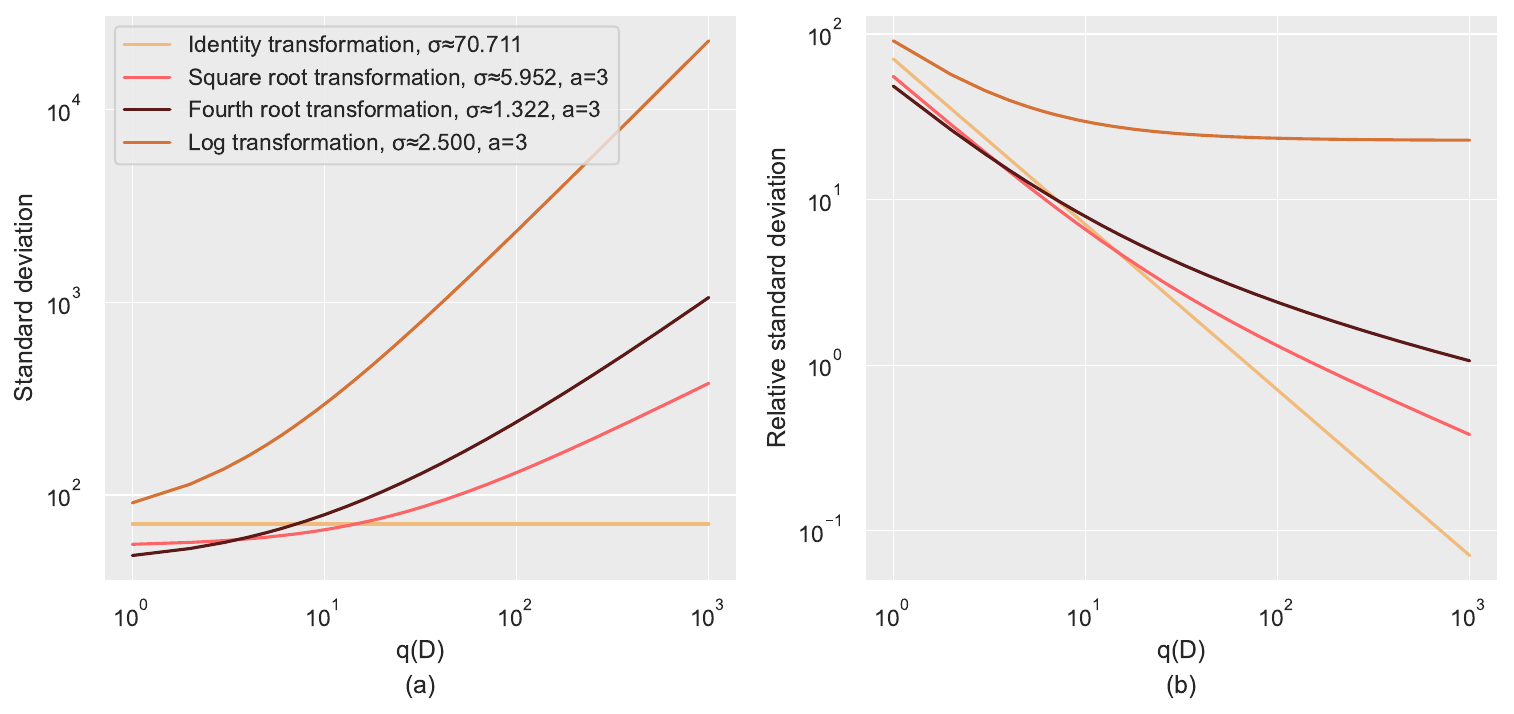}
    \caption{(a) Standard deviation and (b) relative standard deviation of noise added from mean-unbiased transformation mechanisms for sum queries, with transformation functions specified in the legend. Relative standard deviation is defined as the standard deviation divided by $q(D)$. These mechanisms are tuned to achieve PRzCDP privacy loss of 1 when $\Delta(r)=100$.} 
    \label{fig:se_rse}
\end{figure}

\begin{figure}[t]
    \centering
\includegraphics[width = 0.8\textwidth]{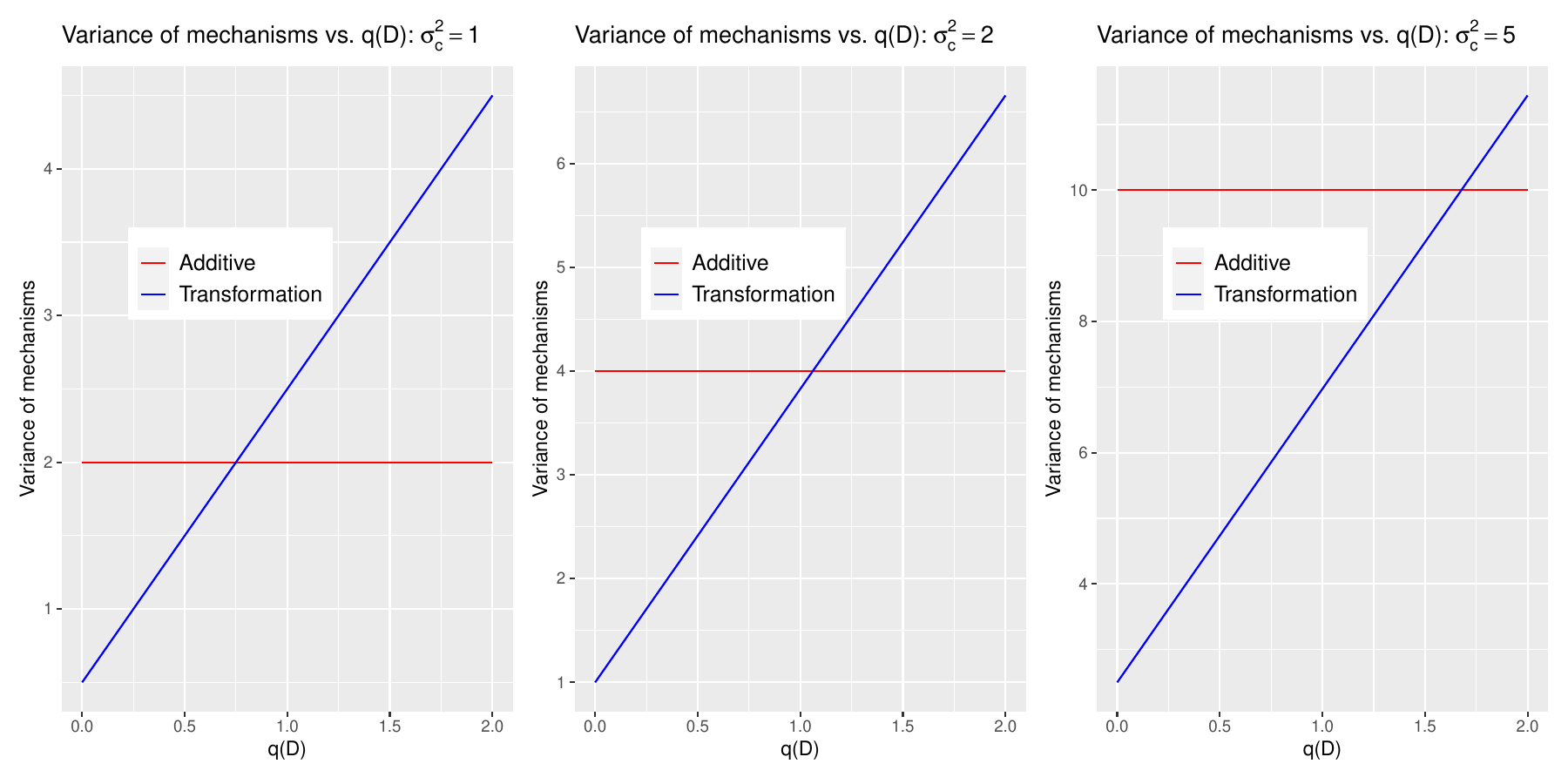}
    \caption{Variances of the generalized Gaussian mechanism with $p=1, \sigma = \sigma_c$ and the square root transformation mechanism with $a=0, \sigma = \sqrt{.5\sigma_c}$. These mechanisms have the same asymptotic policy functions (see Example~\ref{ex:variance_comarison_sqrt}). For all values of $\sigma_c$, the transformation mechanism has lower variance for small query values, while the additive mechanism has lower variance for large query values.}    \label{fig:variance_comparison_sqrt}
\end{figure}

\begin{exa} \label{ex:variance_comarison_sqrt}
Consider the transformation mechanism with transformation $f(x) = \sqrt{x}$, $a=0$, and $\sigma=\sqrt{.5\sigma_c}$. Likewise, consider the additive mechanism using the generalized Gaussian distribution with $p=1$ and $\sigma=\sigma_c$ (i.e., the Laplace mechanism). Per Theorem~\ref{thm:transform}, the transformation mechanism has the policy function $P(r)=\frac{\Delta(r)}{\sigma_c}$. By Corollary~\ref{cor:genGaussMech}, the additive mechanism has the policy function $P'(r)=\tanh(\frac{\Delta(r)}{2\sigma_c})\frac{\Delta(r)}{\sigma_c}$. The difference between these policy functions approaches zero rapidly as $\Delta(r)$ grows, so we say that these mechanisms have the same asymptotic policy functions.

The transformation mechanism has variance $2\sqrt{.5\sigma_c}^4 + 4\sqrt{.5\sigma_c}^2q(D) = .5\sigma_c^2 + 2\sigma_c q(D)$ and the additive mechanism has variance $\sigma_c^2\frac{\Gamma(3)}{\Gamma(1)} = 2\sigma_c^2.$ The transformation mechanism therefore has larger variance than the additive mechanism when $2\sigma_c^2 \leq .5\sigma_c^2 + 2q(D)\sigma_c$, or $.75\sigma_c \leq q(D)$. This tradeoff is visualized in Figure~\ref{fig:variance_comparison_sqrt}. Note that the region where the transformation mechanism dominates the additive mechanism is larger when the privacy guarantee is stronger (i.e., when $\sigma_c$ is larger).
\end{exa}

Likewise, we compare the log transformation mechanism with the additive mechanism using the exponential polylogarithmic distribution when $p = 2$. 

\begin{exa} \label{ex:variance_comparsion_log}
To compare the log transformation mechanism and the additive mechanism with the exponential polylogarithmic distribution we set parameters such that the mechanisms' policy functions nowhere exceed $\ln (\Delta(r)+a)^{2}/2\sigma^2$. In order to achieve this, we set $d=1/2\sigma^{2}$ for the exponential polylog distribution and compare to the log transformation mechanism with noise parameter $\sigma$. The offset parameter $a \geq e$ is the same for both mechanisms. Given this setting, the log-transformation mechanism has variance $(e^{2d}-1)(q(D)+a)^{2}$.

Let $\sigma^{2}_{g}$ denote the variance of the exponential polylogarithmic distribution from Table \ref{tab:Foxdenscases}. Then the transformation mechanism will have larger variance when 
$\sqrt{\frac{\sigma^{2}_{g}}{e^{2d}-1}}-a\leq q(D)$. 
\end{exa}

With similar reasoning as above, the additive mechanism will dominate when the true query value is large. This means that the choice of mechanism class should be decided based on prior knowledge of the data release. If the data curator knows that query values are likely to be small, the transformation mechanisms may incur less error. However, if query values are likely to be large, then additive mechanisms will result in lower overall error.

\section{Transformation Mechanisms with Median Unbiased Estimators} \label{sec:transMechMedianUnbiased}

In Section~\ref{subsubsec:Estimators}, we developed mean-unbiased estimators for use in transformation mechanisms. Another class of estimators sets $g(x) = f^{-1}(x) - a$ for $x$ where $f^{-1}$ exists - that is, for $x \in [f(a), \infty)$, provided that $\lim_{x \to \infty} f(x) = \infty$. The intuitive appeal is clear; if no noise were added to $f(q(D))$ in Line~\ref{line:transform_noise} of Algorithm~\ref{alg:transform}, then this would return $q(D)$ exactly. However, while $f(q(D) + a)$ lies in $[0,\infty)$, the noisy transformed query value $\tilde{v}$ may take any real value, so care must be taken to define $g(x)$ for $x \in (-\infty, f(a))$, as well. Different estimators in this class take different values in this region. 

One option that may be intuitively appealing is to set $g$ using some "canonical" inverse of $f$ that is defined over all of $\R$. For example, if $f(x) = \sqrt[k]{x}$ for a positive integer $k$, one could use $g(x) = x^k - a$. These estimators do not necessarily have any unbiasedness properties, but other estimators in this class are median-unbiased and can be readily obtained for any $f$; when $g$ is defined to always be non-positive where $f^{-1}$ does not exist, the estimator is median-unbiased. Among these median-unbiased estimators, the risk-minimizing estimator simply sets $g(x) = 0$ in this region.

\begin{thm}
    \label{thm:mdnUnbiasTransform}
    Let Assumption~\ref{as:transformInput} hold. Suppose further that $\lim_{t \to \infty} f(t) = \infty$, $g(x) = f^{-1}(x) - a$ for $x \in [f(a),\infty)$, and $g(x) \leq 0$ for $x < f(a)$. Denote Algorithm~\ref{alg:transform}\allowbreak$(q(D), a, \sigma, f, g)$ by $M(D)$ and the median of a random variable, $X$, by $\Med(X)$. Let $L(M(D), q(D))$ be any loss function that is weakly decreasing in $M(D)$ for $M(D) < q(D))$.
    \begin{itemize}
      \item{$\Med[M(D)] = q(D)$.}
      \item{If $g(x) = 0$ for $x < f(a)$, then the risk, $\E[L(M(D), q(D))]$, is minimized among estimators subject to the above assumptions.}
    \end{itemize}
\end{thm}

\begin{proof}
    Recall that the noisy transformed query value created in Line~\ref{line:transform_noise} of Algorithm~\ref{alg:transform} is denoted by $\tilde{v}$ and the algorithm's output is $M(D) = g(\tilde{v})$. Note that $f^{-1}$ is monotonically increasing.
    For the first item in Theorem~\ref{thm:mdnUnbiasTransform}, we have the following: 

    \begin{align*}
    P[q(D) < g(\tilde{v})] 
    &= P[f(q(D) + a) \leq \tilde{v}]P[q(D) < f^{-1}(\tilde{v}) - a | f(q(D) + a) \leq \tilde{v}] \\
      &\;\; + 
      P[f(a) \leq \tilde{v} < f(q(D) + a)]P[q(D) < f^{-1}(\tilde{v}) - a | f(a) \leq \tilde{v} < f(q(D) + a)] \\
      &\;\; + 
      P[\tilde{v} < f(q(D) + a)]P[q(D) < g(\tilde{v}) | \tilde{v} < f(a)] \\
    &= .5 P[f(q(D) + a) < \tilde{v} | f(q(D) + a) \leq \tilde{v}] \\
      &\;\; + 
      P[f(a) \leq \tilde{v} < f(q(D) + a)]P[f(q(D) + a) < \tilde{v} | f(a) \leq \tilde{v} < f(q(D) + a)] \\
      &\;\; + 
      P[\tilde{v} < f(q(D) + a)]P[q(D) < g(\tilde{v}) | \tilde{v} < f(a)] \\
    &= .5
    \end{align*}

To see the last equality, recall our assumption that, when $\tilde{v} < f(a)$, we have $g(\tilde{v}) \leq 0 \leq q(D)$, so $P[q(D) < g(\tilde{v}) | \tilde{v} < f(a)] = 0$. Likewise, it is clear that $P[f(q(D) + a) < \tilde{v} | f(a) \leq \tilde{v} < f(q(D) + a)] = 0$ and $ P[f(q(D) + a) < \tilde{v} | f(q(D) + a) \leq \tilde{v}] = 1$.

It follows that $q(D)$ is the median of $g(\tilde{v}) = M(D)$.

For the second item in Theorem~\ref{thm:mdnUnbiasTransform}, recall that, by assumption,  $g(\tilde{v}) \leq 0 \leq q(D)$ for $\tilde{v} < f(a)$, so $L(M(D),q(D)) = L(g(\tilde{v}),q(D))$ is weakly decreasing in $\tilde{v}$ over the same range.

The risk of the mechanism can therefore be bounded below by  
\begin{align*}
    \E[L(M(D), q(D))]
    &= P[\tilde{v} < f(a)]\E[L(g(\tilde{v}), q(D)) | \tilde{v} < f(a)] \\
    &\;\; + P[f(a) \leq \tilde{v}]\E[L(f^{-1}(\tilde{v})-a, q(D)) | f(a) \leq \tilde{v}] \\
    &\geq P[\tilde{v} < f(a)]\E[L(0, q(D)) | \tilde{v} < f(a)] \\
    &\;\; + P[f(a) \leq \tilde{v}]\E[L(f^{-1}(\tilde{v})-a, q(D)) | f(a) \leq \tilde{v}].
\end{align*}

This lower bound on the risk is clearly achieved by

\[ g(x)  = \begin{cases*}
                    0 & if $x < f(a)$  \\
                    f^{-1}(x)-a & if  $f(a) \leq x$.
                 \end{cases*} \]%
                 
This completes the proof.
\end{proof}

\section{Transformation Mechanism Error Bounds}
\label{subsec:errorbounds}
Here, we present quantitative measures of error incurred by the transformation mechanisms. These are presented in the form of prediction intervals for the noisy count, given the true count. That is, given the true count, the noisy count will be within a range $[a,b]$ with some probability $p$.

\begin{thm}[Transformation Mechanism Error Bounds]\label{thrm:transformation_bounds}
Let $\Phi^{-1}(p,\mu, \sigma) = \mu + \sigma\sqrt{2}\erf^{-1}(2p -1)$ be the inverse CDF of the normal distribution with mean $\mu$ and standard deviation $\sigma$. Let $X$ be the range given by $[\Phi^{-1}((1-p)/2,f(q(D)+a), \sigma),\Phi^{-1}((1+p)/2,f(q(D)+a), \sigma)]$.
Under Assumption~\ref{as:transformInput}, the output of $Algorithm~\ref{alg:transform}\allowbreak(q(D), a, \sigma, f, g)$, $\tilde{S}$, lies within $[\min_{x\in X}g(x),\max_{x\in X}g(x)]$ with probability $p$.
\end{thm}

\begin{proof}
    Recall that the noisy transformed query value created in Line~\ref{line:transform_noise} of Algorithm~\ref{alg:transform} is denoted by $\tilde{v}$ and the algorithm's output is $\Tilde{S} = g(\tilde{v})$. Note that $\tilde{v}$ is distributed according to $N(f(q(D)+a), \sigma^2)$ and, as such, has the prediction interval $[\Phi^{-1}((1-p)/2,f(q(D)+a), \sigma),\Phi^{-1}((1+p)/2,f(q(D)+a), \sigma)]$ with probability $p$. Since $\tilde{S} = g(\tilde{v})$, the prediction interval for $\tilde{S}$ has bounds given simply by the minimum and maximum of the estimator $g(x)$ within the prediction interval for $\tilde{v}$. 
\end{proof}

These bounds give the prediction interval for the output of the transformation mechanism for any valid transformation and estimator. Note that the estimators we discuss here are typically asymmetrically distributed, resulting in asymmetric prediction intervals. This can be seen below, where we derive the $95\%$ prediction intervals for the square root transformation and log transformation.

\begin{cor}[Square Root Transformation $95\%$ Error Bound Mean Unbiased Estimator]
    Let $a > 0$, $f(x) = \sqrt{x}$, $g(x) =x^2 -\sigma^2 -a $, $p=.95$, and $q(D) \in [0, \infty)$.  Let $X$ be the range given by $[\Phi^{-1}((1-p)/2,f(q(D)+a), \sigma),\Phi^{-1}((1+p)/2,f(q(D)+a), \sigma)]$. 
    
    The $95\%$ prediction interval for the transformation mechanism is $[-\sigma^2 -a,q(D) + 2\sqrt{q(D)+a}(1.96\sigma)  + (1.96^2 - 1)\sigma^2 ]$ when $0 \in X$ and $[q(D) - 2\sqrt{q(D)+a}(1.96\sigma)  + (1.96^2 - 1)\sigma^2,q(D) + 2\sqrt{q(D)+a}(1.96\sigma)  + (1.96^2 - 1)\sigma^2 ]$ when $0\notin X$.
\end{cor}
\begin{proof}
Note that, with $p=.95$, we have $\Phi^{-1}((1-p)/2,\sqrt{q(D)+a}, \sigma) \approx \sqrt{q(D)+a} - 1.96\sigma$ and $\Phi^{-1}((1+p)/2,\sqrt{q(D)+a}, \sigma) \approx \sqrt{q(D)+a} + 1.96\sigma$. 
Note that the global minimum of $g(x)$ is at $x=0$. Similarly, since $q(D) \geq 0$, the maximum of $g(x)$ is always at the largest value in $X$. Likewise, since $q(D) \geq 0$, if $X$ does not contain the global minimum at $x=0$, then the minimum of $g(x)$ over $X$ is realized at the minimum value in $X$. Therefore, when $0\in X$, the lower bound is $g(0) = -\sigma^2 -a$ and the upper bound is as follows.
    $$g(\sqrt{q(D)+a} + 1.96\sigma) = $$ 
    $$ (\sqrt{q(D)+a} + 1.96\sigma)^2 - \sigma^2 - a=$$
    $$ q(D)+a + 2\sqrt{q(D)+a}(1.96\sigma)  + (1.96\sigma)^2 - \sigma^2 - a=$$
    $$ q(D) + 2\sqrt{q(D)+a}(1.96\sigma)  + (1.96\sigma)^2 - \sigma^2 =$$
    $$ q(D) + 2\sqrt{q(D)+a}(1.96\sigma)  + (1.96^2 - 1)\sigma^2 $$

When $0 \notin X$, the upper bound remains the same, but the lower bound is the lower boundary of $X$, as follows. 
    $$g(\sqrt{q(D)+a} - 1.96\sigma) = $$ 
    $$ (\sqrt{q(D)+a} - 1.96\sigma)^2 - \sigma^2 - a=$$
    $$ q(D)+a - 2\sqrt{q(D)+a}(1.96\sigma)  + (1.96\sigma)^2 - \sigma^2 - a=$$
    $$ q(D) - 2\sqrt{q(D)+a}(1.96\sigma)  + (1.96\sigma)^2 - \sigma^2 =$$
    $$ q(D) - 2\sqrt{q(D)+a}(1.96\sigma)  + (1.96^2 - 1)\sigma^2 $$
    
\end{proof}

For the parameters $a, \sigma = 1$ and query value $q(D) = 1000$, these bounds result in a range of approximately $[879, 1127]$ which is $127$ in the positive direction and $121$ in the negative direction. Due to the asymmetry of the estimator, these bounds are larger in the positive direction than in the negative direction. This can be seen for the log transformation, as well, in the following.
\begin{cor}[Log Transformation $95\%$ Error Bound Mean Unbiased Estimator]
    Let $a > 0$, $f(x) = \ln{(x)}$, $g(x) = e^{x-\frac{\sigma^2}{2}} - a$, $p=.95$, and $q(D) \in [0, \infty)$. Let $X$ be the range given by $[\Phi^{-1}((1-p)/2,f(q(D)+a), \sigma),\Phi^{-1}((1+p)/2,f(q(D)+a), \sigma)]$. The $95\%$ prediction interval for the transformation mechanism is $[(q(D)+a)e^{-1.96\sigma - \frac{\sigma^2}{2}}-a,(q(D)+a)e^{1.96\sigma - \frac{\sigma^2}{2}}-a]$.
    \label{cor:log}
\end{cor}
\begin{proof}
    Note that $\Phi^{-1}((1-p)/2,\ln(q(D)+a), \sigma) \approx \ln(q(D)+a)) - 1.96\sigma$ and $\Phi^{-1}((1+p)/2,\ln(q(D)+a) , \sigma) \approx \ln(q(D)+a)+ 1.96\sigma$. Since $g(x)$ is monotonic in $x$, the minimum and maximum are respectively at the lower and upper boundaries of $X$. 
    Start with the lower bound. 
    $$g(\ln(q(D)+a) - 1.96\sigma) = $$ 
    $$ e^{\ln(q(D)+a) - 1.96\sigma -\frac{\sigma^2}{2}} -a = $$
    $$(q(D)+a)e^{-1.96\sigma - \frac{\sigma^2}{2}}-a$$
    Similarly, for the upper bound.
    $$g(\ln(q(D)+a) + 1.96\sigma) = $$ 
    $$ e^{\ln(q(D)+a) + 1.96\sigma -\frac{\sigma^2}{2}} -a = $$
    $$(q(D)+a)e^{1.96\sigma - \frac{\sigma^2}{2}}-a$$
\end{proof}

For the parameters $a, \sigma = 1$ and query value $q(D) = 1000$, these bounds result in a range of approximately $[85, 4309]$, which is $3309$ in the positive direction and $915$ in the negative direction. This again results in a much higher range in the positive direction than in the negative direction. This is a result of the asymmetry of the estimator, which stretches positively noised values to even larger values while negatively noised values are compressed together.

\par
Among the mechanisms described here, the log transformation mechanism is unique in having error bounds that are linear in $q(D)$, for any choice of parameters or the confidence value, $p$. This is due to a natural property of exponential functions, where $\exp(a+b) = \exp(a)\exp(b)$. This property transforms the additive differences between the transformed query and its bounds into multiplicative differences when the estimator is applied.
\par
The error bounds for all slowly scaling transformation mechanisms are dependent on the true query values and cannot be released directly in a privacy-preserving manner. Much like the policy function, the functional form of the error bounds can be released instead. This would allow a data curator to find the prediction interval for any known query value without directly leaking the true query values.

\section{Additive Mechanism Error Bounds}
Here, we present quantitative measures of error incurred by the additive mechanisms. These are presented in the form of prediction intervals for the noisy count given the true count. That is, given the true count, the noisy count will be within a range $[a,b]$ with probability $p$. Since the additive mechanisms are sampled directly from known distributions, it is easy to find these bounds using the inverse CDF function as follows.

\begin{thm}[Additive Mechanism Error Bounds]\label{thrm:additive_bounds}
Let $f_Z$ be a noise distribution that has density proportional to $e^{f\left(|z|\right)}$, where $f:[0,\infty) \to \R$ is (weakly) decreasing and (weakly) convex.

Let $\cdf^{-1}(p,\mu) $ be the inverse CDF of $f_Z$  when centered at $\mu$. The output of Algorithm~$\ref{alg:additive}\allowbreak(q(D), f_Z)$ lies within $[\cdf^{-1}((1-p)/2,q(D)) ,\cdf^{-1}((1+p)/2,q(D)) ]$ with probability at least $p$.

\end{thm}
\begin{proof}
This follows directly from the definition of the inverse CDF function of a distribution and the fact that $((1+p)/2-(1-p)/2) = p$.
\end{proof}
Unlike the transformation mechanism, these bounds are symmetric about the true query answer. Further, the error bounds for an additive mechanism (that is, the interval obtained by subtracting the true query answer from every point in the prediction interval from Theorem~\ref{thrm:additive_bounds}) are independent of the true query answer and can be published directly. We use Theorem~\ref{thrm:additive_bounds} to derive the bounds for two examples of additive mechanisms below.

\begin{cor}[Generalized Gaussian $p = 1/2$]
   Let $f_Z(z) = \frac{1}{4\sigma\Gamma(\frac{1}{p})}e^{-\left(\frac{|z|}{\sigma}\right)^{1/2}}$. The $95\%$ prediction interval for Algorithm~\ref{alg:additive}$(q(D), f_Z)$ is $[q(D) -0.963\sigma, q(D) +0.963\sigma]$.
\end{cor}

\begin{proof}
    The inverse CDF for the generalized Gaussian distribution with $p = 1/2$ centered at $0$ is as follows. 
    \begin{equation}
    \text{CDF}^{-1}(x) =    
    \begin{cases}
        0 &  x = \frac{1}{2} \\
        \sigma\left[-W\left(\frac{2x-2}{e}\right)-1\right]^2 &  x > \frac{1}{2} \\
        -\sigma\left[-W\left(\frac{-2x}{e}\right)-1\right]^2  & x < \frac{1}{2}
    \end{cases}
\end{equation}
We then plug in the values for $ x = 0.025$ for the lower bound and $x = 0.975$ for the upper bound, and shift the range to be centered at $q(D)$, to obtain the bound $[q(D) -\sigma\left(-W\left(\frac{-0.05}{e}\right)-1\right)^2, q(D) +\sigma\left(-W\left(\frac{-0.05}{e}\right)-1\right)^2] \approx [q(D)-0.963\sigma, q(D) + 0.963\sigma]$
\end{proof}

The same can be done for special cases of the  exponential polylogarithmic distribution.

\begin{cor}
[Exponential Polylogarithmic Distribution $p =2$]
\label{cor:ExpPolyErrBnd}
   Let$f_Z(z) \propto e^{-d\ln\left(\frac{|z|}{\sigma}+a\right)^2}$, with  $\sigma > 0, a \geq e$, and $d>0$. Let $\Phi^{-1}(\cdot)$ denote the standard normal quantile function. The $95\%$ prediction interval for Algorithm~\ref{alg:additive}$(q(D), f_Z)$ is as follows. $$\left[q(D) - \sigma \exp \left[\left([2d]^{-1/2}\Phi^{-1}\left[\left(\left[1-\Phi\left(\frac{\ln(a)-(2d)^{-1}}{(2d)^{-1/2}}\right)\right][0.95] \right) + \Phi\left(\frac{\ln(a)-(2d)^{-1}}{(2d)^{-1/2}}\right) \right]\right) + (2d)^{-1} \right] + \sigma a, \right. $$ $$ \left.  q(D) + \sigma \exp \left[\left([2d]^{-1/2}\Phi^{-1}\left[\left(\left[1-\Phi\left(\frac{\ln(a)-(2d)^{-1}}{(2d)^{-1/2}}\right)\right][0.95] \right) + \Phi\left(\frac{\ln(a)-(2d)^{-1}}{(2d)^{-1/2}}\right) \right]\right) + (2d)^{-1} \right] - \sigma a\right]$$
\end{cor}

\begin{proof}
    The inverse CDF for the exponential poly-logarithmic distribution with $p = 2$ centered at $0$ is provided in Table~\ref{tab:Foxdenscases} and can be expanded into the following expression. 
\begin{equation}  
    \begin{cases} 
       - \sigma \exp \left[\left([2d]^{-1/2}\Phi^{-1}\left[\left(\left[1-\Phi\left(\frac{\ln(a)-(2d)^{-1}}{(2d)^{-1/2}}\right)\right][1- 2x] \right) + \Phi\left(\frac{\ln(a)-(2d)^{-1}}{(2d)^{-1/2}}\right) \right]\right) + (2d)^{-1} \right] + \sigma a& x <\frac{1}{2}
      \\
      \sigma \exp \left[\left([2d]^{-1/2}\Phi^{-1}\left[\left(\left[1-\Phi\left(\frac{\ln(a)-(2d)^{-1}}{(2d)^{-1/2}}\right)\right][2x - 1]\right) + \Phi\left(\frac{\ln(a)-(2d)^{-1}}{(2d)^{-1/2}}\right) \right]\right) + (2d)^{-1}\right] - \sigma a& x >\frac{1}{2}
      \\
     0 & x = \frac{1}{2}
   \end{cases}
\end{equation}
We then plug in the values for $ x = 0.025$ for the lower bound and $x = 0.975$ for the upper bound, and shift the range to be centered at $q(D)$, to obtain the desired bounds.
\end{proof}

For parameters $\sigma =1 $, $ d = 1$, $a =e$ this results in a bound of approximately $[q(D)-5.418,q(D) + 5.418]$

\section{Protecting Values -- Bounded Neighbor Privacy Guarantees} \label{sec:metric-DPNew}

The privacy concepts treated thus far - those in Definitions~\ref{def:zCDP}, \ref{def:PRzCDP}, and \ref{def:PRDP} - are designed to keep an adversary from distinguishing between pairs of neighboring databases which differ in the \textit{presence} of a single record. These privacy concepts keep each record's presence private from an adversary who may already know that record's value (e.g., an adversary who knows a certain business's revenue, but does not know whether that business in included in a particular dataset). In this section, we develop privacy concepts that protect the privacy of records' \textit{values} from an adversary who may already know of their presence (e.g., an adversary who knows that a certain business is in a dataset, but does not know its revenue).

With these privacy concepts, we derive privacy guarantees that express how well our slowly scaling mechanisms inhibit adversaries from distinguishing between a database containing the record $r$ and an otherwise identical neighboring database that instead contains the record $r'$. 

Denote attribute $c$ of record $r$ by $r.c \in \R$. We focus particularly on the case where $r'.c$ equals either $r.c + \delta$ or $\delta r.c$ for some scalar $\delta$, and where our slowly scaling mechanisms are used to protect a sum of attribute $c$ over all records in the database. Using these neighboring databases, we can quantify the difficulty an adversary has in distinguishing between a record's true value $r$ and nearby values. When the privacy guarantee is high for a wide range of values of $\delta$, a data curator may be satisfied that adversaries will be unable to accurately estimate the record $r$.

Using language from \cite{KiferEtAl2011}, we call databases that differ only in the value of one record "bounded" neighbors. 

\begin{defi}[Bounded Neighboring Databases] \label{def:BndNeighbors}
Two databases $D$ and $D'$ in $\mathcal{D}$ are considered \emph{bounded neighboring databases} with the \emph{differing pair} $(r,r')$ if, for some database $D_0 \in \mathcal{D}$ and records $r, r'$, we have either $D = D_0 \cup \{r\}$ and $D'= D_0 \cup \{r'\}$, or $D = D_0 \cup \{r'\}$ and $D'= D_0 \cup \{r\}$.
\end{defi}

Using this terminology, the neighbors defined in Definition~\ref{def:UnbndNeighbors} and used elsewhere in this paper are called ``unbounded'' neighbors.

We start by defining privacy concepts that bound the privacy loss between bounded neighboring databases with a given differing pair, $(r,r')$. We give these bounded analogues of both PRzCDP (Definition~\ref{def:PRzCDP}) and PRDP (Definition~\ref{def:PRDP}) below.

\begin{defi}[Bounded Per-Record Differential Privacy] \label{def:BndPRDP}
The randomized mechanism $M$ satisfies the following privacy concepts with the policy function $P$ iff the following conditions hold for all bounded neighboring databases $D, D'$ with differing pair $(r,r')$.
\begin{itemize}
    \item Bounded $P$-Per-Record Zero-Concentrated Differential Privacy (Bounded $P$-PRzCDP): 
    $$d_\alpha \left(M(D) \| M(D') \right) \leq \alpha P(r,r') \text{ for all } \alpha \in (1, \infty).$$
    \item Bounded $P$-Per-Record Differential Privacy (Bounded $P$-PRDP): $$d_\infty \left(M(D) \| M(D') \right) \leq P(r,r').$$
\end{itemize}
\end{defi}

When the policy function $P(r,r')$ is a metric, the definition of bounded PRDP in Definition~\ref{def:BndPRDP} is the same as that of another privacy concept variously called "metric DP", "Lipschitz DP", and "$d$-DP" (see, for example, \cite{ChatzikokolakisABP13} and \cite{koufogiannis2015optimality}). A closely related privacy concept is establishment Gaussian DP, introduced in independent, concurrent work by \cite{WebbEtAl2023} for use with the same sort of establishment data that motivated our work. Establishment GDP considers bounded neighbors and bounds a notion of privacy loss based on the Gaussian DP concept introduced in \cite{DongSu2022}. \cite{WebbEtAl2023} also propose establishment GDP mechanisms which include special cases of the transformation mechanisms described in Algorithm~\ref{alg:transform}.

Like in the unbounded per-record DP variants, where a policy function bounds a notion of privacy loss as a function of the record $r$, the bounded per-record privacy concepts have a policy function that bounds the privacy loss as a function of the differing pair $r,r'$. Just as we needed the per-record sensitivity (Definition~\ref{def:PR-sensitivity}) to derive the (unbounded) PRzCDP and PRDP guarantees above, we need an analogous sensitivity notion that captures the maximal sensitivity of the query $q$ to switching between bounded neighbors with a particular differing pair.

\begin{defi}[Bounded Per-Record Sensitivity] 
\label{def:PP-sensitivityNew}
The bounded per-record sensitivity of the query $q:\mathcal{D} \to \R$ to the differing pair ($r$, $r'$) is 
$$\Delta(r,r') \equiv \sup_{\substack{D,D' | D,D' \text{ bounded neighbors } \\ \text{with differing pair } (r,r')}} |q(D) - q(D^\prime)|.$$
\end{defi}

With these concepts in hand, we now derive the bounded PRzCDP and bounded PRDP guarantees of our slowly scaling mechanisms, starting with the transformation mechanisms. 

\begin{thm}[Bounded PRzCDP Guarantees for Transformation Mechanisms]
\label{thm:transform-BndPRzCDPTight}
    Let
    Assumption~\ref{as:transformInput} hold and denote by $\Delta_f(r,r')$ the bounded sensitivity of the query $f(q(D) + a)$, defined in Definition~\ref{def:PP-sensitivityNew}. The output of Algorithm~$\ref{alg:transform}\allowbreak(q(D), a, \sigma, f, g)$ satisfies bounded $P$-PRzCDP for $P(r,r') = \frac{\Delta_f(r,r')^2}{2\sigma^2}$.
\end{thm}

\begin{proof}
    Let $M(D)$ denote Algorithm~$\ref{alg:transform}\allowbreak(q(D), a, \sigma, f, g)$. Note that $M(D) = g(f(q(D) + a) + N(0,\sigma^2))$. It follows from Definition~\ref{def:BndPRDP} that $M$ satisfies bounded PRzCDP for any policy function that is lower bounded by 

    $$\sup_{\substack{\alpha \in (1,\infty), D,D' | D,D' \text{ bounded neighbors } \\ \text{with differing pair } (r,r')}}
    \frac{d_\alpha \left(M(D) \| M(D') \right)}{\alpha}$$

    To derive $P(r,r')$, we bound this supremum using the data processing inequality \citep{renyi1961measures}, Lemma~\ref{lem:BS-Gauss}, and Definition~\ref{def:PP-sensitivityNew}, as follows.
    \begin{align*}
        &\sup_{\substack{\alpha \in (1,\infty), D,D' | D,D' \text{ bounded neighbors } \\ \text{with differing pair } (r,r')}} \frac{d_\alpha \left(M(D) \| M(D') \right)}{\alpha} \\
        &\leq \sup_{\substack{\alpha \in (1,\infty), D,D' | D,D' \text{ bounded neighbors } \\ \text{with differing pair } (r,r')}} \frac{d_\alpha \left(f(q(D) + a) + N(0,\sigma^2) \| f(q(D') + a) + N(0,\sigma^2) \right)}{\alpha} \\
        &= \sup_{\substack{\alpha \in (1,\infty), D,D' | D,D' \text{ bounded neighbors } \\ \text{with differing pair } (r,r')}} \frac{|f(q(D)+a) - f(q(D')+a)|^2}{2\sigma^2} \\
        &= \left(\sup_{\substack{D,D' | D,D' \text{ bounded neighbors } \\ \text{with differing pair } (r,r')}} |f(q(D)+a) - f(q(D')+a)|\right)^2 /2\sigma^2 \\
        &= \frac{\Delta_f(r,r')^2}{2\sigma^2} \\
        &= P(r,r'). 
    \end{align*}    
\end{proof}

We now similarly derive the bounded PRDP and bounded PRzCDP guarantees for the additive mechanisms. This requires that we first state extensions of Lemma~\ref{lem:PRDP-implies-PRzCDP} and Theorem~\ref{thm:AddMechLogProbForm} to our bounded DP concepts. The proofs of these extensions closely follow the proofs of Lemma~\ref{lem:PRDP-implies-PRzCDP} and Theorem~\ref{thm:AddMechLogProbForm}, and so are omitted here.

\begin{lem}[Bounded PRDP Implies Bounded PRzCDP]
\label{lem:PPDP-implies-PPzCDP}
   If a mechanism, $M$, satisfies bounded $P$-PRDP, it also satisfies bounded $P$-PRzCDP and bounded $P'$-PRzCDP, where $P'(r) = \tanh(P(r)/2)P(r)$.
\end{lem}

\begin{thm}
\label{thm:PPAddMechLogProbForm}
    An additive mechanism $M\left(D\right) \equiv q\left(D\right) +  Z$, with a continuous random variable $Z$ having density $f_Z$, satisfies bounded $P$-PRDP if and only if, for all bounded neighboring databases $D, D'$ with differing pair $(r,r')$
    $$\sup_{z} \left|\ln\left(f_Z\left(z\right)\right) - \ln\left(f_Z\left(z+q\left(D\right)-q\left(D^\prime\right)\right)\right)\right| \leq P\left(r, r'\right).$$
\end{thm}

The following theorem is an analogue of Theorem~\ref{thm:BndLogProbSymCnvxDens} for our bounded DP concepts.

 \begin{thm}[Bounded Privacy of Symmetric, Absolute-Log-Convex Densities]
\label{thm:BndPrivSymCnvxDens}
Let $q$ be a query with bounded sensitivity $\Delta(r,r')$, as defined in Definition \ref{def:PP-sensitivityNew}. Let $f_Z$ be a noise distribution that has density proportional to $e^{f\left(|z|\right)}$, where $f:[0,\infty) \to \R$ is (weakly) decreasing and (weakly) convex.

For all pairs of records $r,r'$, Algorithm~\ref{alg:additive}$(q(D), f_Z)$, denoted below by $M$, satisfies bounded $P$-PRDP with $P(r,r')=f(0) - f(\Delta(r,r'))$. This policy function is tight, in the sense that $M$ is not bounded $P^\prime$-PRDP for an alternative policy function $P^\prime$ if there exists a pair of records $r,r'$ such that $P^\prime(r,r') < P(r,r')$. $M$ also satisfies bounded $P''$-PRzCDP with $P''(r,r') = \tanh(P(r,r')/2) P(r,r')$.
\end{thm}

\begin{proof}
Applying Theorem~\ref{thm:PPAddMechLogProbForm} and then Theorem~\ref{thm:BndLogProbSymCnvxDens}, we see that $M$ is bounded $P$-PRDP if and only if, for all $r,r'$,

\begin{align}
    P(r,r') &\geq \sup_{\substack{z,D,D' | D,D' \text{ bounded neighbors } \\ \text{ with differing pair } (r,r')}} |\ln(f_Z(z)) - \ln(f_Z(z+q(D)-q(D^\prime)))| \\
    &= \sup_{\substack{D,D' | D,D' \text{ bounded neighbors } \\ \text{ with differing pair } (r,r')}} f(0) - f(|q(D)-q(D^\prime)|).
\end{align}

Because $f$ is a decreasing function, $f(0) - f(|q(D)-q(D^\prime)|)$ is increasing in $|q(D)-q(D^\prime)|$ and has its supremum at $|q(D)-q(D^\prime)| = \Delta(r,r')$. This yields the result. Furthermore, by Lemma~\ref{lem:PPDP-implies-PPzCDP}, $M$ also satisfies bounded $P''$-PRzCDP with $P''(r,r') = \tanh(P(r,r')/2) P(r,r')$.
\end{proof}

The bounded privacy guarantees allow us to determine the greatest possible privacy loss between bounded neighbors where one neighbor contains a real record $r$ and another contains a hypothetical alternative record $r'$ that is within some neighborhood of $r$.

\subsection{Global Sums}

\begin{figure}[t]
    \centering
    \includegraphics[width = 1.0\textwidth]{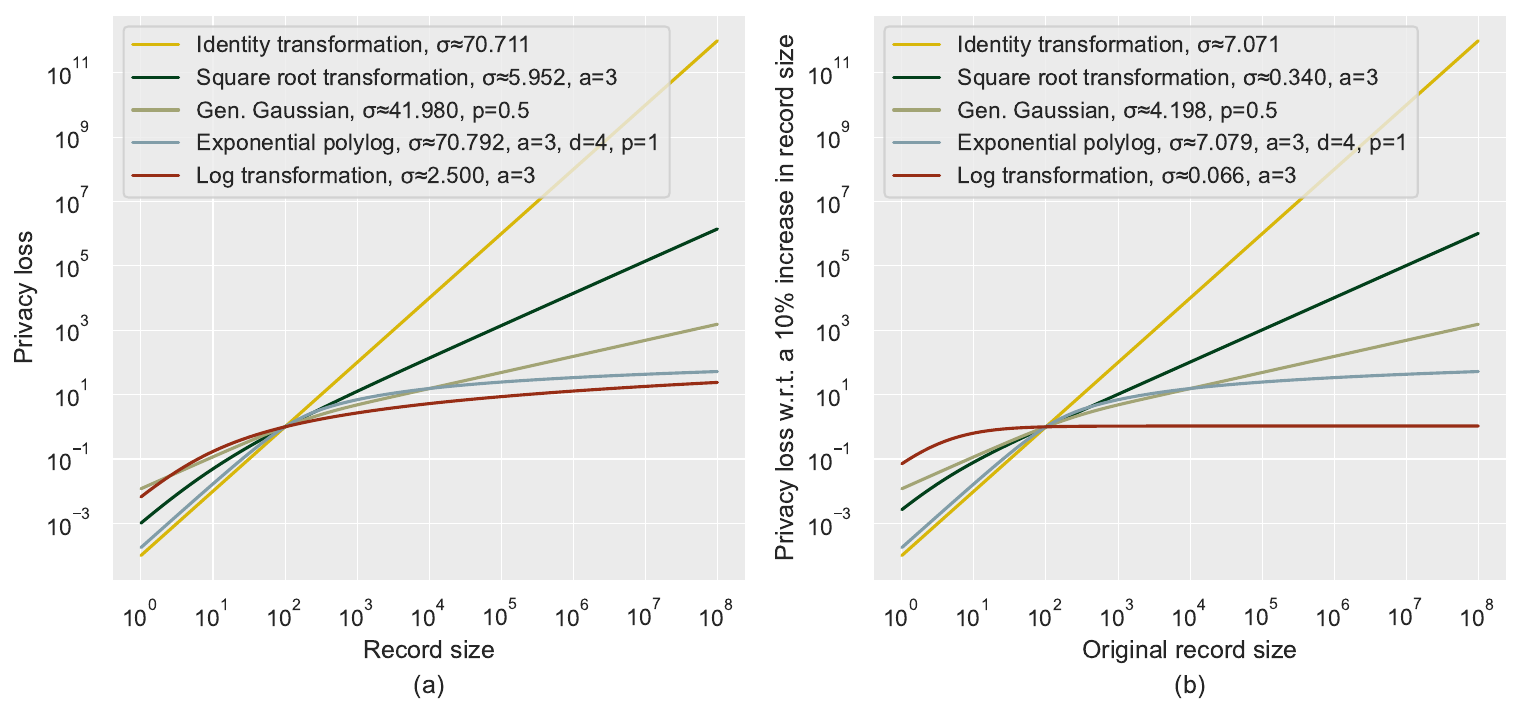}
    \caption{Unbounded and bounded PRzCDP policy functions under several mechanisms applied to the global sum query $q(D) = \sum_{r \in D} r.c$. The $x$-axes show values of $r.c$. Panel (a) shows unbounded PRzCDP policy functions. Panel (b) shows bounded PRzCDP policy functions with respect to a 10\% increase in record value (i.e., with a differing pair such that $r'.c = 1.1 r.c$). All mechanisms are tuned to have policy functions that equal 1 for records with $r.c = 100$.}
    \label{fig:pl_metricDP}
\end{figure}

In this section, we derive the mechanisms' bounded privacy guarantees for global sum queries. Denoting the scalar attribute $c$ of record $r$ by $r.c \in \R$, a global sum query $q$ of attribute $c$ is simply the sum of $r.c$ over every record in the database; $q(D) = \sum_{r \in D} r.c$. We obtain the bounded PRzCDP and bounded PRDP policy functions and further highlight these policy functions' values for differing pairs $r,r'$ satisfying $r'.c = r.c+\delta$ or $r'.c = \delta r.c$ for some scalar $\delta$. These expressions let us easily determine the privacy loss against an adversary trying to distinguish between a true value $r.c$ and other, nearby values.

We start by stating the bounded sensitivity of the global sum query. The proof is trivial.

\begin{lem}[Bounded Sensitivity of Sum Query] \label{lem:BndSensGlobalSum}
    The bounded sensitivity of the global sum query $q(D) = \sum_{r \in D} r.c$ is $|r.c - r'.c|$.
\end{lem}

The above suffices to obtain the privacy guarantees for the additive mechanisms, but, for the transformation mechanisms, Theorem~\ref{thm:transform-BndPRzCDPTight} requires that we obtain the sensitivities of the transformation function $f$ composed with the global sum query. These sensitivities are given in the following theorem.

\begin{lem}[Bounded Sensitivity of Transformed Sum Query]\label{lem:BndSensTransSum}
    Let $a \geq 0$ and $f:[a,\infty) \to \mathcal{F} \subseteq \R$ be an increasing, concave function. Let $r.c \geq 0$ for all records $r$. For the global sum query $q(D) = \sum_{r \in D} r.c$, the bounded sensitivity of $f(q(D) + a)$ is $\Delta_f(r,r') = |f(r.c + a) - f(r'.c + a)|$.
\end{lem}

\begin{proof}
    The bounded sensitivity of $f(q(D) + a)$ is defined as
    \begin{align}
        \Delta_f(r,r') &\equiv 
        \sup_{D,D' | D,D' \text{ bounded neighbors with differing pair } (r,r')} |f(a + \sum_{s \in D} s.c) - f(a + \sum_{s \in D'} s.c)| \\
        &= \sup_{D_0 \in \mathcal{D}; r, r'} |f(a + r.c + \sum_{s \in D_0} s.c) - f(a + r'.c + \sum_{s \in D_0} s.c) |.
    \end{align}

    Because $f$ is concave and increasing, the final supremand is decreasing in $\sum_{s \in D_0} s.c$. $\sum_{s \in D_0} s.c$ has a minimum of 0, so the supremum is
    $$\Delta_f(r,r') = |f(r.c + a) - f(r'.c + a)|.$$    
\end{proof}

The following theorem lays out the bounded PRzCDP guarantees for the transformation mechanisms applied to a global sum query. It also provides easily interpretable bounds on the policy function when the differing pairs satisfy $r'.c = r.c + \delta$ and $r'.c = \delta r.c$.

\begin{cor}[Bounded PRzCDP Guarantee for Transformation Mechanisms on Global Sums] \label{cor:BndPrivTransMechGlobSum}
    Let Assumption~\ref{as:transformInput} hold and additionally assume that $a \geq 0$ and $r.c \geq 0$ for all records $r$. With the global sum query $q(D) = \sum_{r \in D} r.c$, the output of Algorithm~$\ref{alg:transform}\allowbreak(q(D), a, \sigma, f, g)$ satisfies bounded $P$-PRzCDP for $P(r,r') = \frac{|f(r.c + a) - f(r'.c + a)|^2}{2\sigma^2}$. Easily interpretable bounds on special cases of this policy function are given below. \\

    \begin{enumerate}
        \item Let $r'.c = r.c + \delta \geq 0$.
            \begin{enumerate} 
                \item \label{it:AddNeighborTransGeneral} $P(r,r') = \frac{|f(r.c+a) - f(r.c + \delta + a)|^2}{2\sigma^2} \leq \frac{|f(a) - f(|\delta| + a)|^2}{2\sigma^2}$.
                \item \label{it:AddNeighborTransPower}If $f(x) = x^k$ for some $0 < k \leq 1$, then $P(r,r') \leq \frac{|\delta|^{2k}}{2\sigma^2}$.
                \item 
                \label{it:AddNeighborTransLog}
                If $a > 0$ and $f(x)=\ln(x)$, then $P(r,r') \leq \frac{|\ln(a) - \ln(|\delta| + a)|^2}{2\sigma^2}$.
            \end{enumerate}
        \item Let $r'.c = \delta r.c$, with $\delta \geq 0$.
            \begin{enumerate}
                \item \label{it:MultNeighborTransGeneral} $P(r, r') = \frac{|f(r.c+a) - f(\delta r.c + a)|^2}{2\sigma^2}$.
                \item \label{it:MultNeighborTransPower} If $f(x) = x^k$ for some $0 < k \leq 1$, then $P(r, r') \leq \frac{|(1-\delta^{k})r^k|^2}{2\sigma^2}$.
                \item \label{it:MultNeighborTransLog} If $r.c > 0$, $a > 0$, $\delta > 0$, and $f(x) = \ln{(x)}$, then $P(r, r') \leq \frac{|\ln{(\delta)}|^2}{2\sigma^2}$.
            \end{enumerate}
    \end{enumerate}
\end{cor}

\begin{proof}
    The policy function $P(r,r') = \frac{|f(r.c + a) - f(r'.c + a)|^2}{2\sigma^2}$ follows immediately from Lemma~\ref{lem:BndSensTransSum} and Theorem~\ref{thm:transform-BndPRzCDPTight}.

    Recall that, if $f$ is increasing and concave and if $x_0 < x_1$, then $f(x_1 +b) - f(x_0 + b)$ is nonnegative and decreasing in $b$. 
    
    If $\delta \geq 0$, then the inequality in Result~\ref{it:AddNeighborTransGeneral} follows from $r.c \geq 0$ and $f$ being concave and increasing. If $\delta < 0$, then we use these facts and the fact that $r.c + \delta \geq 0$ to obtain $\frac{|f(r.c + a) - f(r.c + \delta + a)|^2}{2\sigma^2} = \frac{|f(r.c + \delta - \delta + a) - f(r.c + \delta + a)|^2}{2\sigma^2} \leq \frac{|f(- \delta + a) - f(a)|^2}{2\sigma^2} = \frac{|f(|\delta| + a) - f(a)|^2}{2\sigma^2}$. 

    Result~\ref{it:AddNeighborTransPower} follows from Result~\ref{it:AddNeighborTransGeneral}, $a \geq 0$, and $f$ being concave, increasing, and defined at $0$.

    Result~\ref{it:AddNeighborTransLog} follows from Result~\ref{it:AddNeighborTransGeneral}.

    Result~\ref{it:MultNeighborTransPower} follows from $P(r, r') = \frac{|(r.c+a)^k - (\delta r.c + a)^k|^2}{2\sigma^2} \leq \frac{|r.c^k - \delta^k r.c^k|^2}{2\sigma^2} = \frac{|(1-\delta^{k})r^k|^2}{2\sigma^2}$. 
    
    Result~\ref{it:MultNeighborTransLog} follows from $P(r, r') = \frac{|\ln{(r.c+a)} - \ln{(\delta r.c + a)}|^2}{2\sigma^2} \leq \frac{|\ln{(r.c)} - \ln{(\delta r.c)}|^2}{2\sigma^2} = \frac{|\ln{(\delta)}|^2}{2\sigma^2}$. In both these derivations, the inequalities use the fact that $f$ is concave and increasing.
\end{proof}

The following theorem gives corresponding privacy guarantees for the additive mechanisms and follows immediately from Lemma~\ref{lem:BndSensGlobalSum} and Theorem~\ref{thm:BndPrivSymCnvxDens}.

\begin{cor}[Bounded Privacy Guarantees for Additive Mechanisms on Global Sums]
\label{cor:BndPrivAddMechGlobSum}
     Let $f_Z$ be a noise distribution that has density proportional to $e^{f\left(|z|\right)}$, where $f:[0,\infty) \to \R$ is (weakly) decreasing and (weakly) convex. With the global sum query $q(D) = \sum_{r \in D} r.c$, the output of Algorithm~\ref{alg:additive}$(q(D), f_Z)$ satisfies bounded $P$-PRDP for $P(r,r')=f(0) - f(|r.c - r'.c|)$. This algorithm also satisfies bounded $P'$-PRzCDP for $P'(r,r') = \tanh(P(r,r')/2) P(r,r')$. Expressions for special cases of this policy function are given below.

    \begin{enumerate}
    \item Let $r'.c = r.c + \delta$, with $\delta \in \R$. $P(r,r') = f(0) - f(|\delta|)$.
        \begin{enumerate}
            \item \label{it:AddNeighborsAddMechGenGauss} (Generalized Gaussian) If $f(|z|) = -(\frac{|z|} {\sigma})^p$ for $0 < \sigma$ and $0 < p \leq 1$, then $P(r,r') = (\frac{|\delta|} {\sigma})^p$
            \item \label{it:AddNeighborsAddMechExpPoly} (Exp polylog) If $f(|z|) = -d\ln(\frac{|z|}{\sigma} + a)^p$, for $\sigma > 0$, $a \geq e^{p-1}$, and either ($p>1$ and $d>0$) or ($p=1$ and $d>1$), then $P(r,r') = d[\ln(\frac{|\delta|}{\sigma} + a)^p - \ln(a)^p]$.
        \end{enumerate}
    \item  Let $r'.c = \delta r.c$ with $\delta \in \R$. $P(r,r') = f(0) - f(|(1-\delta) r.c|)$.
        \begin{enumerate}
            \item \label{it:MultNeighborsAddMechGenGauss} (Generalized Gaussian) If $f(|z|) = -(\frac{|z|} {\sigma})^p$ for $0 < \sigma$ and $0 < p \leq 1$, then $P(r,r') = (\frac{|(1-\delta) r.c|} {\sigma})^p$
            \item \label{it:MultNeighborsAddMechExpPoly} (Exp polylog) If $f(|z|) = -d\ln(\frac{|z|}{\sigma} + a)^p$, for $\sigma > 0$, $a \geq e^{p-1}$, and either ($p>1$ and $d>0$) or ($p=1$ and $d>1$), then $P(r,r') = d[\ln(\frac{|(1-\delta) r.c|}{\sigma} + a)^p - \ln(a)^p]$.
        \end{enumerate}
    \end{enumerate}
\end{cor}

\begin{table}[t]
    \centering
    \begin{tabular}{|c|c|c|c|}
  \multicolumn{4}{c}{Transformation Mechanism with Transformation $f$ -- Bounded PRzCDP Policy Functions}\\\hline
    \hline 
    \textbf{Name} &
     $\boldsymbol{f}$ &   
 \textbf{Additive}: $\boldsymbol{r'.c = r.c +\delta}$ &   
  \textbf{Multiplicative:} $\boldsymbol{r'.c = \delta r.c}$\\
    \hline
    Generic & $f(q(D) + a)$ & $\frac{|f(r.c+a) - f(r.c + \delta + a)|^2}{2\sigma^2}$ & $\frac{|f(r.c+a) - f(\delta r.c + a)|^2}{2\sigma^2}$\\
    \hline
    Identity & $q(D) +a $ & $\frac{|\delta|^2}{2\sigma^2}$ & $\frac{(1-\delta)^2 r^2}{2\sigma^2}$ \\
    \hline
    Square Root* & $\sqrt{q(D) + a}$ & $\frac{|\delta|}{2\sigma^2}$ & $\frac{(1-\sqrt{\delta})^2r}{2\sigma^2}$\\
     \hline
     Fourth Root* & $\sqrt[4]{q(D) + a}$ & $\frac{\sqrt{|\delta|}}{2\sigma^2}$ & $\frac{(1-\sqrt[4]{\delta})^2 \sqrt{r}}{2\sigma^2}$\\
    \hline
    Log* & $\ln(q(D) + a)$ & $\frac{(\ln (|\delta| + a) -\ln(a))^2}{2\sigma^2}$ & $\frac{\ln(\delta)^2}{2\sigma^2} $\\
    \hline
 \multicolumn{4}{c}{}\\
  \multicolumn{4}{c}{Additive Mechanism with Noise Density $\propto e^{f(|z|)}$ -- Bounded PRDP Policy Functions} \\\hline\hline
  \textbf{Name} &
     $\boldsymbol{f}$ &   
 \textbf{Additive:} $\boldsymbol{r'.c = r.c +\delta}$ &   
  \textbf{Multiplicative:} $\boldsymbol{r'.c = \delta r.c}$
  \\\hline
  Generic & $f(|z|)$ & $f(0) - f(|\delta|)$ & $f(0) - f(|(1-\delta) r.c|)$\\\hline
  Gen Gaussian & $-\left(\frac{|z|}{\sigma}\right)^p$ & $(\frac{|\delta|} {\sigma})^p$ & $(\frac{|(1-\delta) r.c|} {\sigma})^p$ \\\hline
  Exp Polylog & $-d\ln\left(\frac{|z|}{\sigma} + a\right)^p$ & $d[\ln(\frac{|\delta|}{\sigma} + a)^p - \ln(a)^p]$ & $d[\ln(\frac{|(1-\delta) r.c|}{\sigma} + a)^p - \ln(a)^p]$\\\hline
    \end{tabular}
    \caption{Bounded PRzCDP and bounded PRDP policy functions for a global sum query privatized by our transformation and additive mechanisms. Policy functions are given for differing pairs satisfying $r'.c = r.c + \delta$ and $r'.c = \delta r.c$. The policy functions shown for transformation mechanisms with stars in the Name column are loose, but tractable, bounds on privacy loss. The generic policy function in the first row gives tighter privacy loss bounds when applied to these mechanisms. See Corollaries~\ref{cor:BndPrivTransMechGlobSum} and \ref{cor:BndPrivAddMechGlobSum} for details.}
   \label{tab:metric-mechSummaryNew}
\end{table}

See Table~\ref{tab:metric-mechSummaryNew} for a summary of the results in Corollaries~\ref{cor:BndPrivTransMechGlobSum} and \ref{cor:BndPrivAddMechGlobSum}. From these corollaries, we see that every mechanism's policy function has an upper bound that depends on $\delta$, but not $r$ when $r'.c = r.c + \delta$. This means that, with appropriately set parameters, all of these mechanisms can guarantee any desired level of privacy against an attacker who is trying to determine the value of $r.c$ up to an additive error $\delta$. 

However, these protections may be inadequate when $r.c$ is very large. Guaranteeing that an attacker cannot learn a business' true annual revenue without a \$1 million margin of error may be adequate protection for, say, a local florist, but this may still be very disclosive for a large business with revenue in the billions. In that case, the large business may care more about protecting its revenue up to a multiplier $\delta$, and so we look at differing pairs with $r'.c = \delta r.c$.

For these differing pairs, only the log transformation mechanism has an upper bound that does not depend on $r.c$. All the other mechanisms have protections that weaken as $r.c$ grows. For this reason, the log transformation mechanism may be acceptable even when its unbounded PRzCDP guarantees can be matched with greater utility by an additive mechanism, as in our examples in Section~\ref{sec:empirical}. 

Despite that advantage, the other mechanisms' bounded policy functions can be quite slowly scaling in $r.c$. For the transformation mechanisms with $f(x)=x^k$, the policy function is bounded above by $\frac{|(1-\delta^{k})r^k|^2}{2\sigma^2}$, so the privacy loss can be made to scale more slowly by picking a smaller value of $k$. Likewise, the generalized Gaussian mechanism has the policy function $(\frac{|(1-\delta) r.c|} {\sigma})^p$, which scales more slowly with smaller values of $p$. The exponential polylog mechanism has the policy function $d[\ln(\frac{|(1-\delta) r.c|}{\sigma} + a)^p - \ln(a)^p]$, which scales at a polylogarithmic rate.  To see this slowly scaling behavior illustrated, see Figure~\ref{fig:pl_metricDP} for graphs of unbounded and bounded PRzCDP policy functions for a variety of mechanisms. Even though these mechanisms do not impose a hard limit on the privacy losses between these types of bounded neighbors, the slow growth rates of the privacy loss may be acceptable for the same reasons as the unbounded slowly scaling privacy guarantees.

\section{Properties of the Exponential Polylogarithmic Distribution} \label{sec:ExpPloyLog}

\subsection{Density Function}
\label{subsec:polylog-dens}
For the general case, with $p\geq 1$, the density of the exponential polylogarithmic distribution is given by
\begin{equation}
    f_{Z}(z)=\frac{pd^{\frac{1}{p}}}{2\sigma\Psi^{(\Gamma)}\left(p, d\ln(a)^{p}, d^{-\frac{1}{p}}\right)}e^{-d\ln\left(\frac{|z|}{\sigma}+a\right)^{p}},
\end{equation}
where
\begin{equation}\label{FoxHupper}
    \Psi^{(\Gamma)}\left(p, d\ln(x)^{p}, c\right)=\prescript{}{1}{\Psi}^{(\Gamma)}_{1}\left[{\left(\frac{1}{p}, \frac{1}{p}, d\ln(x)^{p}\right) \atop (1,0)}\Bigg|c\right]=\sum^{\infty}_{n=0}\Gamma\left(\frac{1}{p}+\frac{n}{p}, d\ln(x)^{p}\right)\frac{c^{n}}{n!}
\end{equation}
is the upper-incomplete Fox-Wright Psi function, introduced by \cite{srivastava2018}, evaluated at $c=d^{-\frac{1}{p}}$ and $x=a$. This converges when $p>1$ and $0<d<\infty$ or when $p=1$ and $d>1$. Here, $\Gamma\left(\frac{1}{p}+\frac{n}{p}, d\ln(x)^{p}\right)$ is the upper incomplete gamma function with parameter $\frac{n+1}{p}$ with lower limit $d\ln(x)^{p}$. 
\begin{proof}
First, we find the constant of integration $C$ such that $C\int^{\infty}_{-\infty}f_{Z}(z)dz=1$. Then, $C\int^{\infty}_{-\infty}f_{Z}(z)dz=2C\int^{\infty}_{0}e^{-d\ln(\frac{z}{\sigma}+a)^{p}}dz$ by symmetry of $f_{Z}(z)$ around zero. Then, letting $u=\ln(\frac{z}{\sigma}+a)$ we have $2C\sigma\int_{\ln(a)}^{\infty}e^{-du^{p}+u}dz=1$. Substituting the Taylor series for $e^{u}$ and changing the order of summation and integration, we have
\begin{equation*}
    2C\sigma \sum_{n=0}^{\infty}\frac{1}{n!}\int_{\ln(a)}^{\infty}u^{n}e^{-du^{p}}=1.
\end{equation*}
Substituting $t=du^{p}$ and simplifying
\begin{equation*}
    \frac{2\sigma C}{pd^{\frac{1}{p}}}\sum_{n=0}^{\infty}\Gamma\left(\frac{1}{p}+\frac{n}{p}, d\ln(a)^{p}\right)\frac{\left(d^{-\frac{1}{p}}\right)^{n}}{n!}=1,
\end{equation*}
The series in the above equation is the upper-incomplete Fox-Wright function. Since the incomplete Fox-Wright function is bounded above by the more familiar complete Fox-Wright function, conditions for convergence are summarized by \cite{kilbas2004}. The series converges when $p>1$ and $d<\infty$ or when $p=1$ and $d>1$. Therefore, $C=\frac{pd^{\frac{1}{p}}}{2\sigma\Psi^{(\Gamma)}\left(p, d\ln(a)^{p}, d^{-\frac{1}{p}}\right)}$.
\end{proof}

\subsection{Higher Order Moments}
\label{subsec:polylog-moments}
Due to the symmetry of $f_{Z}(z)$ around zero, the odd moments are zero. The even moments are given by
\begin{equation}
    E[z^{2k}]=\frac{\sigma^{2k}}{\Psi^{(\Gamma)}\left(p, d\ln(x)^{p}, d^{-\frac{1}{p}}\right)}\sum_{j=0}^{2k}\left({2k \atop j}\right)(-a)^{2k-j}\Psi^{(\Gamma)}\left(p, d\ln(x)^{p}, (j+1)d^{-\frac{1}{p}}\right), \quad k=1,2,3,\dots,
\end{equation}
where $\Psi^{(\Gamma)}\left(p, d\ln(x)^{p}, (j+1)d^{-\frac{1}{p}}\right)$ is the upper incomplete Fox-Wright function \eqref{FoxHupper} evaluated at $c=(j+1)d^{-\frac{1}{p}}$ and $x=a$. The even moments are finite when $p>1$ and $0<d<\infty$ or when $p=1$ and $d>2k+1$.
\begin{proof}
    The even moments for $f_{Z}(z)$ are given by
    \begin{equation*}
    \begin{split}
        E[z^{2k}]&=\int_{-\infty}^{\infty}z^{2k}f_{Z}(z)dz\\
        &=\frac{pd^{\frac{1}{p}}}{\sigma\Psi^{(\Gamma)}\left(p, d\ln(a)^{p}, d^{-\frac{1}{p}}\right)}\int_{0}^{\infty}z^{2k}e^{-d\ln\left(\frac{z}{\sigma}+a\right)^{p}}dz\\
        &=\frac{pd^{\frac{1}{p}}}{\Psi^{(\Gamma)}}\int_{\ln(a)}^{\infty}e^{u}(e^{u}-a)^{2k}e^{-du^{p}}du\\
        &=\frac{\sigma^{2k}pd^{\frac{1}{p}}}{\Psi^{(\Gamma)}}\int_{\ln(a)}^{\infty}e^{u}\sum_{j=0}^{2k}\left({2k \atop j}\right)e^{ju}(-a)^{2k-j}e^{-du^{p}}du\\
        &=\frac{\sigma^{2k}pd^{\frac{1}{p}}}{\Psi^{(\Gamma)}}\sum_{j=0}^{2k}\left({2k \atop j}\right)(-a)^{2k-j}\int_{\ln(a)}^{\infty}e^{u(j+1)}e^{-du^{p}}du\\
        &=\frac{\sigma^{2k}}{\Psi^{(\Gamma)}}\sum_{j=0}^{2k}\left({2k \atop j}\right)(-a)^{2k-j}\sum_{n=0}^{\infty}\Gamma\left(\frac{1}{p}+\frac{n}{p}, d\ln(a)^{p}\right)\frac{\left((j+1)d^{-\frac{1}{p}}\right)^{n}}{n!}\\
        &=\frac{\sigma^{2k}}{\Psi^{(\Gamma)}}\sum_{j=0}^{2k}\left({2k \atop j}\right)(-a)^{2k-j}\Psi^{(\Gamma)}(p, d\ln(a)^{p}, (j+1)d^{-\frac{1}{p}}).
        \end{split}
    \end{equation*}
The second line uses the fact that the integrand is even, and the third line uses the substitution $u=\ln\left(\frac{z}{\sigma}+a\right)$. The fourth line uses the binomial theorem, and the fifth through sixth line interchange summation and integration. The final line uses the definition of the upper-incomplete Fox-Wright function.
\end{proof}

\subsection{Cumulative Distribution Function}
\label{subsec:polylog-cdf}
The CDF is given by
\begin{equation}
    F_{Z}(z)=\frac{1}{2}+\frac{\text{sign}(z)}{2\Psi^{(\Gamma)}\left(p, d\ln(a)^{p}, d^{-\frac{1}{p}}\right)}\times\left(\Psi^{(\gamma)}\left(p, d\ln\left(\frac{|z|}{\sigma}+a\right)^{p}, d^{-\frac{1}{p}}\right)-\Psi^{(\gamma)}\left(p, d\ln(a)^{p}, d^{-\frac{1}{p}}\right)\right),
\end{equation}
where 
\begin{equation}\label{FoxHlower}
    \Psi^{(\gamma)}\left(p, d\ln(x)^{p}, c\right)=\prescript{}{1}{\Psi}^{(\gamma)}_{1}\left[{\left(\frac{1}{p}, \frac{1}{p}, d\ln(x)^{p}\right) \atop (1,0)}\Bigg|c\right]
    =\sum^{\infty}_{n=0}\gamma\Big(\frac{1}{p}+\frac{n}{p}, d\ln(x)^{p}\Big)\frac{c^{n}}{n!},
\end{equation}
is the lower-incomplete Fox-Wright function as a function as a function of $x$, evaluated at $c=d^{-\frac{1}{p}}$, where $\gamma\left(\frac{1}{p}+\frac{n}{p}, d\ln(x)^{p}\right)$ is the lower incomplete gamma function with parameter $\frac{n+1}{p}$ whose upper limit is $d\ln(x)^{p}$. This function converges with the same constraints as the upper-incomplete Fox-Wright function.
\begin{proof}
By definition, $F_{Z}(t)=\int_{-\infty}^{t}f_{z}(z)dz$. Depending on the sign of $t$ we have the following cases
\begin{equation*}
    F_{Z}(t)=
    \begin{cases}
        \int_{-\infty}^{t}f_{Z}(z)dz & t\leq0\\
        \int_{-\infty}^{0}f_{Z}(z)dz + \int_{0}^{t}f_{Z}(z)dz & t>0.\\
    \end{cases}
\end{equation*}
Or, equivalently
\begin{equation*}
    F_{Z}(t)=
    \begin{cases}
        \frac{1}{2}-\int_{0}^{|t|}f_{Z}(z)dz & t\leq0\\
        \frac{1}{2}+\int_{0}^{|t|}f_{Z}(z)dz & t>0.\\
    \end{cases}
\end{equation*}
So we may write $F_{Z}(t)=\frac{1}{2}+\text{sign}(t)\int_{0}^{|t|}f_{Z} (z)dz$, where
\begin{equation*}
    \begin{split}
    \int_{0}^{|t|}f_{Z}(z)dz&\\
    &=\frac{pd^{\frac{1}{p}}}{2\sigma\Psi^{(\Gamma)}\left(p, d\ln(a)^{p}, d^{-\frac{1}{p}}\right)}\int_{0}^{|t|}e^{-d\ln\left(\frac{|z|}{\sigma}+a\right)^{p}}dz\\
    &=\frac{pd^{\frac{1}{p}}}{2\Psi^{(\Gamma)}}\int_{\ln(a)}^{\ln\left(\frac{|t|}{\sigma}+a\right)}e^{u}e^{-du^{p}}du\\
    &=\frac{pd^{\frac{1}{p}}}{2\Psi^{(\Gamma)}}\times\left(\int_{0}^{\ln\left(\frac{|t|}{\sigma}+a\right)}e^{u}e^{-du^{p}}du-\int_{0}^{\ln(a)}e^{u}e^{-du^{p}}du\right)\\
    &=\frac{1}{2\Psi^{(\Gamma)}}\times\left(\Psi^{(\gamma)}\left(p, d\ln(\frac{|z|}{\sigma}+a)^{p}, d^{-\frac{1}{p}}\right)-\Psi^{(\gamma)}\left(p, d\ln(a)^{p}, d^{-\frac{1}{p}}\right)\right),
    \end{split}
\end{equation*}
where the second line follows from substitution and the fourth line follows from the definition of the lower-incomplete Fox-Wright function.
\end{proof}

\end{document}

%% file: myMacros.tex
\newcommand{\R}{\mathbb{R}}

\DeclareMathOperator{\E}{\mathbb{E}}
\DeclareMathOperator{\V}{\mathbb{V}}
\DeclareMathOperator{\Med}{\mathbb{M}}
\DeclareMathOperator{\Dom}{Dom}

\DeclareMathOperator\erf{erf}
\DeclareMathOperator\cdf{CDF}
\DeclarePairedDelimiterX{\infdivx}[2]{(}{)}{%
  #1\;\delimsize\|\;#2%
}
\DeclarePairedDelimiter{\pr}{(}{)}

\theoremstyle{definition}
\newtheorem{defi}{Definition}[section]
\newtheorem{thm}[defi]{Theorem}
\newtheorem{lem}[defi]{Lemma}
\newtheorem{exa}[defi]{Example}
\newtheorem{cor}[defi]{Corollary}

\newtheorem{asm}[defi]{Assumption}